\documentclass[10pt]{article}
\PassOptionsToPackage{table}{xcolor}
\PassOptionsToPackage{obeyspaces}{url}

\providecommand{\venue}{arxiv}
\edef\venue{\venue}

\newif\iftmlrlayout   
\newif\ifanonymous    

\def\venuetmlr{tmlr}
\def\venuetmlrfinal{tmlrfinal}
\ifx\venue\venuetmlr
  \tmlrlayouttrue  \anonymoustrue
  \usepackage{styles/tmlr}
\else\ifx\venue\venuetmlrfinal
  \tmlrlayouttrue  \anonymousfalse
  \usepackage[accepted]{styles/tmlr}
\else
  \tmlrlayoutfalse \anonymousfalse
  \usepackage[preprint]{styles/tmlr}
\fi\fi

\usepackage{hyperref}
\usepackage{url}
\usepackage{booktabs}
\usepackage{longtable}
\usepackage{graphicx}
\usepackage{amsmath,amssymb}
\usepackage{multirow}
\usepackage{xcolor}
\usepackage{enumitem}
\usepackage{array}
\usepackage{tikz}
\usetikzlibrary{arrows.meta,decorations.pathmorphing,decorations.pathreplacing}
\usepackage{pgfplots}
\pgfplotsset{compat=1.17}
\usepackage{colortbl}
\usepackage{lscape}

\definecolor{axblue}{HTML}{4C7E9F}    
\definecolor{axaqua}{HTML}{3D6455}    
\definecolor{axyellow}{HTML}{C29A4E}  
\definecolor{tlfound}{HTML}{D6DDE2}
\definecolor{tlagentic}{HTML}{BC5526}
\definecolor{tlsurvey}{HTML}{2C343D}

\definecolor{rgsingle}{HTML}{3E6F8E}
\definecolor{rgstep}{HTML}{1B8A6B}
\definecolor{rgmulti}{HTML}{B4683D}
\definecolor{rgsystem}{HTML}{6A4C93}
\DeclareRobustCommand{\regime}[2]{{\setlength{\fboxsep}{2pt}%
  \colorbox{#1}{\color{white}\fontsize{6.6}{7.6}\selectfont\sffamily\bfseries #2}}}
\DeclareRobustCommand{\rgST}{\regime{rgsingle}{single-turn}}
\DeclareRobustCommand{\rgSA}{\regime{rgstep}{one agent, many steps}}
\DeclareRobustCommand{\rgMA}{\regime{rgmulti}{multi-agent}}
\DeclareRobustCommand{\rgSY}{\regime{rgsystem}{whole trajectory}}
\DeclareRobustCommand{\rgALL}{\regime{black!45}{all four}}

\newcommand{\cmark}{$\checkmark$}
\newcommand{\xmark}{\textcolor{gray!60}{--}}
\newcommand{\pidx}[2]{\textcolor{#1}{\scriptsize\sffamily\bfseries #2}\,}
\newcommand{\nbar}[1]{\textcolor{axyellow!85}{\rule[-0.1ex]{#1mm}{5pt}}}
\newtheorem{definition}{Definition}
\newtheorem{proposition}{Proposition}

\newcolumntype{L}[1]{>{\raggedright\arraybackslash}p{#1}}
\newcommand{\thc}[1]{\textbf{#1}}
\newcommand{\rothc}[1]{\multicolumn{1}{l}{\makebox[0pt][l]{\rotatebox[origin=l]{90}{\textbf{#1}}}}}
\newcommand{\tband}[2]{\addlinespace[4pt]\multicolumn{#1}{@{}l}{\itshape #2}\\\addlinespace[2pt]}
\newcommand{\takeaway}[2][]{\par\medskip\noindent\fbox{\parbox{\dimexpr\linewidth-2\fboxsep-2\fboxrule\relax}{\small\textbf{Takeaway\if\relax\detokenize{#1}\relax\else\ (Section~\ref{#1})\fi.}\ #2}}\par\medskip}

\definecolor{titleblue}{HTML}{0645AD}

\title{Uncertainty Quantification for LLM Agents: A Taxonomy, an Evaluation Protocol, and an Empirical Study}

\author{\name Moule Lin \email moulel@tcd.ie \\
      \addr Trinity College Dublin, Ireland
      \AND
      \name Qizhen Lan \email Qizhen.Lan@uth.tmc.edu \\
      \addr University of Texas Health Science Center at Houston, USA
      \AND
      \name Shuhao Guan \email shuhao.guan@ucdconnect.ie \\
      \addr University College Dublin, Ireland
      \AND
      \name Weipeng Jing \email jwp@nefu.edu.cn \\
      \addr Northeast Forestry University, China
      \AND
      \name Jiexin Fan \email jifan@tcd.ie \\
      \addr Trinity College Dublin, Ireland
      \AND
      \name David Gregg \email david.gregg@cs.tcd.ie \\
      \addr Trinity College Dublin, Ireland
      \AND
      \name Goetz Botterweck \email goetz.botterweck@tcd.ie \\
      \addr Trinity College Dublin, Ireland}

\def\month{08}
\def\year{2026}
\def\openreview{\url{https://openreview.net/forum?id=XXXXXXXXXX}}

\newcommand{\link}[1]{\href{#1}{\textcolor{titleblue}{{\urlstyle{same}\nolinkurl{#1}}}}}
\newcommand{\shortlink}[2]{\href{#1}{\textcolor{titleblue}{{\urlstyle{same}\nolinkurl{#2}}}}}
\newcommand{\plainurl}[1]{{\urlstyle{same}\nolinkurl{#1}}}

\newcommand{\abstracttext}{%

Large language models (LLMs) are no longer deployed only for single-turn conversation but increasingly act as agents that plan, call tools, retrieve evidence, maintain memory, and interact over long horizons, often together with other agents through multi-turn conversations.
Therefore, knowing when to trust the agentic system is a prerequisite for safe deployment.
However, existing work on quantifying uncertainty for LLMs was built almost entirely for single-turn question answering.
This paper argues that errors and uncertainty arise from multi-turn conversations, environments, and tools rather than from a single-turn question answering setting.
It comes late, however, and is compounded in a single score that is too coarse to represent the unreliability.
We organize the literature with a three-axis taxonomy, (1) \textbf{what} the uncertainty is, (2) \textbf{how} it is estimated, and (3) \textbf{where} uncertainty arises during an agent pipeline.
We investigate step-level and trajectory-level calibration and show with a simple counterexample that the first does not imply the second.
Experiments on real agent traces across four models and up to a 50-step budget show that the proposed metric and reporting protocol (Trajectory-Checkpoint Expected Calibration Error, TC-ECE) can be computed and that step errors are coupled along a trajectory.
We find that confidence estimates from the agent's own responses do not consistently outperform a simple baseline.
The experiments also show that averaging all trajectories together can hide overconfidence at later stages, which becomes visible when results are analyzed across different horizons.
In simpler terms, this paper identifies where the uncertainty comes from in the agentic system pipeline, how to teach agents to know
when they are wrong, and why one confidence number is not enough.
}
\ifanonymous
  \hypersetup{pdfauthor={}}
\else
  \hypersetup{pdfauthor={Moule Lin, Qizhen Lan, Shuhao Guan, Weipeng Jing,
    Jiexin Fan, David Gregg, Goetz Botterweck}}
\fi
\hypersetup{
  pdftitle={Uncertainty Quantification for LLM Agents: A Taxonomy, an Evaluation Protocol, and an Empirical Study},
  colorlinks=true,
  linkcolor=titleblue,
  citecolor=titleblue,
  urlcolor=titleblue,
}

\iftmlrlayout
  
  \hypersetup{linkcolor=black, citecolor=black, urlcolor=black}
\else
  
\fi

\begin{document}

\iftmlrlayout
\maketitle

\begin{abstract}
\abstracttext
\end{abstract}

\else
\IfFileExists{figures/logos/Trinity-Main-Logo.jpg}{%
  \vspace*{-36pt}
  \noindent
  \raisebox{-0.5\height}{\includegraphics[height=1.15cm,trim=186 105 186 130,clip]{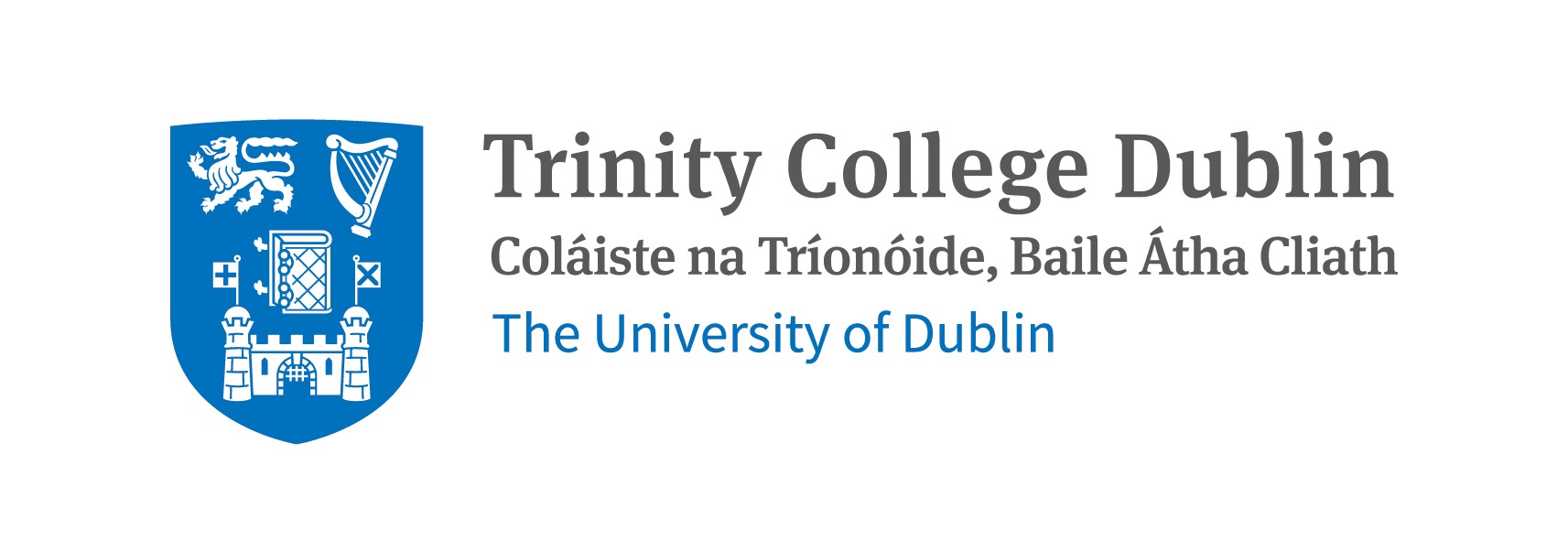}}%
  \hfill
  \IfFileExists{figures/logos/eu-funded.png}{%
    \raisebox{-0.5\height}{\includegraphics[height=1.1cm]{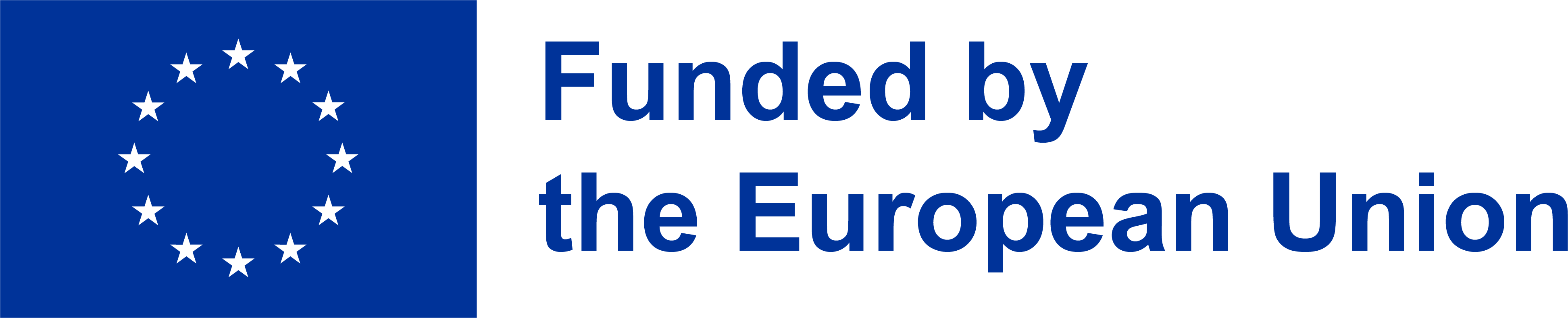}}}{}%
  \par
  \vspace{7pt}
}{\vspace*{-18pt}}
{\color{titleblue}\hrule height 1.2pt}
\vspace{20pt}

\begin{center}
  {\LARGE\bfseries\color{titleblue}
    Uncertainty Quantification for LLM Agents: A Taxonomy, an Evaluation Protocol, and an Empirical Study\par}
  \vspace{16pt}
  {\normalsize\bfseries
    Moule Lin\textsuperscript{1}, Qizhen Lan\textsuperscript{2},
    Shuhao Guan\textsuperscript{3}, Weipeng Jing\textsuperscript{4},
    Jiexin Fan\textsuperscript{1}, David Gregg\textsuperscript{1},
    Goetz Botterweck\textsuperscript{1,*}\par}
  \vspace{7pt}
  {\small
    \textsuperscript{1}\,Trinity College Dublin, Ireland\qquad
    \textsuperscript{2}\,University of Texas Health Science Center at Houston, USA\par
    \vspace{2pt}
    \textsuperscript{3}\,University College Dublin, Ireland\qquad
    \textsuperscript{4}\,Northeast Forestry University, China\par}
  \vspace{5pt}
  {\footnotesize
    moulel@tcd.ie,\enspace Qizhen.Lan@uth.tmc.edu,\enspace
    shuhao.guan@ucdconnect.ie,\enspace jwp@nefu.edu.cn,\par
    \vspace{1pt}
    jifan@tcd.ie,\enspace david.gregg@cs.tcd.ie,\enspace
    goetz.botterweck@tcd.ie\par}
\end{center}
\vspace{8pt}

\noindent{\bfseries\color{titleblue}Abstract:}
\abstracttext

\medskip
\noindent{\bfseries Keywords:} Uncertainty Quantification, LLM Agents, Calibration,
Hallucination Detection, Tool Use, Retrieval-Augmented Generation, Multi-Agent Systems

\smallskip
\noindent{\bfseries Date:} August 2026\\
{\bfseries Corresponding author (*):} Goetz Botterweck,
\href{mailto:goetz.botterweck@tcd.ie}{goetz.botterweck@tcd.ie}

\smallskip
\noindent{\footnotesize This work was funded by the European Union's Horizon Europe
2021--2027 Framework Programme under the Marie Sk\l{}odowska-Curie grant agreement
No.~101072456. Views and opinions expressed are those of the author(s) only and do not
necessarily reflect those of the European Union or the granting authority.}

\vspace{6pt}
{\color{titleblue}\hrule height 1.2pt}
\vspace{18pt}

\fi


\newpage
\begingroup
\hypersetup{linktoc=page, linkcolor=titleblue}
\tableofcontents
\endgroup
\newpage

\section{Introduction}
\label{sec:intro}

\begin{quote}
\itshape``Doubt is not a pleasant condition, but certainty is an absurd one.''
\par\vspace{3pt}
{\normalfont\upshape\hfill Voltaire, 1770}
\end{quote}

The deployment and evaluation of Large Language Models (LLMs) have changed dramatically.
An LLM was once deployed as a single-turn question-answering system where a question went in and an answer came out \citep{brown2020language,ouyang2022training}.
Today, however, the same model increasingly performs as an \emph{agent} with reasoning capabilities, a system that decomposes goals into plans, invokes external tools and APIs, retrieves documents, maintains and updates memory, and executes multi-step trajectories in an environment, frequently communicating with other agents \citep{yao2023react,wang2024survey,xi2023rise,yu2025survey}.
This shift undoubtedly improves the capabilities of LLMs but also varies the ways in which they can fail. A wrong action or erroneous information arising at an early stage will accumulate and flow in a trajectory, making it almost impossible to retrospectively locate the error points when the final output is only a sequence of tokens.
The main issue is not that agents make mistakes; all statistical methods make mistakes, but these mistakes occur \emph{early, confidently, and irreversibly}, in ways that a system-level checker cannot catch in time.
This naturally raises the question: \textbf{when should we trust the system, and when should it defer, re-plan, ask, or stop?}
\par
Uncertainty is an effective signal for addressing this question. It captures not only uncertainty in the final answer, but also when, where, and to what extent the model becomes uncertain throughout the trajectory.
There is much research on uncertainty mechanisms in LLMs. However, most existing work still focuses on single-turn QA, even as agents are increasingly deployed for complex tasks \citep{oh2026uncertainty}.
A large body of work in classical machine learning has studied calibration and uncertainty estimation. This includes post-hoc recalibration methods \citep{platt1999probabilistic,zadrozny2002transforming}, calibration metrics \citep{naeini2015obtaining}, Bayesian and ensemble-based uncertainty estimation \citep{gal2016dropout,blundell2015weight,lakshminarayanan2017simple}, methods for separating aleatoric and epistemic uncertainty \citep{kendall2017what,hullermeier2021aleatoric}, and approaches to selective prediction \citep{geifman2017selective,kamath2020selective}.
Prior work studies uncertainty in single-turn LLMs through confidence estimation, calibration, and selective prediction \citep{brier1950verification,chow1970optimum,guo2017calibration,hullermeier2021aleatoric}.
Models can verbalize their output uncertainty by being asked to explicitly report it \citep{lin2022teaching,tian2023just,xiong2023can}.
Uncertainty can be measured over the space of \emph{meanings} rather than surface sequence \citep{kuhn2023semantic,farquhar2024detecting}.
Model consistency can be used to detect hallucinations without external resources \citep{manakul2023selfcheckgpt,lin2023generating}.
Conformal prediction brings distribution-free guarantees \citep{vovk2005algorithmic,quach2024conformal}.
These methods are mature, well benchmarked \citep{vashurin2025benchmarking,fadeeva2023lmpolygraph}, and increasingly well understood, including their failure modes \citep{santilli2025revisiting}.
However, as agents are given more complex tasks, the object of study is moving from a single answer to an agentic system \citep{kirchhof2025position,xia2026propagation}.
\par
Uncertainty in an agentic system is more complex because uncertainty arising at an early stage can accumulate and propagate across subsequent steps \citep{zhang2026agentic,zhao2024saup,donaldson2026bayesian}. Since each step conditions on the agent's own previous outputs, errors across steps are correlated. As Section~\ref{sec:formal} formalizes, per-step calibration therefore provides limited information about the probability that the entire trajectory succeeds.
Uncertainty stems not only from the model's parameters but from noisy tools and dynamic environments \citep{han2024towards,xuan2026confidence}, an external component that no amount of introspection on the model can estimate.
And a final-step confidence score is diagnostically too late.
Writing $h_t$ for the interaction history after step $t$, what the next step needs is a reliability estimate $\hat{R}(h_t)$ available \emph{during} the trajectory. The key question is not how likely the trajectory is to fail but whether a wrong behavior \emph{now} would change the outcome, which we will formalize as the intervention advantage $A_i(h_t)$ of Eq.~\eqref{eq:advantage} \citep{zhang2026calibration}.
\begin{figure}[t]
  \centering
  \includegraphics[width=\textwidth]{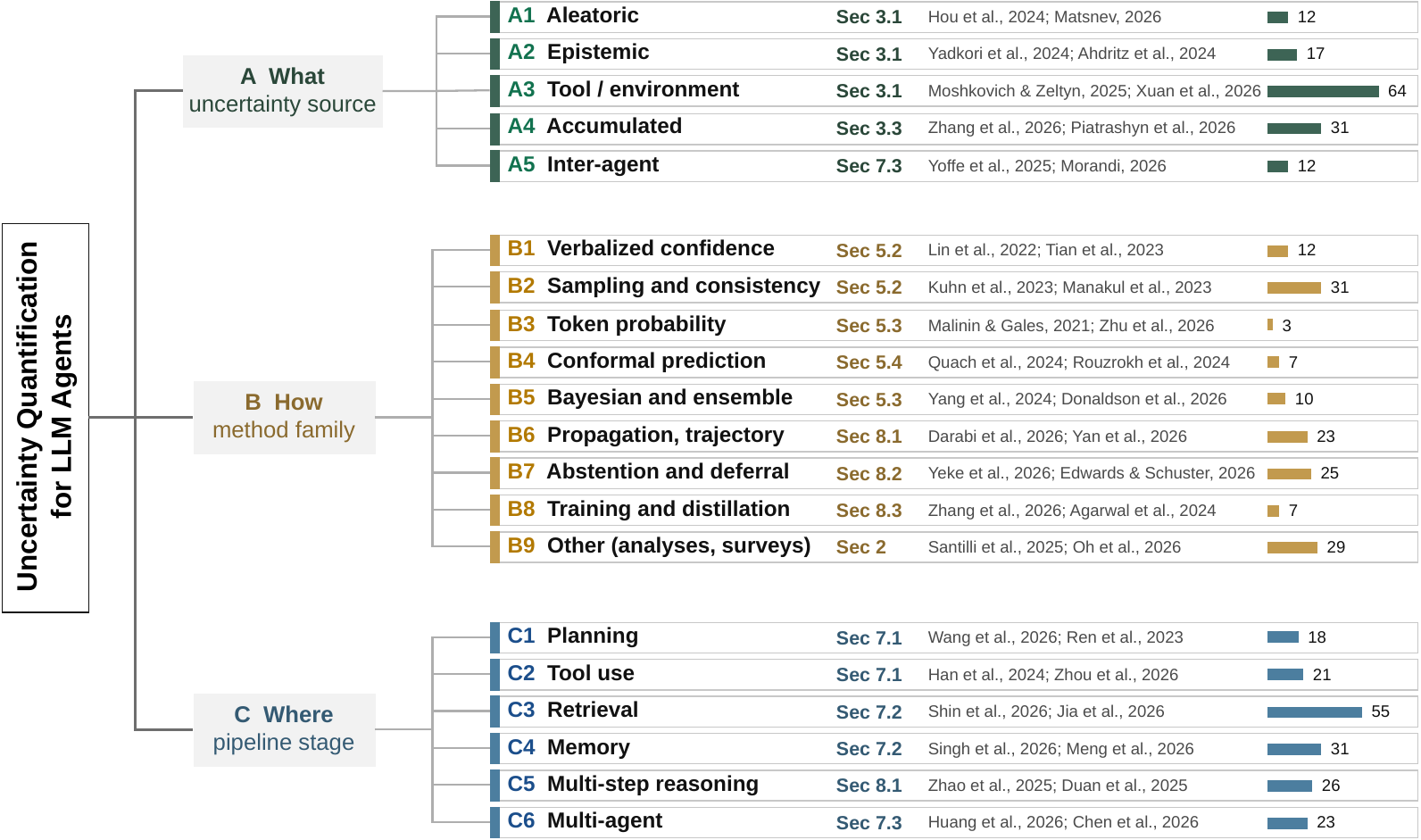}
  \caption{Three-axis taxonomy of agent uncertainty: source (A), method family
  (B), and pipeline stage (C). Each leaf shows its corpus count and two
  representative works. Papers may carry multiple tags within each panel.}
  \label{fig:tree}
\end{figure}
This paper reviews and organizes the fast-growing but scattered literature on uncertainty estimation for LLM agents. Our central object is a simple transition from \emph{single-turn confidence} to
\emph{trajectory-level reliability}, and this transition needs new methods,
new benchmarks, and new conceptual distinctions. More importantly, the analysis must change from the final output to the trajectory, and the evaluation standard must change from a correlational perspective to a decision-relevant one. Each of these changes is already visible in the literature we survey, but no survey has synthesized this research as a scientific report.
\begin{itemize}[leftmargin=1.4em,itemsep=2pt]
  \item \textbf{A documented corpus and a three-axis taxonomy}
  (Figure~\ref{fig:tree}). The taxonomy sorts methods by \emph{source} (what
  the uncertainty is), \emph{method family} (how it is estimated), and
  \emph{pipeline stage} (where it applies). The axes are largely non-nested.
  Table~\ref{tab:bigtable} in Appendix~\ref{app:bigtable} classifies every
  paper, and we release the corpus
  with its protocol so the counts can be recomputed
  (Appendix~\ref{app:corpus}).
  \item \textbf{A formal lens} (Section~\ref{sec:formal}). We define
  step-level and trajectory-level calibration and record the overlooked fact
  that, under the locally defined step labels the field uses, the first does
  not compose into the second.
  Proposition~\ref{prop:positive} bounds the error of the marginal-product
  approximation by measurable error correlations, and intervention advantage
  states the decision-theoretic target that separates estimation from control.
  The lens locates where trajectory calibration breaks rather than repairing
  it.
  \item \textbf{A stage-by-stage methods review}
  (Sections~\ref{sec:foundations} through \ref{sec:families}) covering
  single-turn foundations, tool use and planning, retrieval and memory,
  propagation and trajectory-level estimation, multi-agent coordination,
  abstention and control, and uncertainty-aware training, tracing each stage
  back to its single-turn ancestry.
  \item \textbf{A trajectory-level metric and a measurement on real traces}
  (Sections~\ref{sec:eval-metrics} and \ref{sec:empirical}). TC-ECE
  (Definition~\ref{def:tcece}) reads calibration at checkpoints along a
  trajectory rather than at the final answer alone. We compute it on four
  models over three experiments at horizons of up to $50$ steps. In these
  configurations the agent's stated confidence does not consistently beat a
  step-index baseline, and pooling checkpoints hides overconfidence that
  appears only late. Traces, code, and analysis are released.
  \item \textbf{A concrete research agenda} (Section~\ref{sec:challenges})
  distilled from the gaps the taxonomy makes visible, stated as ten problems,
  each with the reason it has resisted current tools and the form progress
  would take (Table~\ref{tab:agenda}).
\end{itemize}

Considering the different needs of readers interested in uncertainty estimation for LLM agents, we provide several entry points into this paper.
Section~\ref{sec:roadmap} maps the whole structure and marks the section that covers each block, so a reader can locate any of the entry points below before committing to a path through it.
A practitioner seeking a rapid deployment of an agent can refer directly to the decision guide in Figure~\ref{fig:guide} and the method families it points to.
A researcher looking for a research problem can start from Section~\ref{sec:formal}, and land on the agenda of Section~\ref{sec:challenges} (Table~\ref{tab:agenda}); the heatmap of Figure~\ref{fig:heatmap} compresses the same gap analysis into one view.
Newcomers can start with the background and single-turn foundations. These sections introduce the main concepts used later on, before moving to uncertainty in agentic systems.
\par
We built a documented, multi-source corpus by searching arXiv, the ACM Digital Library, OpenReview, IEEE Xplore, and PMLR with keywords derived from the taxonomy axes above. Every retained paper was read and classified independently by two annotators, with an LLM judge proposing a resolution for each disagreement and a third annotator making the final call; every taxonomy assignment was independently checked by a second annotator. The resulting core contains 120 papers spanning single-turn foundations, agentic methods, and related surveys.
We also include 454 additional works from the classical UQ, LLM confidence, agent systems, decision-theoretic, and retrieval literatures that the core builds on; Appendix~\ref{app:corpus} documents the protocol.
Figure~\ref{fig:timeline} shows the temporal distribution of the corpus: agent-focused work is nearly absent before 2024, emerges as a distinct wave through 2025, and grows rapidly into 2026, while single-turn foundations accumulate steadily throughout.
The pattern is worth registering. The field we survey is young enough that its conventions, its metrics, and its benchmark culture are all still fluid, which is why a synthesis now can shape them.

A reader who knows the survey literature will ask what this work adds to it.
Recent survey works define the problem and list its challenges. None of them sorts the estimation methods by pipeline stage, and none of them starts from a set of papers collected in a documented way \citep{oh2026uncertainty,kirchhof2025position,xia2026propagation}. \citet{oh2026uncertainty} do not treat abstention, deferral, or uncertainty-aware training as categories in their taxonomy. This paper covers all three, and Section~\ref{sec:related} sets out the full comparison row by row in Table~\ref{tab:surveys}.



\section{Related Work and Positioning}
\label{sec:related}

Existing surveys on uncertainty estimation mainly focus on non-agent or single-turn settings, or study agents without making uncertainty the central object. A pipeline-wide synthesis is therefore needed to organize this fragmented literature. Building upon previous surveys, we organize the area along three axes: (1) \textbf{what} the uncertainty is, (2) \textbf{how} it is estimated and used, and (3) \textbf{where} in the agent pipeline the methods operate.

There are some related surveys on uncertainty quantification in LLMs and agents, but none of them covers the full pipeline that uncertainty comes from.
Table~\ref{tab:surveys} summarizes the landscape and compares related surveys with ours. The comparison shows why the research gap remains open.
The first block is uncertainty without agents, mainly focusing on methods developed for traditional deep learning \citep{abdar2021review,gawlikowski2023survey,he2023survey,wang2023calibration} such as Bayesian approximations, ensembles, and single-forward-pass estimators. In these studies, information theory is often used to decompose the uncertainty into an epistemic part, reducible by adding more training data, and an aleatoric part, generally irreducible regardless of the amount of data. However, these methods often assume that the data is i.i.d., which does not generally hold for agents that must interact with an environment over multiple steps.
\par
With the rapid development of LLMs in recent years, a growing body of research moves UQ to single-turn LLMs \citep{geng2024survey,huang2024survey,shorinwa2024survey,xia2025survey,liu2025uncertainty}. For example, verbalized uncertainty asks the LLM to explicitly report an uncertainty value for its output; sampling and consistency methods generate multiple answers and measure the agreement among them. However, these remain single-generation settings: trajectories, tool calls, and interaction with environments are not yet considered.
Conformal prediction \citep{campos2024conformal,wen2024know} is one branch of uncertainty quantification. It provides a prediction set with a statistical coverage guarantee for each output.
\par
The second block focuses on hallucination.
These surveys \citep{ji2023survey,zhang2023sirens,huang2025survey} mainly study \emph{what} goes wrong and summarize benchmarks for measuring these failures. However, they do not treat uncertainty as the main tool for predicting \emph{when} such failures may occur.
\begin{figure}[t]
  \centering
\begin{tikzpicture}
\begin{axis}[
  width=0.96\textwidth, height=6.2cm,
  ybar stacked,
  bar width=16pt,
  ymin=0, ymax=38,
  ytick={0,10,20,30},
  ylabel={Papers per quarter},
  ylabel style={font=\footnotesize},
  symbolic x coords={pre2024,2024Q1,2024Q2,2024Q3,2024Q4,2025Q1,2025Q2,2025Q3,2025Q4,2026Q1,2026Q2,2026Q3},
  xtick=data,
  xticklabels={$\leq$\,2023,24Q1,24Q2,24Q3,24Q4,25Q1,25Q2,25Q3,25Q4,26Q1,26Q2,26Q3},
  x tick label style={font=\scriptsize},
  y tick label style={font=\scriptsize},
  ymajorgrids,
  grid style={gray!18},
  axis line style={gray!45},
  tick style={gray!45},
  axis x line*=bottom,
  axis y line*=left,
  legend style={at={(0.03,0.97)}, anchor=north west, draw=none,
                fill=none, font=\footnotesize, cells={anchor=west}},
  reverse legend,
]
\addplot+[ybar, fill=tlfound, draw=white, line width=0.4pt] coordinates {
  (pre2024,5) (2024Q1,1) (2024Q2,3) (2024Q3,0) (2024Q4,0) (2025Q1,0)
  (2025Q2,2) (2025Q3,2) (2025Q4,0) (2026Q1,2) (2026Q2,11) (2026Q3,0)};
\addlegendentry{Single-turn foundations}
\addplot+[ybar, fill=tlagentic, draw=white, line width=0.4pt] coordinates {
  (pre2024,2) (2024Q1,1) (2024Q2,1) (2024Q3,0) (2024Q4,2) (2025Q1,1)
  (2025Q2,7) (2025Q3,4) (2025Q4,8) (2026Q1,21) (2026Q2,23) (2026Q3,2)};
\addlegendentry{Agentic methods}
\addplot+[ybar, fill=tlsurvey, draw=white, line width=0.4pt] coordinates {
  (pre2024,2) (2024Q1,1) (2024Q2,1) (2024Q3,2) (2024Q4,2) (2025Q1,5)
  (2025Q2,2) (2025Q3,3) (2025Q4,1) (2026Q1,2) (2026Q2,1) (2026Q3,0)};
\addlegendentry{Surveys}
\node[font=\scriptsize, anchor=south] at (axis cs:2026Q1,25.4) {25};
\node[font=\scriptsize, anchor=south] at (axis cs:2026Q2,35.4) {35};
\end{axis}
\end{tikzpicture}
  \caption{Quarterly distribution of the 120-paper corpus, pre-2024 pooled. Agent-focused work
  emerges as a distinct wave after 2024 and grows rapidly into 2026. The final
  quarter is incomplete.}
  \label{fig:timeline}
\end{figure}
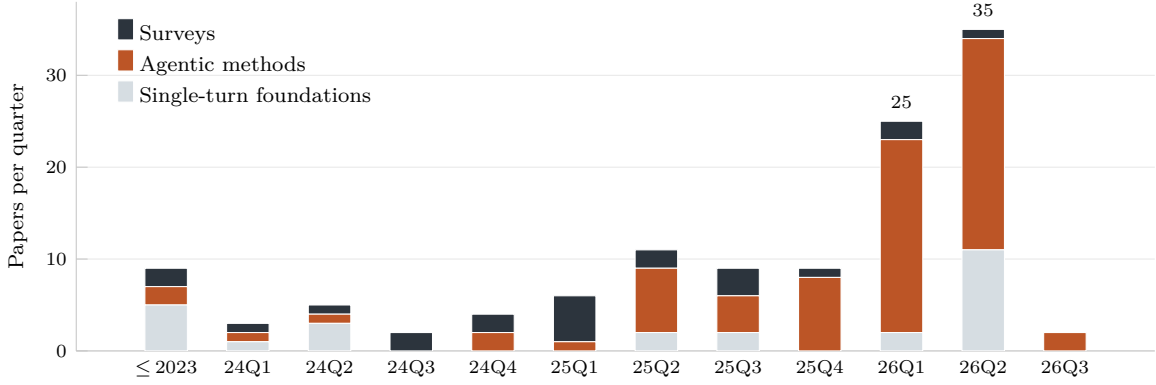
\begin{table}[!t]
\centering
\caption{Comparison of related surveys by scope, taxonomy, evaluation,
uncertainty propagation, and corpus disclosure. Checks superscripted
\textsuperscript{A}, \textsuperscript{B}, and \textsuperscript{C} mark the
three taxonomy panels used in this paper, in the panel colors of
Figure~\ref{fig:tree}; the letter, not the color, carries the distinction.
The last column records whether a survey
states how its papers were collected and releases the resulting paper set; it
is not a claim that this paper meets the standard of a PRISMA-style
systematic review, which Appendix~\ref{app:corpus} states it does not. The
estimator-internals column marks surveys that treat single-turn estimator
internals in greater depth than this paper attempts; we cede that column by
design and refer readers to those surveys for it.}
\label{tab:surveys}
\small
\setlength{\tabcolsep}{4.6pt}
\renewcommand{\arraystretch}{1.15}
\newcommand{\panelmark}[3]{\textcolor{#1}{$\checkmark$\kern0.5pt\raisebox{0.55ex}{\scriptsize\sffamily\bfseries #2}}\kern-#3pt}
\newcommand{\cmarkA}{\panelmark{axaqua!65!black}{A}{1.2}}
\newcommand{\cmarkB}{\panelmark{axyellow!75!black}{B}{1.2}}
\newcommand{\cmarkC}{\panelmark{axblue!65!black}{C}{1.2}}
\newcommand{\coveredin}[1]{\strut{\scriptsize\color{gray!45!black}#1}}
\begin{tabular}{@{}L{3.95cm}L{1.15cm}L{4.82cm}ccccccccc@{}}
\toprule
\thc{Topic} & \thc{\scriptsize Sections} & \thc{Representative surveys} & \rothc{UQ} & \rothc{Agentic} & \rothc{Sources} & \rothc{Estimators} & \rothc{Estimator internals} & \rothc{Stages} & \rothc{Benchmarks} & \rothc{Propagation} & \rothc{Corpus disclosed}\\
\addlinespace[4pt]
\midrule
\tband{12}{Uncertainty, without agents}
UQ and calibration in deep learning & \coveredin{\ref{sec:sources}--\ref{sec:metrics}} & \citet{abdar2021review,gawlikowski2023survey,he2023survey,wang2023calibration} & \cmark & \xmark & \cmarkA & \cmarkB & \cmark & \xmark & \xmark & \xmark & \xmark\\
\addlinespace[2pt]
UQ for single-turn LLMs & \coveredin{\ref{sec:foundations}} & \citet{geng2024survey,huang2024survey,shorinwa2024survey,xia2025survey,liu2025uncertainty} & \cmark & \xmark & \cmarkA & \cmarkB & \cmark & \xmark & \cmark & \xmark & \xmark\\
\addlinespace[2pt]
Conformal prediction; abstention & \coveredin{\ref{sec:conformal}, \ref{sec:control}} & \citet{campos2024conformal,wen2024know} & \cmark & \xmark & \xmark & \cmarkB & \cmark & \xmark & \xmark & \xmark & \xmark\\
\tband{12}{Failure modes and grounding, without uncertainty}
Hallucination in NLG and LLMs & \coveredin{\ref{sec:sampling}} & \citet{ji2023survey,zhang2023sirens,huang2025survey} & \xmark & \xmark & \xmark & \xmark & \xmark & \xmark & \cmark & \xmark & \xmark\\
\addlinespace[2pt]
Retrieval-augmented generation & \coveredin{\ref{sec:rag}} & \citet{gao2023retrieval} & \xmark & \xmark & \xmark & \xmark & \xmark & \cmarkC & \cmark & \xmark & \xmark\\
\tband{12}{LLM agents, without uncertainty}
Architectures and multi-agent systems & \coveredin{\ref{sec:multiagent}} & \citet{xi2023rise,wang2024survey,guo2024large} & \xmark & \cmark & \xmark & \xmark & \xmark & \cmarkC & \xmark & \xmark & \xmark\\
\addlinespace[2pt]
Trust, security, and evaluation & \coveredin{\ref{sec:eval}} & \citet{yu2025survey,su2025survey,mohammadi2025evaluation} & \xmark & \cmark & \xmark & \xmark & \xmark & \cmarkC & \cmark & \xmark & \xmark\\
\addlinespace[2pt]
Communication, provenance, small models & \coveredin{\ref{sec:multiagent}} & \citet{chen2026five,wang2026agent,sharma2025small} & \xmark & \cmark & \xmark & \xmark & \xmark & \cmarkC & \xmark & \xmark & \xmark\\
\tband{12}{Uncertainty for agents}
Positions and perspectives & \coveredin{\ref{sec:taxonomy}} & \citet{kirchhof2025position,zhang2026from,papamarkou2026position} & \cmark & \cmark & \cmarkA & \xmark & \xmark & \xmark & \xmark & \xmark & \xmark\\
\addlinespace[2pt]
Agent UQ foundations; benchmark classes & \coveredin{\ref{sec:taxonomy}, \ref{sec:eval}} & \citet{oh2026uncertainty} & \cmark & \cmark & \cmarkA & \cmarkB & \xmark & \xmark & \cmark & \xmark & \xmark\\
\addlinespace[2pt]
Uncertainty propagation mechanisms & \coveredin{\ref{sec:propagation}} & \citet{xia2026propagation} & \cmark & \cmark & \xmark & \xmark & \xmark & \xmark & \xmark & \cmark & \xmark\\
\midrule
\textbf{This paper} & \coveredin{all} & Corpus-grounded taxonomy, methods, and evaluation & \cmark & \cmark & \cmarkA & \cmarkB & \xmark & \cmarkC & \cmark & \cmark & \cmark\\
\bottomrule
\end{tabular}
\end{table}

\par
Surveys on the agent side focus more on trustworthiness and security \citep{yu2025survey,su2025survey,he2024emerged}, evaluation \citep{mohammadi2025evaluation}, provenance and traceability \citep{wang2026agent}, agent architectures \citep{wang2024survey,xi2023rise}, multi-agent systems \citep{guo2024large}, memory \citep{zhang2024surveymemory}, tool learning and augmented language models \citep{qin2024toollearning,mialon2023augmented}, and communication \citep{chen2026five,sarkar2025survey}.
In these works, uncertainty appears only as a paragraph in a trustworthiness taxonomy or a row in an evaluation checklist, rather than as the central organizing subject.
These literatures are broad and sophisticated, but remain largely disconnected on the questions this paper takes as primary. In particular, they do not provide a unified uncertainty-estimation framework organized around the full agent pipeline.
The closest work is \citet{oh2026uncertainty}. They give the first general formulation of agent UQ, lay out four agent-specific technical challenges, and analyze them numerically on $\tau^2$-bench. They also separate the uncertainty of the agent's own actions from the uncertainty of what tools and users return; our \textcolor{axaqua!65!black}{Panel A} follows that split, and our \textcolor{axyellow!75!black}{Panel B} widens the three estimator families they compare into eight indexed families. What they do not do is sort methods by pipeline stage, which is our \textcolor{axblue!65!black}{Panel C}. Table~\ref{tab:ohcompare} makes the relation explicit dimension by dimension.
\par
Three of the contributions listed above have no counterpart there. Section~\ref{sec:formal} records the elementary fact that step-level calibration does not compose into trajectory-level calibration and quantifies the resulting gap, so the quantity the field reports is not the quantity deployment needs. Section~\ref{sec:empirical} measures the step-error correlation that governs the gap. Abstention, deferral, and uncertainty-aware training, not treated as categories there, each get a section here.

\begin{table}[t]
\centering
\caption{Item-by-item relation to \citet{oh2026uncertainty}, the closest
prior synthesis. Each cell states what the work treats as an organizing
element; a dash means the dimension is not an organizing element of that
work, not that it is never mentioned.}
\label{tab:ohcompare}
\small
\setlength{\tabcolsep}{5pt}
\renewcommand{\arraystretch}{1.15}
\begin{tabular}{@{}L{3.2cm}L{4.6cm}L{5.8cm}@{}}
\toprule
\thc{Dimension} & \thc{\citet{oh2026uncertainty}} & \thc{This paper}\\
\midrule
Formal object & first general agent-UQ formulation & step- and trajectory-level calibration (Definitions~\ref{def:step}--\ref{def:traj}) with Propositions~\ref{prop:compose}--\ref{prop:positive}\\
\addlinespace[2pt]
Uncertainty sources & agent's own actions vs.\ tool and user returns & five-way \textcolor{axaqua!65!black}{Panel A} (\textcolor{axaqua!65!black}{A1}--\textcolor{axaqua!65!black}{A5})\\
\addlinespace[2pt]
Estimator coverage & three families compared & eight indexed families (\textcolor{axyellow!75!black}{B1}--\textcolor{axyellow!75!black}{B8}) plus a residual for work introducing no method (\textcolor{axyellow!75!black}{B9})\\
\addlinespace[2pt]
Pipeline-stage axis & --- & \textcolor{axblue!65!black}{Panel C} with per-stage sections (Section~\ref{sec:pipeline})\\
\addlinespace[2pt]
Abstention, deferral, training & not categories in its taxonomy & dedicated sections (Sections~\ref{sec:control} and~\ref{sec:training})\\
\addlinespace[2pt]
Trajectory calibration & --- & Definition~\ref{def:traj} and the TC-ECE protocol (Section~\ref{sec:eval-metrics})\\
\addlinespace[2pt]
Control target & --- & intervention advantage $A_i(h_t)$ (Eq.~\ref{eq:advantage})\\
\addlinespace[2pt]
Empirical component & numerical analysis on $\tau^2$-bench & protocol experiments on four models, three tasks, horizons to a $50$-step budget\\
\addlinespace[2pt]
Literature methodology & challenges and research directions & documented $120$-paper corpus with dual annotation, adjudication, and recall audits (Appendix~\ref{app:corpus})\\
\bottomrule
\end{tabular}
\end{table}

\par
\citet{xia2026propagation} is the closest on the propagation axis: it focuses on how uncertainty propagates through the system once it has been produced, while we focus on how uncertainty is estimated and which stages of the agent pipeline these estimation methods apply to.
\par
On the estimator side, the closest single method is Holistic Trajectory Calibration \citep{zhang2026agentic}, which produces a calibrated trajectory confidence from process-level features of the whole run and evaluates it across eight benchmarks. It is an estimator rather than a synthesis, so it does not appear in Table~\ref{tab:surveys}; Section~\ref{sec:eval-metrics} explains why it and the TC-ECE protocol proposed here are complementary rather than competing.

\section{Background and Preliminaries}
\label{sec:background}

This section defines the main concepts used throughout this paper.
It reviews how calibration is defined and measured and explains why uncertainty in language should be computed over meanings rather than tokens.
The formal analysis defines agent trajectories and calibration at the step and trajectory levels. It explains why calibration at individual steps does not necessarily imply calibration over a full trajectory, a concern that arises in sequential agent behavior.
We use $\tau = (s_0, a_1, o_1, \dots, a_T, o_T)$ to denote an agent trajectory and $h_t = (s_0, a_{1:t}, o_{1:t})$ to denote the history after step $t$. Let $Y \in \{0,1\}$ indicate trajectory success, defined by the task rather than by the steps, and let $Y_t \in \{0,1\}$ indicate success at step $t$. Under an \emph{absorbing-failure} convention, in which no step failure can be repaired later, the two are linked by $Y = \prod_t Y_t$. We denote step confidence by $c_t = c(a_t, h_{t-1})$. At the trajectory level, $R = \Pr(Y=1)$ denotes trajectory reliability, $\hat{R}(h_t)$ its history-conditioned estimate, $A_i(h_t)$ the advantage of an intervention, and $U$ the task utility.

\subsection{Sources of Uncertainty}
\label{sec:sources}

A standard decomposition separates two types of uncertainty. \emph{Aleatoric} uncertainty is irreducible noise or ambiguity in the task itself. \emph{Epistemic} uncertainty reflects the model's lack of knowledge and can, in principle, be reduced with more data or computation \citep{kendall2017what,depeweg2018decomposition,hullermeier2021aleatoric}.
This distinction has a precise information-theoretic form. Let $\theta$ denote model parameters with posterior $p(\theta \mid \mathcal{D})$. For a single model, let $p(y \mid x, \theta)$ be its predictive distribution. The Bayesian predictive distribution $p(y \mid x) = \mathbb{E}_{\theta}[p(y \mid x, \theta)]$ has the exact decomposition
\begin{equation}
\underbrace{H\big(\mathbb{E}_{\theta}[\,p(y \mid x, \theta)\,]\big)}_{\text{total predictive entropy}}
\;=\;
\underbrace{\mathbb{E}_{\theta}\big[H\big(p(y \mid x, \theta)\big)\big]}_{\text{expected aleatoric entropy}}
\;+
\underbrace{I\big(y ;\, \theta \mid x\big)}_{\text{epistemic mutual information}}
\end{equation}
where the expectations are over $\theta \sim p(\theta \mid \mathcal{D})$.
The aleatoric term is the posterior average of the task noise assigned by each model. The epistemic term measures how much the models in the posterior \emph{disagree}. It is zero exactly when the predictive distribution does not depend on $\theta$ under the posterior \citep{depeweg2018decomposition,hullermeier2021aleatoric}.
Classical deep-learning UQ applies this decomposition by approximating the posterior. Common approaches include variational weight distributions \citep{blundell2015weight}, Monte Carlo dropout \citep{gal2016dropout}, and deep ensembles \citep{lakshminarayanan2017simple}. \citet{ovadia2019can} tested these methods under distribution shift, and \citet{abdar2021review} and \citet{gawlikowski2023survey} provide broader surveys. Scaling these posterior approximations to modern networks remains an active problem. Recent approaches share stochastic weights across a network \citep{pmlr-v258-lin25a} or use flow-induced process priors to keep the posterior tractable \citep{lin2026flow}.
\par
For LLMs, this decomposition is difficult to estimate because the true posterior is not available. A deployed LLM is a single checkpoint, and treating it as a degenerate posterior sets the mutual-information term to zero, assigning all remaining uncertainty to aleatoric noise even when it reflects missing knowledge; the decomposition also shifts with the chosen model class and reference distribution \citep{hullermeier2021aleatoric,baan2023uncertainty}. Existing work therefore uses proxies, including LoRA-ensembles \citep{balabanov2024uncertainty,yang2024bayesian}, a posterior over the adapter \citep{lin2026bayesian}, calibration-tuned models \citep{kapoor2024large}, Bayesian views of in-context learning \citep{ling2024uncertainty}, and information-theoretic bounds from iterative prompting \citep{yadkori2024believe}. The distinction remains observable in practice, since small probes can separate inherently unanswerable questions from those a larger model can answer, so hidden states carry information the output distribution alone does not reveal \citep{ahdritz2024distinguishing}.
\par
The distinction matters for agents because the two types call for different \emph{actions}: ambiguity in the request motivates \emph{clarification} \citep{kuhn2023clam,hou2024decomposing}, whereas epistemic uncertainty about a fact motivates \emph{retrieval} or \emph{deferral} \citep{mallen2023not,shin2026era}. Methods that compress both into one scalar cannot by themselves identify the right response; \citet{yadkori2024believe} isolate the epistemic component information-theoretically, \citet{hou2024decomposing} attribute disagreement across clarified interpretations to ambiguity rather than ignorance, and \citet{kirchhof2025position} argue that even two scalars may not suffice for agents.
\citet{papamarkou2026position} go further and argue that the agent pipeline should itself be Bayes-consistent, with each component reporting a posterior that the orchestration layer combines coherently instead of thresholding scalars with different meanings. Under this view every \textcolor{axyellow!75!black}{Panel B} family would output a distribution and the \textcolor{axaqua!65!black}{Panel A} labels would become part of the interface between stages. We treat this as a position rather than an established result, because no agent system in our corpus implements the design.
\par
Agents introduce at least three further sources of uncertainty that are not fully captured by the classical distinction and do not fit cleanly into the decomposition above.
Tools and environments return noisy, partial, or stale observations \citep{han2024towards,chen2024benchmarking,xuan2026confidence}. At step $t$, this noise is aleatoric from the agent's current perspective because it comes from the observation rather than from uncertainty about the model. The agent can sometimes reduce this uncertainty through another action, such as repeating a query to an unreliable API or checking a conflicting search result against another source. This possibility makes the boundary between aleatoric and epistemic uncertainty less clear.
Uncertainty at each step can also accumulate along a trajectory \citep{zhao2024saup,duan2025uprop,zhang2026agentic}. Section~\ref{sec:formal} shows that this accumulation depends on correlations between errors, which marginal per-step quantities cannot capture.
Communication also passes uncertainty between agents in multi-agent systems \citep{yoffe2024debunc,huang2026counterfactual,chen2026every}. An agent that accepts a confident claim from another agent inherits the uncertainty behind that claim, often without knowing its source. Agents may share training histories, so their agreement may reflect correlated errors rather than independent support. Section~\ref{sec:multiagent} explains when this false consensus occurs and how it can be detected.
Table~\ref{tab:sources} summarizes the five sources in \textcolor{axaqua!65!black}{Panel A} of our taxonomy and lists representative responses to each source.

\begin{table}[t]
\centering
\caption{Five sources of agent uncertainty (Panel A), with examples and
corresponding responses. The upper block is classical; the lower block is
specific to agentic interaction.}
\label{tab:sources}
\small
\setlength{\tabcolsep}{5pt}
\renewcommand{\arraystretch}{1.15}
\begin{tabular}{@{}L{3.2cm}L{6.3cm}L{5.5cm}@{}}
\toprule
\thc{Source} & \thc{Agent example} & \thc{Mitigation}\\
\midrule
\tband{3}{The classical decomposition}
\pidx{axaqua!65!black}{A1}\textbf{Aleatoric} & Underspecified user instruction with several valid readings & Request clarification\\
\addlinespace[2pt]
\pidx{axaqua!65!black}{A2}\textbf{Epistemic} & Missing or outdated model knowledge about a required fact & Retrieval or deferral\\
\tband{3}{What the agentic setting adds}
\pidx{axaqua!65!black}{A3}\textbf{Tool / environment} & Unreliable API, contradictory search results, partial observability & Verification, repeated querying, or external grounding\\
\addlinespace[2pt]
\pidx{axaqua!65!black}{A4}\textbf{Accumulated} & Early parsing error that propagates to subsequent steps & Replanning, rollback, or early intervention\\
\addlinespace[2pt]
\pidx{axaqua!65!black}{A5}\textbf{Inter-agent} & False consensus from correlated agent failures & Claim weighting, arbitration, or retention of dissenting views\\
\bottomrule
\end{tabular}
\end{table}

\subsection{Calibration and Its Metrics}
\label{sec:metrics}

A confidence estimate is \emph{calibrated} if, among all predictions made with confidence $p$, a fraction $p$ are correct: writing $Y \in \{0,1\}$ for correctness and $c$ for the reported confidence, $\Pr(Y = 1 \mid c = p) = p$ for every $p$ in the range of $c$. Calibration does not require confidence to be informative, since the constant predictor $c \equiv \Pr(Y=1)$ is perfectly calibrated yet says nothing about which predictions are more likely correct. The concept predates deep learning; the definition matches \citeauthor{dawid1982well}'s \emph{well-calibrated} forecaster \citep{brier1950verification,dawid1982well,gneiting2007strictly}, and \citet{degroot1983comparison} separate calibration from \emph{refinement}, the same distinction as reliability against resolution, which explains how a near-constant predictor can be calibrated but useless. Section~\ref{sec:recalibration} covers the separate family of methods that correct a miscalibrated estimate.
\par
Calibration is typically measured with the expected calibration error (ECE), whose population form is $\mathbb{E}_{c}[\,\lvert \Pr(Y=1 \mid c) - c \rvert\,]$ and whose standard estimator bins $n$ predictions into $M$ confidence bins and averages the absolute accuracy--confidence gap per bin, weighted by bin mass \citep{naeini2015obtaining}. Proper scoring rules assess probabilities without binning. The Brier score $\mathrm{BS} = \frac{1}{n}\sum_i (c_i - Y_i)^2$ is strictly proper \citep{brier1950verification,gneiting2007strictly}, and \citet{murphy1973vector} decompose it into a \emph{reliability} term that vanishes under perfect calibration, a \emph{resolution} term that rewards separating easy from hard cases, and an irreducible base-rate term, which quantifies the caveat above: a useful signal needs both reliability and resolution. The logarithmic score is also strictly proper and is the standard sequence-level score for language models \citep{jiang2021how,malinin2021uncertainty}.
\par
Two other families treat confidence as a \emph{ranking}. AUROC measures how well confidence separates correct from incorrect outputs, equal (absent ties) to $\Pr(c_i > c_j \mid Y_i = 1, Y_j = 0)$, and has been standard for misclassification and hallucination detection since \citet{hendrycks2017baseline}. Selective prediction thresholds the ranking, answering when $c \geq \gamma$; varying $\gamma$ trades coverage against selective risk along the risk--coverage curve, whose area is the AURC \citep{geifman2017selective,elyaniv2010foundations}. This follows Chow's rejection rule under symmetric costs \citep{chow1970optimum} and extends to \emph{learning to defer}, where the downstream expert's error rate enters the objective \citep{madras2018predict,mozannar2020consistent}; selective prediction for LLMs uses the same framework \citep{cole2023selectively,srinivasan2024selective}, and Section~\ref{sec:control} gives the formal statement and its limits for agents. Table~\ref{tab:metrics} lists the metrics used throughout this paper and their main limitations for agents.
\par
Calibration has important limitations. \citet{kalai2024calibrated} show that a language model calibrated on a corpus must have a nonzero hallucination rate for facts that appear only once in training; this rate is bounded below by the fraction of such facts. This lower bound implies that calibration alone cannot eliminate fabrication for arbitrary, low-frequency claims. The result does not imply that calibration is an unsuitable objective: a miscalibrated model may also hallucinate without providing a reliable signal of when this occurs. For agents, the practical value of confidence depends on how it guides action. Retrieval, abstention, or clarification at a checkpoint can prevent errors that confidence estimation alone cannot.
\par
Evidence on LLM calibration is mixed. \citet{kadavath2022language} find that sufficiently large models are approximately calibrated on multiple-choice questions under suitable prompting. They also find that these models can assess the validity of their own claims, providing early evidence for the self-evaluation signals used by later methods. Verbalized confidence provides a direct way to elicit such signals. The same study finds that RLHF-style tuning can make these reports less reliable and more overconfident. Section~\ref{sec:foundations} discusses elicitation methods, training objectives, and countermeasures in detail.
\par
Measurement choices also affect the interpretation of empirical results.
The binned ECE is a plug-in estimator carrying discretization error and finite-sample bias; its value depends on the bin count and on equal-width against equal-mass binning, and these choices can change method rankings \citep{nixon2019measuring,kumar2019verified,roelofs2022mitigating}. Debiasing, adaptive binning, and consistent kernel or $L_p$ estimators each address part of the problem, but none makes an arbitrary fixed-bin estimate canonical \citep{widmann2019calibration,popordanoska2022consistent}.
Sequence-level scores add a length confound. Longer outputs accumulate more log-loss, so a measure correlated with output length can appear predictive merely because correctness is also length-correlated, an issue that confounds common LLM-UQ evaluations \citep{santilli2025revisiting}; length normalization reduces but does not remove it \citep{malinin2021uncertainty}. Agents are especially exposed because trajectory lengths vary widely, and many metrics in Table~\ref{tab:metrics} also require step-level correctness labels that are often unavailable.
\par
One further caveat is specific to language. Different token sequences express the same meaning, so token-level entropy can be high where semantic uncertainty is low; ``Paris,'' ``It's Paris,'' and ``The capital of France is Paris'' are one meaning. When correctness is defined semantically, uncertainty should be computed over meanings rather than surface strings \citep{kuhn2023semantic,farquhar2024detecting}, a point already known from pre-LLM machine translation and structured prediction \citep{ott2018analyzing,malinin2021uncertainty,baan2023uncertainty}.
Section~\ref{sec:sampling} reviews estimators that apply this idea. It also discusses a central obstacle to applying them to agents: sentence equivalence concerns meaning, whereas action equivalence also depends on task-specific consequences.

\begin{table}[t]
\centering
\caption{Common uncertainty metrics and their limitations when applied to agent
trajectories. The rows cover calibration, ranking, prediction sets, and
trajectory-level targets.}
\label{tab:metrics}
\small
\setlength{\tabcolsep}{5pt}
\renewcommand{\arraystretch}{1.15}
\begin{tabular}{@{}L{4.8cm}L{4.2cm}L{5.5cm}@{}}
\toprule
\thc{Metric} & \thc{What it measures} & \thc{Limitations for agents}\\
\midrule
\tband{3}{Proper scoring and calibration}
Sequence NLL, perplexity \citep{jiang2021how,malinin2021uncertainty} & Likelihood the model assigns to its own output & Token-level; overstates paraphrase disagreement; blind to tool feedback \citep{kuhn2023semantic}\\
\addlinespace[2pt]
ECE \citep{naeini2015obtaining,guo2017calibration} & Binned gap between stated confidence and empirical accuracy & Needs a correctness label per prediction, often unavailable mid-trajectory; binning-dependent \citep{nixon2019measuring}\\
\addlinespace[2pt]
Brier score \citep{brier1950verification,gneiting2007strictly} & Proper scoring of probabilistic forecasts & Same labeling problem; hides \emph{where} in the trajectory miscalibration occurs\\
\tband{3}{Ranking and selective prediction}
AUROC \citep{hendrycks2017baseline} & Whether confidence ranks correct outputs above incorrect ones & Confounded by output length \citep{santilli2025revisiting}, which varies widely across trajectories\\
\addlinespace[2pt]
Risk-coverage / AURC \citep{geifman2017selective,elyaniv2010foundations} & Error at each abstention rate & Ignores the cost of abstaining mid-trajectory: partial work, irreversibility\\
\tband{3}{Meanings and prediction sets}
Semantic entropy \citep{kuhn2023semantic,farquhar2024detecting} & Dispersion across the meanings in sampled outputs & Sampling whole trajectories is costly; equivalence between \emph{actions} is hard to define\\
\addlinespace[2pt]
Coverage of conformal sets \citep{vovk2005algorithmic,angelopoulos2023gentle} & Whether prediction sets contain the true outcome at the target rate & Exchangeability fails across history-dependent steps \citep{barber2023conformal}\\
\tband{3}{Trajectory-level targets}
Task success rate \citep{liu2024agentbench,mialon2023gaia} & End-to-end task performance & Scores the agent, not the confidence estimate; silent on whether failure was anticipated\\
\addlinespace[2pt]
Intervention advantage \citep{zhang2026calibration} & Utility gain from intervening now rather than continuing & Decision-relevant, but estimators and benchmarks are not yet standardized\\
\bottomrule
\end{tabular}
\end{table}

\subsection{Why Step Calibration Does Not Compose}
\label{sec:formal}

We model an agent as producing the trajectory $\tau$ of Section~\ref{sec:background}, in which $s_0$ is the initial state, $a_t$ an action (including a tool call, retrieval, or message), and $o_t$ an observation returned by a tool or the environment.
In a single-turn setting, UQ assigns a confidence score $c(a, s_0)$ to one action given the initial state.
Agent UQ instead considers a history-dependent confidence estimate at each step, $c(a_t \mid s_0, a_{1:t-1}, o_{1:t-1})$. It also considers trajectory-level reliability $R$, defined as the probability that the trajectory achieves the goal.
The distinction between these quantities can be expressed through step-level and trajectory-level calibration.
Throughout this subsection $Y_t$ is read as the \emph{local} goal of step $t$, so the absorbing-failure link $Y = \prod_t Y_t$ of Section~\ref{sec:background} applies to locally defined step labels; the remark after Proposition~\ref{prop:compose} returns to that choice.

\begin{definition}[Step-level calibration]
\label{def:step}
A step estimator $c_t = c(a_t, h_{t-1})$ is calibrated if
$\Pr\left(Y_t = 1 \mid c_t = p\right) = p$ for all $p$ in its range.
\end{definition}

\begin{definition}[Trajectory-level calibration]
\label{def:traj}
A reliability estimator $\hat{R}(h_t)$ is calibrated if, for every checkpoint
index $t$ separately,
$\Pr\left(Y = 1 \mid \hat{R}(h_t) = r\right) = r$ for all $r$ in the
estimator's range, where the probability at index $t$ is taken over the
trajectories that reach checkpoint $t$.
\end{definition}

\noindent\emph{Remark.} The order of the quantifiers matters: calibration is
required separately for each $t$, not after pooling $(r,t)$ pairs.
Miscalibration at late checkpoints cannot be offset by miscalibration in the
opposite direction at early checkpoints. The conditioning population also
changes with $t$: index $t$ is scored on the survivor subset of trajectories
that run for at least $t$ steps, a subset that shrinks and drifts as the
horizon grows. Pooled variants are weaker;
Section~\ref{sec:eval} discusses them together with the estimation protocol.

Both definitions require a well-defined step label $Y_t$.
In practice, this requires choosing what counts as step-level success. Whether step $t$ achieves its local goal may depend on later events in the trajectory. This creates a labeling problem, which we discuss in Section~\ref{sec:transition} and revisit in Section~\ref{sec:eval}.
The formal statements below apply to any fixed choice of $Y_t$, but that choice determines the real-world quantity they describe.
\par
Definition~\ref{def:step} captures the calibration target commonly used in single-turn UQ \citep{kadavath2022language,guo2017calibration}, whereas Definition~\ref{def:traj} gives the trajectory-level target relevant to agent deployment.
Both propositions below assume the absorbing-failure convention $Y=\prod_t Y_t$, and their scope is exactly that assumption. Agents that retry a failed tool call, re-issue a query, or re-plan after a wrong step do not satisfy it, and for them a step failure does not determine $Y$; Section~\ref{sec:control} treats recoverability as the property that separates two histories carrying equal risk. Definition~\ref{def:traj}, Definition~\ref{def:tcece}, and the experiment of Section~\ref{sec:empirical} do not rely on the convention, because they score $\hat{R}(h_t)$ against an independently observed task outcome $Y$ rather than against a product of step labels.
A simple way to relate the two is to multiply confidence scores across steps, but this calculation generally ignores dependence among step outcomes.
For dependent Bernoulli variables, the product of marginal probabilities need not equal the joint probability. Proposition~\ref{prop:compose} records this standard result for agent trajectories and separates the two calibration notions used in this paper. This distinction matters when a running product of step confidences is used as an estimate of trajectory reliability. Proposition~\ref{prop:positive} later quantifies the error in the product of step marginals and identifies its source.

\begin{proposition}[Calibration does not compose]
\label{prop:compose}
There are agents whose step-level confidence estimates are perfectly calibrated according to Definition~\ref{def:step}, but the product $\prod_t c_t$ can still be miscalibrated for trajectory-level success, either overestimating or underestimating the true probability of success.
\end{proposition}

\noindent\emph{Proof sketch.}
Consider a two-step agent in which both steps are always executed, so that $Y = Y_1Y_2$. Suppose the success of both steps depends on the same fair coin.
In the correlated case, both steps succeed when the coin lands heads.
Each step succeeds with probability $\tfrac{1}{2}$, so the constant estimator $c_t=\tfrac{1}{2}$ is perfectly calibrated at the step level. Multiplying the two confidence scores gives $\tfrac{1}{4}$, even though the true trajectory success probability is $\tfrac{1}{2}$.
In the anti-correlated case, step 1 succeeds on heads while step 2 succeeds on tails. The same step-level estimator is still perfectly calibrated, and the product still gives $\tfrac{1}{4}$, but now the true trajectory success probability is $0$. $\square$
\par
\medskip The discrepancy arises from dependence among step outcomes, which per-step confidence scores alone do not capture. For two steps, the effect can be written explicitly.
For two steps with non-degenerate outcomes, let $p_1 = \Pr(Y_1=1)$ and $p_2 = \Pr(Y_2=1)$ be the marginal success probabilities, and let $\rho$ be the Pearson correlation between $Y_1$ and $Y_2$. Their joint success probability is
\begin{equation}
\Pr(Y_1 = 1,\, Y_2 = 1)
\;=\;
p_1 p_2
\;+\;
\rho\, \sqrt{p_1 (1 - p_1)}\, \sqrt{p_2 (1 - p_2)}
\end{equation}
The joint probability varies linearly with $\rho$, equals the independence product at $\rho=0$, and reaches the Fr\'echet--Hoeffding bounds $\max(0,\, p_1 + p_2 - 1) \leq \Pr(Y_1 = 1, Y_2 = 1) \leq \min(p_1, p_2)$ at the extremes; the proof sketch's two examples are exactly $\rho=1$ (joint probability $\tfrac{1}{2}$) and $\rho=-1$ (joint probability $0$). Here $p_1$ and $p_2$ are step \emph{marginals}, not the reported confidences of Definition~\ref{def:step}, which satisfy $\mathbb{E}[c_t]=p_t$ under step-level calibration without any individual $c_t$ having to equal $p_t$; the remark after Proposition~\ref{prop:positive} develops this distinction.
Early errors can affect later decisions and induce positive correlations among step outcomes \citep{zhang2026agentic,zhou2026exploring}, in which case multiplying the step \emph{marginals} underestimates joint success. Neither scalar product identifies which step is likely to fail, and beyond two steps pairwise correlations do not determine the joint distribution, so marginals and pairwise statistics alone cannot guarantee calibrated trajectory-level confidence. This motivates propagation methods that model dependence across steps rather than assuming independence \citep{zhao2024saup,duan2025uprop}, reviewed in Section~\ref{sec:propagation}, and explicitly Markovian formulations of trajectory reliability \citep{trantruong2026measuring}.
\par
\medskip\noindent\emph{Remark (the counterexample depends on the step-label convention).}
Proposition~\ref{prop:compose} is a statement about locally defined step labels rather than about trajectories as such. Consider instead a \emph{survival} convention that sets $Y_t=1$ exactly when the prefix $h_t$ still admits a successful continuation. Under this convention, $Y=\prod_t Y_t$ holds by construction rather than by assumption, and the anti-correlated branch of the counterexample dissolves. When step 1 achieves its local goal but makes the failure of step 2 inevitable, the survival label sets $Y_1=0$, so the constant estimator $c_1=\tfrac{1}{2}$ is no longer calibrated for that label. Non-composition is therefore a product of the local-goal convention, and the survival convention removes it at a price. The survival label depends on whether any successful continuation exists, which is not locally observable during execution, and estimating it is precisely the trajectory-reliability problem $\hat{R}(h_t)$ that this paper studies. Choosing a step-label convention is thus a modeling decision that trades local observability against composability.
\par
The chain rule gives an exact product expression in terms of conditional step-success probabilities. Comparing this expression with the product of marginal step-success probabilities yields a bound on the approximation error.
Let $\lambda_t = \Pr\left(Y_t = 0 \mid Y_{1:t-1} = \mathbf{1}\right)$ denote the conditional probability that step $t$ fails given that all earlier steps succeeded. By the chain rule,
\begin{equation}
R = \Pr(Y=1) = \prod_{t=1}^{T} \Pr\left(Y_t=1 \mid Y_{1:t-1}=\mathbf{1}\right) = \prod_{t=1}^{T}(1-\lambda_t)
\label{eq:hazard}
\end{equation}
Whenever the conditional probabilities are defined, Eq.~\eqref{eq:hazard} is exact regardless of dependence among step outcomes. Define the \emph{survival factor} $q_t=1-\lambda_t$ as the probability that step $t$ succeeds given that all earlier steps succeeded. This conditional probability may differ from the marginal success probability $p_t=\Pr(Y_t=1)$. The implications of Eq.~\eqref{eq:hazard} for monitoring and estimation are discussed after Proposition~\ref{prop:positive}.
Proposition~\ref{prop:positive} bounds the error of the marginal-product approximation $\prod_t p_t$. The bound depends on the step marginals, prefix reliabilities, and correlations between each step outcome and the corresponding clean-prefix indicator. We call these correlations \emph{prefix-coupling coefficients}. At $k=2$ the coefficient coincides with the pairwise correlation between the first two step outcomes, but at later steps it couples $Y_k$ to the joint indicator of a clean prefix and should not be read as a pairwise correlation between two steps. All three quantities can be estimated from held-out trajectories to assess the possible size of the approximation error.
This result does not directly characterize the running product of an agent's reported confidences, $\prod_t c_t(h_{t-1})$, which is a history-dependent random variable. Step-level calibration in Definition~\ref{def:step} constrains only $\Pr(Y_t=1 \mid c_t=p)=p$ and does not imply that a particular value $c_t(h_{t-1})$ equals the marginal $p_t$. Proposition~\ref{prop:positive} concerns the difference between the joint probability and a product of marginal accuracies, not the calibration error of $\prod_t c_t$. The remark after the proposition explains the relation between these quantities.

\begin{proposition}[Composition error under dependent step outcomes]
\label{prop:positive}
Fix any step-label convention. Assume that every prefix event used for
conditioning has positive probability, so $R_{k-1} > 0$ for all $k \le T$.
Trajectories with a zero-probability prefix fall outside this statement because
later survival factors are undefined.
Let
$p_t = \Pr(Y_t = 1)$ denote the step marginals. Let
$R_k = \Pr(Y_{1:k} = \mathbf{1})$ denote the prefix reliabilities, with
$R_0 = 1$ and $R_T = R$. Let
$q_k = \Pr(Y_k = 1 \mid Y_{1:k-1} = \mathbf{1})$ denote the survival factors
in Eq.~\eqref{eq:hazard}. When both variables are non-degenerate, let $\rho_k$
denote the \emph{prefix-coupling coefficient}, the Pearson correlation between
$Y_k$ and the clean-prefix indicator $\prod_{s<k} Y_s$. Then
\begin{equation}
R \;-\; \prod_{t=1}^{T} p_t
\;=\;
\sum_{k=2}^{T} \Big(\prod_{t<k} q_t\Big)\,(q_k - p_k)\,\Big(\prod_{t>k} p_t\Big),
\qquad\text{hence}\qquad
\Big| R - \prod_{t=1}^{T} p_t \Big| \;\le\; \sum_{k=2}^{T} \lvert q_k - p_k \rvert
\end{equation}
and, when $0 < R_{k-1} < 1$ and $0 < p_k < 1$, each
conditional--marginal gap has the closed form
\begin{equation}
q_k - p_k
\;=\;
\rho_k\, \sqrt{p_k\,(1 - p_k)}\, \sqrt{\frac{1 - R_{k-1}}{R_{k-1}}}
\end{equation}
When $R_{k-1}=1$ or $p_k \in \{0,1\}$, the gap is zero.
In particular, the product of marginals equals trajectory reliability when
$q_k=p_k$ for every $k$, which is guaranteed under mutual independence.
Equivalently, when the prefix indicators and step outcomes are non-degenerate,
this condition holds when $\rho_k=0$ for all $k=2,\ldots,T$.
\end{proposition}

\noindent\emph{Proof sketch.}
The equality follows from a telescoping identity together with Eq.~\eqref{eq:hazard}, and the closed form follows by expanding the covariance between $Y_k$ and the clean-prefix indicator $\prod_{s<k} Y_s$. Appendix~\ref{app:proofs} gives the full proof. $\square$
\par
\medskip\noindent\emph{Remark (relation to the reported-confidence product $\prod_t c_t$).}
The estimator used in practice is $\prod_t c_t(h_{t-1})$, not $\prod_t p_t$. Step-level calibration links each reported confidence to its marginal success probability only in expectation. If each $c_t$ is calibrated in the sense of Definition~\ref{def:step}, then $\mathbb{E}[c_t] = p_t$, but a realized value $c_t(h_{t-1})$ need not equal $p_t$, and expectation cannot in general be moved through the product. The observed composition error $\prod_t c_t - R$ therefore separates pathwise into a reported-confidence--marginal gap and the dependence term of Proposition~\ref{prop:positive}, two mechanisms that are often combined under the phrase ``the product rule fails'' and that can act in opposite directions. Only the dependence term is controlled here; under nonnegative clean-prefix correlations it makes the marginal product \emph{conservative}, so a large overestimate of trajectory reliability cannot be explained by positive cross-step dependence alone and instead motivates separate checks of the calibration and dependence structure of the reported confidences $c_t$. Appendix~\ref{app:proofs-pathwise} states the identity, quantifies both terms on the traces of Section~\ref{sec:empirical}, and derives a companion bound on the mean bias of $\prod_t c_t$; neither bound establishes calibration of the product, and a multicalibration-style result that would is listed as an open problem in Section~\ref{sec:challenges}.
\par
\medskip Proposition~\ref{prop:positive} has practical implications. The marginal product may be a useful approximation for short trajectories and for pipelines that reduce step dependence through independent verification or context isolation, and the bound, whose ingredients $\rho_k$, $p_k$, and $R_{k-1}$ are all measurable on labeled trajectories, quantifies how far the product can stray. The proposition also explains why long trajectories are difficult \emph{to estimate}, a point distinct from the exponential decay of $R$ itself: the factor $\sqrt{(1-R_{k-1})/R_{k-1}}$ grows as prefix reliability falls, so even modest coupling opens a large gap at long horizons, and under nonnegative clean-prefix correlations the marginal product understates reliability. For the same reason the bound degrades with the horizon; in a representative $T=20$ setting it exceeds the trivial bound of $1$ once correlations reach roughly $0.13$, so it is mainly informative for short horizons or near-zero coupling (Appendix~\ref{app:proofs}), although the signed decomposition of Appendix~\ref{app:proofs-pathwise} remains valid at every horizon. Estimating $q_k$ and $\rho_k$ also requires labeled trajectories from the deployment distribution, and clean prefixes become rare as the horizon grows, so when enough labels exist, fitting a trajectory-level calibrator directly may be preferable; the proposition explains how reliability depends on step relationships rather than replacing empirical calibration.
\par
The practical value of the bound therefore depends on the strength of step dependence observed in practice.
Direct measurements in deployed agents remain scarce. Prior work mainly provides qualitative evidence that early mistakes propagate \citep{zhang2026agentic,zhou2026exploring}. Section~\ref{sec:empirical} measures this dependence directly on chained QA items. It finds a small positive step-error correlation in one of four conditions: reasoning chains that share a context. In that condition, realized reliability differs from the marginal product in the direction predicted by Proposition~\ref{prop:positive} and remains within its bound. The empirical evidence is limited: the interval barely excludes zero, only one of four tested conditions is positive, and no multiplicity correction is applied. These results support a hypothesis about reused reasoning content rather than an established mechanism. Section~\ref{sec:empirical} also measures the coupling on ALFWorld, where a real action--observation loop and exact per-step labels are available. There the within-game prefix-coupling coefficient on outcome-tracking step labels is positive, with game-bootstrap intervals that exclude zero in both models tested, and the marginal product understates prefix reliability at every horizon beyond a few steps, in the direction and within the bound that the proposition gives. The coupling is therefore stronger in a genuine agentic loop, where later steps directly use earlier actions and observations, than in chains that share only a context window. We did not find direct measurements of this effect elsewhere in our corpus.
Equation~\eqref{eq:hazard} gives a survival-based view of trajectory reliability, with $\lambda_t$ playing the role of a discrete-time hazard in the sense of classical survival analysis \citep{kalbfleisch2002statistical}. At each step, the relevant quantity is the probability of success conditional on all previous steps having succeeded, rather than the marginal success probability.
For online monitoring, the hazard can also be conditioned on the observed history: $\lambda_t(h_{t-1}) = \Pr\left(Y_t=0 \mid Y_{1:t-1}=\mathbf{1}, h_{t-1}\right)$. Averaging this quantity over histories that reach step $t$ without an earlier failure recovers $\lambda_t$.
Equation~\eqref{eq:hazard} distinguishes the valid multiplication of conditional survival probabilities from the naive multiplication of marginals. The naive product $\prod_t c_t$ is not guaranteed to condition on all previous steps having succeeded, and the equality $q_t=p_t$ that would license it, with independence as one sufficient condition, can fail in agent trajectories because earlier errors may remain hidden in the observable history \citep{zhang2026agentic}. The equation also shows the horizon effect directly: a constant hazard gives $R=(1-\lambda)^T$, so reliability decays exponentially with length, which raises the stakes for long-horizon execution \citep{wei2026longhorizon}, and the stepwise hazard $\lambda_t(h_{t-1})$ is a natural target for the monitoring methods of Section~\ref{sec:control}.
Hazard estimation also creates a statistical challenge for protocols based on Definition~\ref{def:traj}. The hazard at step $t$ is informed only by the roughly $N R_{t-1}$ trajectories that survive to it, an effective sample size that shrinks geometrically with the horizon and gives the marginal-hazard estimate a standard error of order $O\big((N R_{t-1})^{-1/2}\big)$; conditioning on the history requires further modeling assumptions. Trajectory-calibration results should therefore be reported separately by horizon with group sizes, so that poor late-step calibration is not hidden under the mass of early-step examples (Section~\ref{sec:eval}), and rollout-based labeling can supply Monte Carlo estimates of $\Pr(Y=1\mid h_t)$ where surviving prefixes are rare \citep{wang2024mathshepherd}.
\par
The chain rule for entropy gives a related but weaker view, separating uncertainty in the agent's actions from uncertainty introduced by tools and the environment \citep{han2024towards}. Even that full decomposition does not provide trajectory reliability, because policy entropy may reflect benign diversity among equally valid actions rather than the probability of task success \citep{kuhn2023semantic}. Appendix~\ref{app:proofs} states the decomposition and its limits.
Trajectory reliability alone does not determine an intervention decision. A calibrated $\hat{R}(h_t)$ estimates the probability of trajectory failure but does not determine whether intervention improves expected utility. Following \citet{zhang2026calibration}, we define the \emph{intervention advantage} of an oversight action $i$ at history $h_t$ as
\begin{equation}
A_i(h_t)
=
\mathbb{E}\left[U \mid \mathrm{do}(i), h_t\right]
-
\mathbb{E}\left[U \mid \mathrm{continue}, h_t\right]
\label{eq:advantage}
\end{equation}
where $U$ is the task utility. If intervention costs are not included in $U$, an intervention is beneficial when its expected utility gain exceeds its cost.
Under the loss assumptions of Chow's rule, thresholding a calibrated risk score is optimal in a one-step setting, whereas learning-to-defer routes selected cases to another expert \citep{chow1970optimum,madras2018predict,mozannar2020consistent}. Section~\ref{sec:control} gives the formal statement and explains its failure modes for agents.
For agent trajectories, the same failure risk does not always imply the same need for intervention. Two trajectories may have the same value of $1-\hat{R}(h_t)$ but differ greatly in how easy it is to recover from an error. For example, opening an incorrect file may be easy to reverse, whereas sending an email to the wrong recipient may have irreversible consequences. Similar differences between reversible and irreversible actions also appear in control settings studied by \citet{greenblatt2024ai}.
Thresholding a calibrated risk score can therefore be insufficient. Risk quantifies the probability of failure, whereas intervention advantage quantifies whether an action is expected to improve utility. We return to this distinction in Sections~\ref{sec:control} and~\ref{sec:challenges}.

Appendix~\ref{app:pomdp} connects these quantities to classical decision theory, including POMDPs and the value of information.

\takeaway[sec:formal]{Step-level calibration does not compose into trajectory-level
calibration. The observed composition error separates into a
step-miscalibration term and a dependence term that can act in opposite
directions, and only the dependence term is bounded by
Proposition~\ref{prop:positive}. The quantity to estimate online is the
history-conditioned hazard $\lambda_t(h_{t-1})$ of Eq.~\eqref{eq:hazard},
whose effective sample size decays geometrically with the horizon. Even a
calibrated risk estimate does not by itself determine when to intervene;
that decision is governed by the intervention advantage $A_i(h_t)$ of
Eq.~\eqref{eq:advantage}.}

\section{The Taxonomy and Paper Roadmap}
\label{sec:taxonomy}

This section defines the three axes used to classify every method in the corpus. It also explains how the remaining sections use these axes.
We also distinguish the \emph{interaction regime} in which a method operates; the roadmap in Section~\ref{sec:roadmap} tags each methods section with the regime it covers. There are four regimes: a \rgST{} setting, in which one prompt produces one answer; a \rgSA{} setting, in which one agent conditions on its earlier outputs across a trajectory; a \rgMA{} setting, in which several agents exchange messages; and the \rgSY{} view, which treats the full run as one object and studies how uncertainty accumulates along it.
The distinction between \rgSA{} and \rgMA{} is important. Both are agentic settings, but they involve different dependence structures. Section~\ref{sec:formal} studies correlation associated with one agent conditioning on its own earlier outputs. Section~\ref{sec:multiagent} studies correlation associated with shared training histories and communication across agents. Combining these settings into one category would hide which dependence structure a method addresses.

\subsection{Three Taxonomic Axes and the Distribution of Existing Work}

Our taxonomy describes each method along three axes. Figure~\ref{fig:tree} presents them as \textcolor{axaqua!65!black}{\textbf{Panel A}}, \textcolor{axyellow!75!black}{\textbf{Panel B}}, and \textcolor{axblue!65!black}{\textbf{Panel C}}, in the order used in this paper.
\textcolor{axaqua!65!black}{\textbf{Panel A}} identifies the \emph{source of uncertainty} addressed by a method. Table~\ref{tab:sources} introduces these sources, and Section~\ref{sec:sources} discusses them in detail.
\textcolor{axyellow!75!black}{\textbf{Panel B}} identifies the \emph{method family}. It contains nine indexed families, \textcolor{axyellow!75!black}{B1}--\textcolor{axyellow!75!black}{B9}. Five are uncertainty estimators inherited from single-turn research. Three are agent-native families: uncertainty propagation and trajectory-level estimation, abstention and deferral, and uncertainty-aware training. The remaining \emph{other} family contains analyses, benchmarks, and surveys that introduce no method of their own.
\textcolor{axblue!65!black}{\textbf{Panel C}} identifies the \emph{pipeline stage} at which a method is applied. Its categories include planning, tool use, retrieval, memory, multi-step reasoning, and multi-agent coordination. Single-turn and agent-general methods are recorded as broader categories.

The panels answer different questions. Panel A describes \emph{what} a method aims to measure, Panel B describes \emph{how} uncertainty is estimated or used, and Panel C describes \emph{where} the method operates. For most labels, a paper's value on one panel does not determine its value on another. The cross-tabulation in Figure~\ref{fig:heatmap} therefore captures combinations across axes rather than reducing to a block-diagonal structure.
We do not claim that the axes are statistically independent or that they partition the corpus. Figure~\ref{fig:heatmap} shows a non-uniform distribution, and some label pairs are structurally related. A paper may also receive several source, family, or stage tags.
One pair is nested by definition. Inter-agent uncertainty \textcolor{axaqua!65!black}{A5} exists only when agents communicate. Every paper with this source tag therefore also has the multi-agent stage tag \textcolor{axblue!65!black}{C6}. The reverse implication does not hold. Of the $23$ corpus papers tagged \textcolor{axblue!65!black}{C6}, $12$ study uncertainty created or transmitted through communication and consensus and therefore carry \textcolor{axaqua!65!black}{A5}. The remaining $11$ use a multi-agent scaffold for another purpose, such as a domain application or an architecture survey, and assign uncertainty to other sources.

The labels within a panel do not always sit at the same level of abstraction. \textcolor{axaqua!65!black}{Panel A} places the two classical uncertainty \emph{types} beside three labels that describe where uncertainty enters or how it moves, and these agentic labels can overlap the classical pair. \textcolor{axyellow!75!black}{Panel B} combines estimation mechanisms with a control mechanism, abstention and deferral \textcolor{axyellow!75!black}{B7}, and a learning procedure, uncertainty-aware training \textcolor{axyellow!75!black}{B8}, so it should be read as describing how uncertainty is \emph{estimated or used}. \textcolor{axblue!65!black}{Panel C} likewise combines pipeline components with an interaction regime and an architectural setting, which is why the interaction regime is marked separately by the roadmap badges (Section~\ref{sec:roadmap}) rather than read off Panel C. Appendix~\ref{app:taxnotes} records these conventions in full. It explains why the five source labels remain one panel, why estimation, control, and training are not split into a separate response axis, and how a four-dimensional reading of the taxonomy can be recovered from the three panels.

Figure~\ref{fig:tree} provides the canonical definition of \textcolor{axyellow!75!black}{Panel B}. It contains exactly nine indexed families, \textcolor{axyellow!75!black}{B1}--\textcolor{axyellow!75!black}{B9}, and these meanings remain fixed throughout this paper. The other displays are projections of this taxonomy.
Table~\ref{tab:families} expands the eight indexed families \textcolor{axyellow!75!black}{B1}--\textcolor{axyellow!75!black}{B8}. It omits the residual category \textcolor{axyellow!75!black}{B9} and adds two unindexed practical rows, internal-state probes and post-hoc calibration, for a total of ten rows.
Table~\ref{tab:bigtable} in Appendix~\ref{app:bigtable} retains these labels and adds an unindexed self-evaluation label to make the background literature easier to locate. These entries are evidence-table subcategories rather than additional Panel-B families. Figure~\ref{fig:heatmap} uses only the eight indexed columns \textcolor{axyellow!75!black}{B1}--\textcolor{axyellow!75!black}{B8}. 
One family has different short labels across displays. Category \textcolor{axyellow!75!black}{B8} is called \emph{training and distillation} in Figure~\ref{fig:tree}, \emph{uncertainty-aware training} in Table~\ref{tab:families}, and \emph{distillation} in the column header of Figure~\ref{fig:heatmap}. All three label the same family, which Section~\ref{sec:training} treats as covering both topics.

The leaf counts in Figure~\ref{fig:tree} are tag counts rather than paper counts. Papers may carry several tags per panel, some papers carry no tag on a given panel, and calibration is deliberately recorded as a property rather than a leaf, as explained below. The counts therefore neither cover every paper on every panel nor sum to the corpus size. Appendix~\ref{app:taxnotes} gives the exact per-panel accounting. The clearest case is memory, whose $31$ papers mostly co-occur with retrieval, which is why Section~\ref{sec:rag} describes memory-specific methods as scarce despite the large stage count.

Crossing \textcolor{axblue!65!black}{Panel C} with \textcolor{axyellow!75!black}{Panel B} gives the matrix in Figure~\ref{fig:heatmap}. We present the figure in Section~\ref{sec:transition}, after introducing the relevant mechanisms. Its distribution also explains the relative emphasis of later sections.
The largest cell is retrieval $\times$ sampling and consistency, with $18$ corpus papers. This concentration may reflect how readily single-turn resampling transfers to retrieval settings, where correctness labels are often available, horizons are short, and decisions are relatively local.
The next largest cell is multi-step reasoning $\times$ propagation, with $15$ papers. This combination is more closely associated with agentic settings because errors can accumulate across long trajectories.
Calibration is not a heatmap column. A paper may report calibration while using any estimation mechanism, so we record calibration as a separate property. Treating it as a family would make it the largest cell without explaining how the underlying methods operate.
Sparse cells have different causes. Some combinations are structurally difficult. For example, conformal methods in multi-agent settings face strong inter-agent dependence and may violate standard exchangeability assumptions. Other gaps have less direct technical explanations. Verbalized confidence has received limited attention in planning and multi-agent settings despite its low cost. Distillation methods rarely evaluate whether trajectory-level calibration transfers to the student, even though process-level supervision could support such an evaluation. Probe-based methods have also rarely been applied to \emph{action} correctness, as discussed in Section~\ref{sec:foundations}. These under-covered combinations motivate several directions in Section~\ref{sec:challenges}; the taxonomy-independent sensitivity searches in Appendix~\ref{app:corpus} identify neighboring work and delimit these claims.
Two limitations apply when interpreting sparse cells. The search terms were derived from these axes, as described in Appendix~\ref{app:corpus}. A cell may therefore appear sparse because another research community uses different terminology.
The search covered arXiv, the ACM Digital Library, OpenReview, IEEE Xplore, and PMLR, but coverage can still vary across research communities and platform-specific vocabularies. Work in statistics, robotics, control, human--computer interaction, and decision theory was also examined through reference-list checks and targeted follow-up searches. Section~\ref{sec:challenges} therefore describes sparse cells as under-covered combinations in the surveyed literature. Where a neighboring field contains relevant work, we identify that connection rather than infer absence from the corpus count alone.

\subsection{A Roadmap to the Remaining Sections}
\label{sec:roadmap}

This subsection summarizes the structure of this paper and marks the section that covers each block. The badges below mark the interaction regime each section studies.
Section~\ref{sec:foundations} reviews the \rgST{} inheritance, and Section~\ref{sec:transition} examines what changes when the same model instead acts over a trajectory.
The two methods sections organize the same corpus along different axes. Section~\ref{sec:pipeline} uses \textcolor{axblue!65!black}{Panel C}, the pipeline stage, because its methods differ mainly in where they operate in the agent loop: its planning, tool-use, retrieval, memory, and perception-surface stages study a \rgSA{} setting, and its coordination stage the \rgMA{} setting. The method families used at each stage are diverse and secondary to the setting.
Section~\ref{sec:families} uses \textcolor{axyellow!75!black}{Panel B}, the method family. Propagation, control, and uncertainty-aware training apply across stages rather than belonging to one stage: propagation and the evaluation protocol of Section~\ref{sec:eval} take the \rgSY{} view, abstention and deferral act during a \rgSA{} run, and uncertainty-aware training draws its objectives from \rgALL{} regimes.
These organizations are complementary. A stage-based section identifies which families appear at a given stage, while a family-based section identifies the stages to which a family has been applied. Both sections use \textcolor{axaqua!65!black}{Panel A} to distinguish the uncertainty source addressed by each method.

\subsection{What This Paper Reports and What It Proposes}
\label{sec:ours}

Most claims in this paper report findings from prior work. Our proposals
include the three-panel taxonomy and its indexing in
Figure~\ref{fig:tree}; Propositions~\ref{prop:compose}
and~\ref{prop:positive}, together with Definitions~\ref{def:step}
and~\ref{def:traj}; and a trajectory-aware reporting protocol, of which the
TC-ECE statistic in Definition~\ref{def:tcece} is one component alongside
trajectory-level resampling, the causally restricted horizon baseline,
position stratification, the signed calibration gap, and the nested
bootstrap. They also include the three experiments in
Section~\ref{sec:empirical}, the measured cost--quality trade-off in
Figure~\ref{fig:costquality}, and the ten-problem agenda in
Section~\ref{sec:challenges}, including its interpretation of sparse taxonomy
cells.
The claim that oversight should be evaluated through intervention advantage
rather than risk is due to \citet{zhang2026calibration}, and we adopt it here.
All other claims are attributed to their sources.
Appendix~\ref{app:llm-use} states where large language models were and were
not used in preparing this paper.
Appendix~\ref{app:corpus} documents the construction of the 120-paper core
corpus, the broader set of 454 background works used in the narrative
synthesis, and the targeted gap audit used to examine selected sparse
combinations independently of the taxonomy labels. Works outside the core
corpus are marked with $^{\circ}$ in Table~\ref{tab:bigtable}.
The core corpus provides the evidentiary basis for corpus-level quantitative
claims, including the taxonomy counts, heatmap, and statements about the
distribution of the surveyed literature. The broader narrative synthesis also
draws on background work from uncertainty estimation, statistics, decision
theory, control, robotics, information retrieval, software engineering, and
human--computer interaction, particularly when connecting the agent-UQ
literature to adjacent research areas. These background works provide
conceptual and methodological context but do not enter the 120-paper corpus
counts or the taxonomy heatmap.
Accordingly, claims about sparsity or prevalence are stated specifically with
respect to the core corpus rather than as claims about the absence of work
from the full literature. For selected sparse combinations that motivate the
research agenda, Appendix~\ref{app:corpus} additionally reports a
taxonomy-independent targeted search used to identify neighboring work and
to delimit the corresponding claims.
We release the complete taxonomy assignment for the core corpus as a
machine-readable table in \texttt{data/taxonomy\_table.csv}. Each row records
one core-corpus paper and its pipeline-stage, uncertainty-source, and
method-family tags, supporting direct corpus queries.
Table~\ref{tab:bigtable} is a broader evidence table and reading guide: it
contains the core-corpus papers together with selected background works, and
should therefore not be interpreted as a direct printed copy of the
machine-readable core-corpus table.

\takeaway[sec:taxonomy]{Every corpus paper is indexed along three axes. Panel A records the
uncertainty source, Panel B records how uncertainty is estimated or used, and
Panel C records the pipeline stage. The axes ask largely non-nested questions
but are neither statistically independent nor a partition, and leaf counts
are tag counts rather than paper counts (Appendix~\ref{app:taxnotes}).
Section~\ref{sec:pipeline} walks the corpus by stage and
Section~\ref{sec:families} by family, while the sparse cells of
Figure~\ref{fig:heatmap} motivate the agenda of
Section~\ref{sec:challenges}.}

\section{Single-Turn LLMs}
\label{sec:foundations}

\subsection{Problem Setting}

Uncertainty estimation for agents builds on techniques developed for single-turn models. This section reviews these techniques and their underlying assumptions.
Dedicated surveys cover this material in depth \citep{geng2024survey,huang2024survey,shorinwa2024survey,xia2025survey,liu2025uncertainty}. Our narrower aim is to summarize the model access required by each family, its computational cost, and the limitations that arise when the model operates as an agent rather than as a single-turn answer generator. The subsections below therefore keep the family-level view and the agent-facing limitations; Appendix~\ref{app:singleturn} reviews the individual methods within each family.
The object of study is a model $p_\theta(y \mid x)$ over outputs $y$ given an input $x$. An uncertainty estimator maps the input, the generated output $\hat{y}$, and the available model information to either a scalar confidence $c(x,\hat{y}) \in [0,1]$ or a set-valued prediction. We evaluate an estimator by its calibration, its ability to rank correct outputs above incorrect outputs, and any formal guarantees it provides (Section~\ref{sec:metrics}).
Table~\ref{tab:families} compares the eight indexed families \textcolor{axyellow!75!black}{B1}--\textcolor{axyellow!75!black}{B8} by their core ideas, model access, computational cost, and representative works. Five families are inherited from single-turn work, and three are agent-native. The table also includes two unindexed rows, internal-state probes and post-hoc calibration, for practical comparison. The ten rows do not correspond to ten taxonomy families. Section~\ref{sec:taxonomy} defines the nine Panel-B families, including the residual category \textcolor{axyellow!75!black}{B9}, which is omitted here.

\begin{table}[t]
\centering
\caption{Method families (Panel B), model access, inference cost, and
representative works. BB, TP, and HS denote black-box, token-probability, and
hidden-state access; $T$ is trajectory length and $k$ the sample count. The
lower block contains the three agent-native families.}
\label{tab:families}
\small
\setlength{\tabcolsep}{4pt}
\renewcommand{\arraystretch}{1.15}
\begin{tabular}{@{}lL{4.2cm}ccL{4.6cm}@{}}
\toprule
\thc{Family} & \thc{Core idea} & \thc{Access} & \thc{Cost} & \thc{Representative works}\\
\midrule
\tband{5}{Inherited from single-turn methods (Section~\ref{sec:foundations})}
\pidx{axyellow!75!black}{B1}Verbalized confidence & Elicitation of a confidence statement from the model & BB & $1\times$ & \citet{lin2022teaching,tian2023just,xiong2023can,band2024linguistic}\\
\addlinespace[2pt]
\pidx{axyellow!75!black}{B2}Sampling / consistency & Agreement among resampled generations & BB & $k\times$ & \citet{wang2023selfconsistency,manakul2023selfcheckgpt,lin2023generating,kuhn2023semantic}\\
\addlinespace[2pt]
\pidx{axyellow!75!black}{B3}Token probability & Sequence- or token-level likelihoods and entropies & TP & $1\times$ & \citet{jiang2021how,malinin2021uncertainty,yadkori2024believe,zhu2026towards}\\
\addlinespace[2pt]
\pidx{gray!55}{--\,}Internal states / probes & Classifiers on hidden activations & HS & $\approx1\times$ & \citet{azaria2023internal,burns2023discovering,kossen2024semantic,orgad2025llms}\\
\addlinespace[2pt]
\pidx{axyellow!75!black}{B4}Conformal / selective & Prediction sets or abstention with coverage guarantees & varies & $1$--$k\times$ & \citet{vovk2005algorithmic,kamath2020selective,quach2024conformal,mohri2024language}\\
\addlinespace[2pt]
\pidx{axyellow!75!black}{B5}Bayesian / ensemble & Posterior or ensemble disagreement & varies & $k\times$ & \citet{gal2016dropout,lakshminarayanan2017simple,ovadia2019can,hou2024decomposing}\\
\addlinespace[2pt]
\pidx{gray!55}{--\,}Calibration & Post-hoc mapping that aligns stated confidence with observed accuracy & varies & $\approx1\times$ & \citet{platt1999probabilistic,zadrozny2002transforming,guo2017calibration,ulmer2024calibrating}\\
\tband{5}{Agent-native (Sections~\ref{sec:propagation},~\ref{sec:control}, and~\ref{sec:training})}
\pidx{axyellow!75!black}{B6}Propagation / trajectory & Combination of step-level signals into a trajectory-level estimate & varies & $T\times$ & \citet{zhao2024saup,duan2025uprop,donaldson2026bayesian,darabi2026groundcontrol}\\
\addlinespace[2pt]
\pidx{axyellow!75!black}{B7}Abstention / deferral & Use of uncertainty estimates for stopping, clarification, or handoff decisions & varies & $\approx1\times$ & \citet{ren2023robots,piatrashyn2026redact,edwards2026ask,yeke2026yesman}\\
\addlinespace[2pt]
\pidx{axyellow!75!black}{B8}Uncertainty-aware training & Training for calibration or use of uncertainty as a reward signal & training & offline & \citet{zhang2024rtuning,band2024linguistic,zhang2026selaur,zhou2026exploring}\\
\bottomrule
\end{tabular}
\end{table}

\subsection{Black-Box Signals: Verbalization, Sampling, and Consistency}
\label{sec:verbalized}
\label{sec:sampling}
Two inherited families require only the model's text output. This limited access requirement makes them common choices for agents built on closed APIs.
Verbalized confidence directly elicits a confidence statement in natural language, beginning with dialogue agents trained to hedge at rates matching their empirical error \citep{mielke2022reducing} and models fine-tuned to append a stated probability under what amounts to a proper scoring rule \citep{lin2022teaching}. Calibration of the elicited value depends on the elicitation format and the prompt wording, and alignment tuning can degrade the calibration of the underlying token probabilities \citep{tian2023just,xiong2023can,openai2023gpt4}. Appendix~\ref{app:singleturn-blackbox} reviews the elicitation formats, the training-based refinements, and the recent verbalized-confidence literature.
Verbalized confidence is compatible with black-box agent pipelines because a step-level score $c_t$ can be elicited alongside each action $a_t$ at little additional inference cost.
A common limitation is systematic overconfidence. Models tend to produce high, rounded scores and hedge less than human respondents, while reward-model training may reinforce this pattern \citep{zhou2023navigating,zhou2024relying,leng2024taming}. The effect can become more pronounced over multi-step trajectories \citep{xuan2026confidence}.
A separate limitation concerns the interpretation of reported confidence. \citet{turpin2023language} show that chain-of-thought rationales can be systematically unfaithful: when an unreported feature influences an answer, a model may generate a plausible justification that does not reflect the computation that produced the answer. Verbalized confidence is produced through a similar output process. It may therefore be calibrated as a distribution-level statistic without reliably representing the model's internal state for a particular input. \citet{sharma2024towards} further show that such reports can adapt to perceived interlocutor preferences. This issue is important for agents because a step-level score $c_t$ is elicited from a context produced partly by the agent and may be used as evidence by another agent (Section~\ref{sec:multiagent}). Confidence in an \emph{action} also requires predicting the environment's response, which may lie outside the model's training objective (Section~\ref{sec:transition}).

Another black-box approach samples multiple generations and measures their agreement, as in self-consistency decoding, where the vote share of the most common answer serves as the confidence score \citep{wei2022chain,wang2023selfconsistency}. Because different token sequences can express the same meaning (Section~\ref{sec:metrics}), sampling-based scores are usually computed over meanings rather than surface strings. Semantic entropy clusters the samples into equivalence classes with bidirectional NLI entailment and computes entropy over the classes \citep{kuhn2023semantic,farquhar2024detecting}, and probes can predict it from a single forward pass \citep{kossen2024semantic}. Appendix~\ref{app:singleturn-blackbox} states the estimator, its refinements, and the wider consistency-based family \citep{manakul2023selfcheckgpt,lin2023generating}.
These methods can be applied to agents by resampling individual tool decisions or complete trajectories, although the cost grows with the horizon $T$.
Meaning-level uncertainty is harder to apply to actions than to plain text because surface form does not determine action equivalence. Two different SQL queries may return the same result, whereas identical API calls may produce different results after the environment changes. Equivalence between browsing actions also depends on the agent's goal rather than on the action text.
Semantic entropy over actions therefore requires a task-specific equivalence relation that is not provided by standard NLI models. Variation across sampled trajectories also combines model uncertainty with environmental randomness (Sections~\ref{sec:propagation} and~\ref{sec:control}).

\subsection{White-Box Signals: Token Probabilities, Internal States, and Ensembles}
\label{sec:tokenlevel}
\label{sec:bayesian}
When logits are available, sequence likelihood and entropy provide uncertainty estimates from a single model run \citep{jiang2021how,malinin2021uncertainty}. Self-evaluation offers a related signal. In P(True)-style self-evaluation, confidence is the probability assigned to the token ``True'' when the model is shown the question and its proposed answer \citep{kadavath2022language}. Iterated prompting turns the same access into a lower bound on epistemic uncertainty \citep{yadkori2024believe}, and efficiency-oriented variants avoid full generation \citep{zhu2026towards}; Appendix~\ref{app:singleturn-whitebox} gives the details.

Another white-box approach estimates uncertainty from hidden states rather than from the output distribution. Supervised probes predict truthfulness from intermediate activations \citep{azaria2023internal}, unsupervised variants find a truth direction without labels \citep{burns2023discovering}, and hidden states have been shown to encode error type, unanswerability, and early signs of hallucination before they appear in the output \citep{orgad2025llms,slobodkin2023curious,snyder2024early}. Appendix~\ref{app:singleturn-whitebox} reviews this evidence, which suggests that internal states carry reliability information not expressed in model outputs.
Probes are useful for agents because they are inexpensive enough to run at every step and can directly provide $c_t=c(a_t,h_{t-1})$. Their use requires white-box access, which is unavailable for many deployed agents. Within the multi-source corpus described in Appendix~\ref{app:corpus}, we found no study of whether probes trained on factual QA transfer to predicting action correctness in long tool-use trajectories. This is a gap in the surveyed literature, not evidence that such transfer is impossible. Related robotics research has studied introspective failure prediction from a policy's internal signals. The closest LLM-agent examples in our corpus are introspective planning \citep{liang2024introspective} and the trajectory-level monitors reviewed in Section~\ref{sec:propagation}, but neither evaluates the transfer of a QA-trained probe. A taxonomy-independent audit likewise surfaces a neighboring preprint that uses calibrated hidden-state probes to abort likely-failing agent episodes early, but no work that tests whether a QA-trained probe transfers to this setting (Appendix~\ref{app:corpus}).

Bayesian approaches approximate the posterior predictive distribution $p(y \mid x, \mathcal{D})$ and measure its variation using the methods reviewed in Section~\ref{sec:sources}. Maintaining a posterior over billions of parameters is impractical, so LLM studies use proxies such as ensembles over low-rank adapters, variational adapter posteriors, and fine-tuned uncertainty heads (Section~\ref{sec:sources} and Appendix~\ref{app:singleturn-whitebox}).
These methods require several model evaluations or training runs. More importantly for agents, a posterior over model parameters does not capture uncertainty introduced by tools or the environment.

\subsection{Conformal Guarantees, Calibration, and Self-Knowledge}
\label{sec:conformal}
\label{sec:recalibration}
The methods discussed above produce uncertainty scores whose interpretation requires additional criteria. The remaining approaches provide coverage guarantees, recalibrate existing scores, or evaluate whether a model can identify the limits of its knowledge.
Conformal prediction returns a prediction \emph{set} rather than a score. On exchangeable data, the set constructed from a held-out calibration sample contains the true output with probability at least $1-\alpha$, a distribution-free guarantee \citep{vovk2005algorithmic,angelopoulos2023gentle}. Adaptations to language generation restrict the output space, calibrate when to stop sampling, or prune unsupported claims from long answers \citep{kumar2023conformal,quach2024conformal,mohri2024language}. Selective prediction instead answers only when confidence exceeds a threshold and is summarized by the risk--coverage curve introduced in Section~\ref{sec:metrics} and developed for agents in Section~\ref{sec:control} \citep{geifman2017selective,kamath2020selective}. Appendix~\ref{app:singleturn-guarantees} states the split conformal construction and reviews both lines in detail.
Conformal prediction provides coverage guarantees, while selective prediction can provide risk guarantees under suitable assumptions about calibration and test data. These assumptions are difficult to maintain for agent trajectories because steps are history-dependent, earlier decisions alter the distribution of later states, and step correctness may depend on the final task outcome. Existing extensions address covariate shift \citep{tibshirani2019conformal}, online adaptation of $\alpha$ \citep{gibbs2021adaptive}, and bounded departures from exchangeability \citep{barber2023conformal}. Standard conformal methods, however, do not directly cover fully history-dependent trajectories; Section~\ref{sec:challenges} returns to this issue.

Post-hoc calibration does not construct a new uncertainty score; it transforms an existing score.
Calibration is the most common property label in our corpus, appearing in $59$ of the $120$ papers. This frequency reflects that calibration can be evaluated or improved for any of the estimators discussed above. Section~\ref{sec:taxonomy} therefore treats calibration as a property rather than as a Panel-B family. That count is not the size of the recalibration family, which is a separate and disjoint label: a further $4$ papers fit an explicit post-hoc map, and it is those $4$ that this subsection reviews.
Standard post-hoc maps fitted on held-out data include Platt scaling, isotonic regression, and temperature scaling \citep{platt1999probabilistic,zadrozny2002transforming,guo2017calibration}; how much recalibration a model needs varies with architecture and evaluation setting \citep{minderer2021revisiting,desai2020calibration}, and Appendix~\ref{app:singleturn-guarantees} reviews the method families.
Two limitations are important for agents. Post-hoc calibration typically requires held-out data that represent the deployment distribution. In an agent trajectory, the distribution at later steps may differ from the distribution used to calibrate earlier steps. In addition, per-prediction recalibration does not model dependence between steps. Section~\ref{sec:formal} shows how this dependence can prevent step-level calibration from composing. Improving step-level calibration alone therefore does not guarantee trajectory-level calibration.

A related line of work examines whether models can identify the limits of their knowledge. ``Knowing what you don't know'' benchmarks categorize questions as answerable or unanswerable and evaluate whether the model responds appropriately \citep{yin2023large,kadavath2022language,lin2022truthfulqa}. These benchmarks evaluate binary decisions rather than probability calibration, so they complement rather than replace the calibration metrics in Section~\ref{sec:metrics}, and training methods can improve the underlying self-knowledge \citep{zhang2024rtuning,cheng2024ai}; Appendix~\ref{app:singleturn-guarantees} reviews both lines.
Abstention is the single-turn form of learning to defer \citep{madras2018predict,mozannar2020consistent}. For agents, deferral may occur at intermediate steps, where partial progress affects the decision.
\par
Overall, the seven single-turn rows in Table~\ref{tab:families} share a common assumption: one input, one output, a correctness label for that output, and roughly i.i.d.\ evaluation data. This setup does not directly extend to agents. The output is a trajectory $\tau$, correctness may depend on the whole trajectory rather than on one action $a_t$, and later states depend on the agent's earlier decisions.
The next section examines these differences in detail.

\takeaway[sec:foundations]{The single-turn techniques in Table~\ref{tab:families} mainly differ in the model access they require, their computational cost, and the uncertainty sources they can capture. Despite these differences, they largely assume a single input, a single output, a corresponding correctness label, and near-i.i.d.\ evaluation. These assumptions may not hold in agentic settings.}

\section{From Single-Turn LLMs to Agentic Uncertainty}
\label{sec:transition}

\subsection{The Agentic Stack}

Agent systems typically comprise interacting components for planning, retrieval, tool use, memory, and execution. Each component introduces distinct sources of uncertainty.
Planning introduces uncertainty because future actions must be evaluated before execution under incomplete knowledge of the environment. This requirement applies across tree search, graph-based reasoning, plan-then-execute prompting, learned world models, and physical-affordance approaches \citep{yao2023tree,besta2024graph,wang2023plan,zhou2023least,hao2023reasoning,ahn2022saycan}. Partial plans are often ranked using the model's self-evaluation signals \citep{xie2023self}, whose calibration is seldom assessed.
\par
Agents may interleave reasoning with calls to external tools \citep{yao2023react,schick2023toolformer} and select among large, heterogeneous tool sets \citep{qin2024toolllm,patil2023gorilla,shen2023hugginggpt}. Tool-use benchmarks and surveys examine these settings \citep{li2023apibank,zhuang2023toolqa,huang2024metatool,qin2024toollearning,mialon2023augmented}. Each call requires a decision about whether to invoke a tool, which tool to select, and which arguments to provide. Errors in these decisions may not be immediately apparent: a plausible but incorrect argument can still produce a seemingly valid result, which may then be treated as correct.
\par
Reflection loops allow agents to critique and revise their outputs \citep{shinn2023reflexion,madaan2023selfrefine}, but do not ensure reliable error correction. Without external feedback, models often fail to correct reasoning errors \citep{huang2024large}. They are more effective at repairing an error when its location is provided than at identifying the error independently \citep{tyen2024llms}, while step-level self-checking is beneficial only when predictions are aggregated appropriately \citep{miao2024selfcheck}. Reflection therefore does not eliminate uncertainty; it adds uncertainty about the reliability of the model's self-evaluation.
\par
Memory allows agents to store and reuse information across episodes \citep{park2023generative,packer2023memgpt,zhong2024memorybank,sumers2024cognitive,zhang2024surveymemory}, but can also carry uncertainty forward. Low-confidence information may later be retrieved without its original confidence information, while content that was accurate when stored may become outdated. Memory retrieval therefore inherits the uncertainty associated with external retrieval and adds risks associated with reusing agent-generated content.
\par
Multi-agent systems divide tasks and exchange information \citep{li2023camel,wu2023autogen,hong2024metagpt,qian2024chatdev,chen2024agentverse}, and may use debate or consensus to improve factuality \citep{du2024improving}. Correlation complicates the interpretation of agreement. When agents share a base model, prior knowledge, or retrieval pipeline, agreement provides limited evidence of correctness because the agents may reproduce the same error. Coordination can also introduce additional failure modes that remain insufficiently measured \citep{guo2024large}.
Uncertainty therefore arises throughout the agent stack: planners may misjudge feasibility, tool calls may fail without detection, memories may become outdated, and correlated agents may converge on an incorrect conclusion.
Current benchmarks focus primarily on end-to-end task success in website interaction, operating-system control, coding, and general assistance \citep{zhou2024webarena,deng2023mind2web,xie2024osworld,jimenez2024swebench,liu2024agentbench,mialon2023gaia,yao2024tau,koh2024visualwebarena}. Safety benchmarks also document cases in which agents select risky actions \citep{ruan2024toolemu,kinniment2024evaluating}.
With limited exceptions \citep{kirmayr2026carbench}, these benchmarks do not assess whether agents indicate uncertainty about decisions that are likely to fail.

\subsection{Four Mechanisms That Limit Single-Turn Methods}

The transition from single-turn models to agents changes both the sources and consequences of uncertainty.
We organize these differences into four mechanisms that can limit the transfer of single-turn uncertainty methods to agent settings. The discussion uses the notation from Section~\ref{sec:formal}.
\par
\emph{Errors can compound along trajectories} because each action is conditioned on the agent's preceding outputs, so an early error influences later decisions instead of remaining localized. Writing trajectory success as $\Pr(Y=1)=\prod_{t=1}^{T}\Pr(Y_t=1 \mid Y_{1:t-1}=\mathbf{1})$ and assuming independent steps with a shared success rate $p$ gives $p^T$; even at $p=0.99$ this falls to about $90\%$ after $10$ steps, $61\%$ after $50$, and $37\%$ after $100$, so a per-step confidence score alone does not determine trajectory reliability \citep{wei2026longhorizon}. Dependence complicates the picture further, since later steps condition on a history $h_{t-1}$ that contains self-generated outputs, resembling exposure bias in autoregressive generation \citep{ott2018analyzing} and the compounding-error analyses of imitation learning \citep{ross2010efficient,ross2011reduction}, with calibration additionally degrading under the induced distribution shift \citep{ovadia2019can}. The effect predates agents, since a model may defend an early incorrect claim in later sentences even when it can identify that claim as wrong in isolation \citep{zhang2024snowball}, and multi-step agents amplify it, with unsupported tool calls and hallucinations accumulating across interactions \citep{zhou2026exploring} and single errors irreversibly damaging the remainder \citep{piatrashyn2026redact}. This mechanism is what Proposition~\ref{prop:compose} formalizes: step-level calibration constrains $\Pr(Y_t=1\mid c_t)$, whereas trajectory success depends on conditional probabilities over the entire history, so an estimator calibrated at individual steps can still misestimate trajectory success, limiting the direct transfer of single-turn calibration methods \citep{zhang2026agentic}.
\par
\emph{Tools and environments introduce external uncertainty}. Each observation is generated by the environment, $o_t \sim P_{\mathrm{env}}(\cdot \mid h_{t-1}, a_t)$, and the same action may yield different observations, including incorrect ones. These uncertainties do not live in the model parameters. Retrieved documents may be irrelevant or contradictory \citep{chen2024benchmarking,yoran2024making}, tools may return noisy results, and environments may be partially observable or drift \citep{han2024towards,moshkovich2025taming}, so calibration or ensembling applied only to the model cannot address them. Conflicts between internal knowledge and retrieved evidence add a further problem \citep{xie2024adaptive,xu2024knowledge,shin2026era}, since a single confidence score may not show whether uncertainty comes from the answer or from disagreement between sources \citep{ren2023investigating}, and an incorrect retrieval accepted with high confidence can anchor later decisions \citep{julka2026when}. Uncertainty estimation for agents must therefore cover model predictions, tool outputs, retrieved information, and environmental observations together.
\par
\emph{Feedback may be delayed, sparse, or difficult to verify}. In single-turn question answering with LLMs, correctness can often be evaluated against a reference answer. In long-horizon agent tasks, feedback may arrive only at the end, and a reference trajectory may be unavailable \citep{han2026can,trantruong2026measuring}.
This complicates both training and evaluation. Metrics such as ECE and AUROC require correctness labels, but an unambiguous step-level label $Y_t$ may be unavailable during a trajectory because the value of an action can depend on later outcomes.
Process supervision faces the same problem \citep{uesato2022solving,lightman2024lets}. Outcome supervision uses only the final reward \citep{cobbe2021training}, while process supervision either requires dense step-level labels or estimates them through Monte Carlo rollouts \citep{wang2024mathshepherd}.
In our notation, such a rollout estimates $\hat{R}(h_t)=\Pr(Y=1\mid h_t)$, the probability that the current trajectory can still succeed. This is closely related to trajectory-level uncertainty estimation.
Verification is more difficult when the terminal outcome cannot be readily assessed or is observed only after further interaction \citep{itkin2026delayed}.
\par
\emph{Calibration does not determine control actions}. A well-calibrated risk score does not uniquely specify how an agent should respond. Two trajectory states may have the same estimated risk yet require different actions because recovery may be possible in one state but not the other \citep{zhang2026calibration}.
The \emph{intervention advantage} $A_i(h_t)$ in Eq.~\eqref{eq:advantage} represents this decision criterion by measuring the expected benefit of taking an oversight action $i$, such as pausing, rolling back, asking a human, or aborting, rather than continuing.
If intervention costs are not included in the task utility, intervention is warranted when its expected utility gain exceeds its cost. This criterion reduces to risk thresholding only when intervention value varies monotonically with risk, as in Chow's rule (Section~\ref{sec:control}).
For agents, this relation may not hold because the best decision also depends on factors such as recoverability, reversibility, and the remaining budget. A scalar risk score does not capture all of these factors.
Agent uncertainty estimates should therefore support oversight and control decisions, rather than serve only as measures of risk \citep{greenblatt2024ai,dixon2026adaptive}.
\par
\medskip \noindent These mechanisms motivate the stage-specific methods discussed in subsequent sections and inform the organization of Figure~\ref{fig:heatmap}.
Recent work addresses these mechanisms at different stages of the agent pipeline. Uncertainty propagation targets compounding errors; retrieval- and tool-level estimators address exogenous noise; process-based evaluation responds to unavailable labels; and abstention and intervention methods connect calibration with control.

\takeaway[sec:transition]{Agentic uncertainty differs from single-turn uncertainty in four respects: errors can compound along self-conditioned trajectories; tools and environments introduce exogenous uncertainty; step-level correctness labels may be unavailable; and calibrated risk does not uniquely determine an appropriate action. These sources arise throughout the agent stack, yet most capability benchmarks evaluate task success without assessing whether agents signal when failure is likely.}
\begin{figure}[t]
  \centering
  \includegraphics[width=\textwidth]{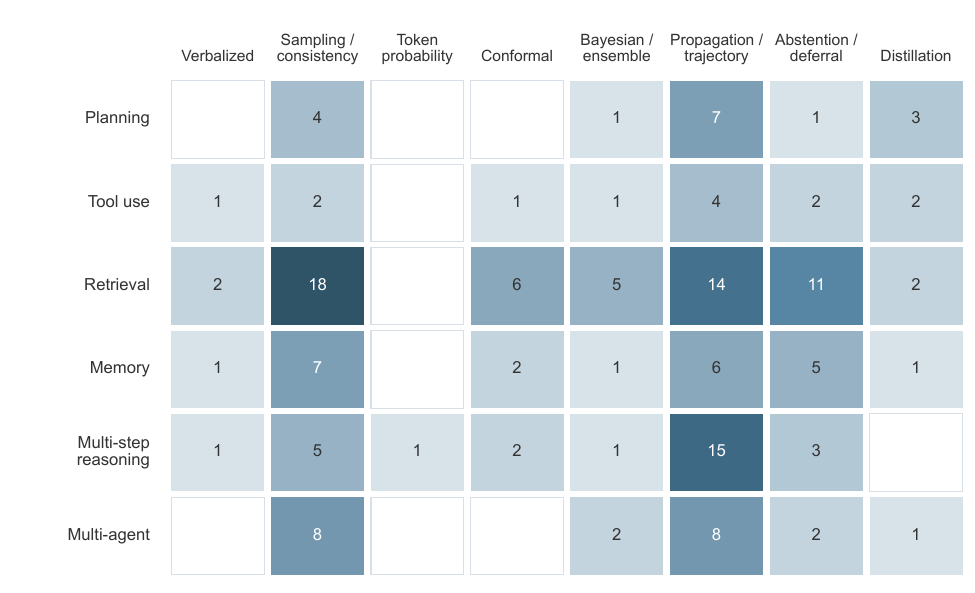}
  \caption{Distribution of corpus papers across pipeline stages (rows) and method families
  (columns). Because papers may receive multiple labels, counts are not mutually
  exclusive; an empty cell indicates that no matching paper was identified. The
  figure includes the eight indexed families B1--B8 and omits the residual
  category B9, which collects work that introduces no method.}
  \label{fig:heatmap}
\end{figure}

\section{Agentic Uncertainty, Stage by Stage}
\label{sec:pipeline}

This section organizes methods by the pipeline stage in which they operate: tool use and planning, retrieval and memory, and multi-agent coordination. It closes with three deployment surfaces that cut across these stages, namely multimodal perception, computer use, and code.

\subsection{Planning and Tool Use}
\label{sec:toolplanning}

\subsubsection{Gating and Calibrating Action Decisions}

Research on agent-specific uncertainty often begins with action selection. An agent decides whether an action is needed and which action to take.
\citet{han2024towards} introduce the Uncertainty-Aware Language Agent (UALA), which answers directly when its confidence $c_t$ in a direct answer clears a threshold selected on a small calibration set and enters the tool loop otherwise. The gate cuts unnecessary tool calls and keeps noisy external information out when the model already knows the answer.
Retrieval methods use a similar rule, fetching external information only when the model's knowledge appears insufficient, with trigger signals ranging from entity popularity \citep{mallen2023not} to low token probabilities \citep{jiang2023active} and current information needs \citep{su2024dragin}; other work directly predicts whether a tool is needed \citep{huang2024metatool}, and coding agents choose between asking for clarification and proceeding under an assumption \citep{edwards2026ask}.
\citet{zhou2026exploring} instead train the decision with uncertainty-aware reinforcement learning, rewarding tool calls under uncertainty and direct answers under confidence so that actions reflect the estimated uncertainty. Extensions calibrate tool-use behavior at training time, shape policies under sparse rewards, filter low-confidence function calls at inference, treat UQ as an agent-wide design problem, and use entropy-modulated policy gradients in long-horizon optimization \citep{chen2026et,bhatta2026uncertainty,broecker2026the,zhang2026uncertainty,wang2025harnessing}. \citet{moshkovich2025taming} monitor uncertainty from tools and environments during interaction with real systems, without access to the model's internal states.

A gating rule also requires a calibrated confidence estimate, which is difficult to obtain in tool-use settings.
\citet{xuan2026confidence} identify a \emph{confidence dichotomy}: an agent's confidence about whether to call a tool, $\Pr(\text{call warranted} \mid c_t^{\mathrm{inv}}=c)=c$, and its confidence in the final answer, $\Pr(Y_t=1 \mid c_t^{\mathrm{ans}}=c)=c$, can be miscalibrated in different directions, so one calibration rule cannot align both. The two decisions also carry different costs, a wasted or noise-injecting tool call against a hallucinated direct answer, and therefore need separate calibration and thresholds.
Recent methods estimate confidence in tool outputs, retrieve similar past cases, measure uncertainty over calls and their arguments, detect incorrect tool selection from internal representations before execution, verify that reported calls occurred, and return calibrated verifier signals to the agent \citep{xu2026when,pang2026case,ye2026uncertainty,healy2026internal,basu2026tool,vinod2026calvert}; related work addresses generated code \citep{shi2026code} and perception--tool coordination \citep{he2026confidence}.
Longer reasoning traces can reduce final-answer calibration \citep{lacombe2025dont}, and planning confidence can be \emph{epistemically} miscalibrated, assigning high feasibility to a plan despite lacking the information needed to execute it \citep{wang2026when}. Internal consistency alone therefore does not establish plan reliability. The self-evaluation scores used in plan search \citep{yao2023tree,xie2023self,hao2023reasoning} act as confidence signals, but their calibration is rarely evaluated. Planning-specific methods use uncertainty to drive information seeking or offline imagination, while conformal methods guide interactive information acquisition and reasoning \citep{chan2025conformal,hu2024uncertainty,hamidi2026dreamphase,frankel2024conformal}.
\subsubsection{Embodied Planning and Formal Guarantees}
Embodied planning has a longer history of uncertainty estimation and includes several methods with formal guarantees.
SayCan combines language-model scores with learned affordance values \citep{ahn2022saycan}. A candidate skill is preferred only when it is relevant to the instruction and likely to succeed in the current physical state. The rule combines a language-model score with an estimate of what the environment permits.
Inner Monologue extends this approach by incorporating environmental feedback into planning \citep{huang2022inner}.
KnowNo adds a conformal prediction layer to the planner \citep{ren2023robots}. Candidate actions are scored, and a calibration set is used to choose a nonconformity threshold $\hat q$. At test time, the robot forms
\begin{equation}
C(x_t)
=
\bigl\{
y \in \mathcal{Y}_t :
\hat f(y \mid x_t) \ge 1-\hat q
\bigr\},
\qquad
\Pr\!\left(y_t^{\star} \in C(x_t)\right)
\ge
1-\epsilon
\end{equation}
If $C(x_t)$ contains one action, the robot proceeds without assistance. If it contains several actions, the robot requests help. The set size therefore provides a statistical rule for help-seeking rather than an uncalibrated confidence threshold.
The guarantee is marginal and assumes exchangeability between calibration and test examples. Although exchangeability is less plausible in sequential settings, the method links uncertainty estimation to an explicit intervention rule.
Related work extends conformal methods to language-model decision interfaces and to planning in dynamic environments \citep{kumar2023conformal,lindemann2023safe}, separates ambiguity from difficulty as causes of planning uncertainty \citep{liang2024introspective}, revises plans against explicit assumptions \citep{seo2026from}, or combines step- and trajectory-level uncertainty in web agents \citep{zhang2026webuncertainty}.
Related robotics work estimates runtime failure and the value of assistance for learned policies, supplying neighboring evidence for online intervention in embodied agents \citep{romer2025failure,hagenow2025realm}.

\takeaway[sec:toolplanning]{Different decisions in tool use and planning require different uncertainty estimates. Tool use requires estimates for whether to call a tool and how to set its arguments. Planning requires estimates of plan feasibility and the need for assistance. These decisions have different failure costs and should therefore be calibrated separately.}

\subsection{Retrieval and Memory}
\label{sec:rag}

\subsubsection{Whether to Retrieve}

Retrieval augmentation is the largest stage in our corpus and has extensive overlap with the existing UQ literature \citep{lewis2020retrieval,guu2020realm,karpukhin2020dense,borgeaud2022improving,izacard2023atlas,ram2023incontext,shi2024replug,gao2023retrieval}.
Retrieval both reduces and introduces uncertainty. It reduces epistemic uncertainty by supplying evidence the parametric model lacks \citep{shuster2021retrieval}, but adds external uncertainty when the evidence is irrelevant, insufficient, or conflicting \citep{chen2024benchmarking,yoran2024making,xu2024knowledge}; it is useful when the reduction in residual uncertainty exceeds the uncertainty the retrieval process adds. Standard uncertainty methods do not always capture this distinction and may fail to raise uncertainty when the evidence conflicts with the generated answer \citep{soudani2025why}. \citet{shin2026era} separate epistemic uncertainty, which comes from missing knowledge, from knowledge conflict, which arises when sources disagree; \citet{binz2026uncertainty} extend the distinction to multimodal RAG, where text can contradict visual evidence.
One family of methods uses uncertainty to decide \emph{whether} retrieval is needed.
Adaptive retrieval consults external memory only when the model is likely to lack the required knowledge \citep{ren2023investigating}, with entity popularity as a simple trigger \citep{mallen2023not} because performance is lower on long-tail knowledge \citep{kandpal2023large}; \citet{jeong2024adaptiverag} instead route by predicted question complexity, a classifier trained from strategy success rather than calibrated against correctness.
FLARE triggers retrieval during generation when any token of the drafted next sentence falls below a probability threshold, then rebuilds the query from the low-confidence draft \citep{jiang2023active}. DRAGIN refines the trigger with the model's current information needs, combining token uncertainty with attention signals \citep{su2024dragin}. SeaKR reads the trigger from the model's internal states \citep{yao2025seakr}, and Self-RAG trains the model to emit \emph{reflection tokens} that indicate whether retrieval is needed, whether a passage is relevant, and whether the answer is supported \citep{asai2024selfrag}.
Related work makes the search process uncertainty-aware by suppressing low-quality queries, decomposing questions with confidence-guided search trees, and producing interpretable estimates \citep{wu2025search,jiao2026prunerag,dey2026interpretable}, following earlier work interleaving retrieval with chain-of-thought reasoning \citep{trivedi2023interleaving}.
\subsubsection{Trusting the Evidence}
Another family estimates the reliability of retrieved evidence.
Conformal methods filter the retrieved evidence with finite-sample guarantees. CONFLARE calibrates the retrieval similarity threshold on held-out question--evidence pairs, choosing the empirical $\alpha$-quantile with the usual finite-sample correction so that answer-supporting evidence is retained with probability at least $1-\alpha$. Later work extends the guarantee across retrieval and generation \citep{rouzrokh2024conflare,chakraborty2025principled,kotla2025conformal}, following conformal methods for generation \citep{quach2024conformal,mohri2024language}; the guarantee can weaken under realistic retrieval shift \citep{chen2026is}.
Several methods score evidence reliability before use, through confidence-based rerankers, span-level uncertainty for long contexts, corrective retrieval after unreliable evidence, and sufficiency checks for answer support \citep{song2026car,song2026llm,li2024uncertaintyrag,yan2024corrective,qiu2026surerag}; relevant signals may appear in hidden states even when absent from the output \citep{slobodkin2023curious}. Others localize the source of uncertainty, separating retrieval-induced from generation-induced errors \citep{ren2026when}, inferring what the model likely learned from pretraining-corpus statistics \citep{min2025qucorag,zhang2025measuring}, calibrating shared memory across agents, or propagating faithfulness-aware uncertainty through fact-checking pipelines \citep{meng2026equimem,fadeeva2025faithfulnessaware}. Newer evidence-side estimators differ in what they read the signal off. Three use the structure of the reasoning or the evidence: a semantic-level internal reasoning graph \citep{hu2026detecting}, the consistency of an evidence graph \citep{shen2026evidence}, and facet-level tracing that localizes the uncertainty to the specific claim it came from \citep{elchafei2026facet}. Three instead gate on a statistic of the generation: induction-aware entropy \citep{bazarova2026intrygue}, agreement across reformulated queries \citep{sun2026cqc}, and counterfactual risk rather than semantic relevance \citep{liu2026beyond}. A further method calibrates the hedging language of the answer itself against the retrieved evidence \citep{yeh2026retrieval}.
Both families must distinguish the source of the uncertainty they measure. When retrieved evidence conflicts with parametric knowledge, the model has two sources that disagree rather than missing information, and the responses differ, since missing knowledge calls for more retrieval while conflict calls for comparison or escalation. \citet{xie2024adaptive} show that models may accept coherent external evidence that contradicts their parametric beliefs while favoring evidence that agrees with them, and \citet{xu2024knowledge} separate parametric-contextual conflicts from conflicts within the retrieved context. Evidence position is a further nuisance variable. Accuracy drops when the supporting passage sits mid-context \citep{liu2024lost}, the order of conflicting documents can change the answer in biomedical RAG \citep{han2026whenevidence}, and reordering the same retrieved set can induce hallucinations \citep{zhang2026stable}. The distinction between epistemic and total predictive uncertainty is therefore important \citep{yadkori2024believe}, and reliability alignment brings it directly into the RAG setting \citep{shin2026era}.

\subsubsection{Escalation, Memory, and Infrastructure}
A routing approach uses uncertainty to decide how a query should be handled. \citet{jia2026balancerag} answer directly when the LLM-only risk estimate $\hat r_{\mathrm{LM}}(q)$ is low, retrieve when the retrieval branch's risk clears its own threshold, and abstain otherwise, with the two thresholds calibrated jointly because the queries sent to retrieval are exactly the difficult cases the direct branch rejected. Related work routes with energy-based abstention scores, past experience, trained self-knowledge, and explicit trust-or-abstain decisions \citep{shankar2025energy,sarkar2026leveraging,stoisser2025towards,zhu2026trust}, using uncertainty to select an action rather than only report a score \citep{zhang2026calibration} and following earlier cost-aware cascades \citep{chen2023frugalgpt}. Other work determines how much evidence to retrieve, trains agents to recognize when parametric knowledge suffices, and studies how early retrieval confidence can lock the pipeline onto an incorrect path \citep{dong2026know,feng2026kbsd,julka2026when}.

Routing determines how a query is processed; memory determines which evidence is available at that point, and it introduces uncertainty of its own \citep{park2023generative,packer2023memgpt,zhong2024memorybank,zhang2024surveymemory}. Two problems recur. Agent-generated memory can store a low-confidence guess and later present it as a fact unless confidence is stored with the entry, and entries can become stale while remaining relevant. Shared memory amplifies both because every reader can receive the same incorrect entry, so \citet{meng2026equimem} calibrate shared memory before using it as evidence and \citet{essam2026trustaware} attach confidence to shared knowledge graphs. Memory can also serve estimation, as past cases become an abstention signal \citep{sarkar2026leveraging} and an explicit belief state replaces an ever-growing raw history \citep{singh2026agentbrace}.
Most work nevertheless treats memory as one component of a larger system; its large stage count is mostly co-tagging with retrieval (Appendix~\ref{app:taxnotes}), and methods designed specifically for memory uncertainty, deciding when to write or evict an entry and how much to trust self-generated content, remain scarce. Evaluation is also moving beyond answer correctness, with faithfulness-oriented RAG benchmarks \citep{es2024ragas,saadfalcon2024ares,niu2024ragtruth} complemented by resources that measure how retrieval changes uncertainty and where hallucinations arise (Section~\ref{sec:eval}).

\takeaway[sec:rag]{Retrieval can reduce epistemic uncertainty by providing missing evidence, but the retrieval process can also introduce uncertainty. A reliable estimator should distinguish confidence in the answer from agreement between the model and the retrieved evidence. Memory presents a similar challenge because it relies on a self-generated corpus. Stored information should therefore be assessed before reuse.}

\subsection{Coordination Across Multiple Agents}
\label{sec:multiagent}

\subsubsection{Confidence-Weighted Debate}

Collaboration makes uncertainty dependent across agents. Multi-agent debate aims to improve factuality through argument and has been studied under asymmetric persuasion \citep{du2024improving,liang2024encouraging,irving2018ai,khan2024debating}; related systems use role-based collaboration \citep{li2023camel,wu2023autogen,hong2024metagpt,chen2024agentverse,qian2024chatdev} or LLM judges \citep{zheng2023judging}, whose reliability on agent trajectories is now being audited \citep{kc2026babeljudge}. Debate is not uniformly beneficial; its value depends on the number of rounds, speaking order, and whether agents observe one another's reasoning or only conclusions \citep{smit2024should}. We index agents by $j=1,\dots,m$, with vote $v_j$ and confidence $c_j$.
Confidence can modify the protocol itself. ReConcile selects the confidence-weighted plurality $\hat y = \arg\max_{y} \sum_j w(c_j)\,\mathbf{1}[v_j = y]$, where $w(\cdot)$ recalibrates verbalized confidence across heterogeneous scales, and routes the reasoning of confident peers to dissenters \citep{chen2024reconcile}. DebUnc attaches the sender's uncertainty to each message, in text or through receiver attention, so unreliable messages are discounted on receipt rather than only at the final vote \citep{yoffe2024debunc}. Confidence gating decides which arguments are admitted \citep{baba2026argument}; ablations identify stated-confidence calibration and panel diversity as the two operative factors \citep{zhu2026demystifying}; response variance can delay premature agreement \citep{tang2026the}; and confidence-aware routing sends each query to the cheapest agent whose confidence clears a threshold \citep{wang2026orchestrating}. Other work estimates black-box uncertainty through multiple agents, calibrates confidence through deliberation, and supplies trust or opinion formalisms \citep{feng2025rethinking,yang2024confidence,cheng2021general,cheng2020hope}.
\subsubsection{False Consensus and Stopping Decisions}
Agreement alone does not establish reliability. The limitation has a precise statistical form.
Under conditional independence given the truth, posterior log-odds are additive across agents:
\begin{equation}
\log\frac{\Pr(Y=1 \mid v_1,\dots,v_m)}{\Pr(Y=0 \mid v_1,\dots,v_m)} \;=\; \log\frac{\Pr(Y=1)}{\Pr(Y=0)} \;+\; \sum_{j=1}^{m} \log\frac{\Pr(v_j \mid Y=1)}{\Pr(v_j \mid Y=0)}
\end{equation}
Each vote then contributes a separate increment of evidence.
Conditional independence alone does not make the vote share $\frac{1}{m}\sum_j \mathbf{1}[v_j=y]$ a sufficient statistic. The sum depends only on the vote count when the per-agent likelihood ratios are also identical, as under conditional exchangeability or equal accuracy. Otherwise, aggregation must retain agent identity and reliability. A confident vote from an accurate agent is not equivalent to one from a less accurate agent.
When both conditions hold, a $5$--$4$ split among nine independent experts provides evidence whose strength depends on their shared accuracy.
Communication generally removes this independence.
\citet{huang2026counterfactual} show that the same vote share can represent different amounts of evidence. It may arise from $m$ effectively independent assessments or from a \emph{false consensus} in which one early, confident, incorrect message propagates through the group and produces many dependent votes.
Formally, the posterior $\Pr(Y=1 \mid v_1,\dots,v_m)$ depends on the joint dependence structure of the votes, not on their marginal counts: at the extremes, $m$ perfectly correlated votes carry the evidential weight of one.
Their counterfactual method reconstructs the communication graph and intervenes on its edges. It tests whether agent $j$ would have voted differently without observing neighboring messages. Votes that change under this intervention are treated as derivative and discounted, producing an effective sample size for the consensus.
Related behavior also appears outside multi-agent systems. \citet{sharma2024towards} show that preference-trained assistants may revise correct answers after a user expresses doubt or asserts a conflicting answer, because human raters reward agreement. An agent receiving a confident peer message faces a similar signal. A protocol that circulates confidence scores may therefore reinforce this tendency. False consensus can arise from both statistical dependence and learned conformity. Graph-based and confidence-based discounting address different parts of this problem.
The diagnosis reframes multi-agent UQ as a problem about the \emph{joint} distribution of agent outputs.
\citet{chen2026every} analyze the same dependence with tensor decomposition. They organize responses by agent, round, and sample. Low-rank components represent the shared signal that persists through communication, while the residual represents uncertainty from communication dynamics and role dependence. A single-turn method applied to one agent cannot capture this structure.
Recent work formalizes this failure mechanism. \emph{Collective hallucination}, in which a fabrication becomes a shared group belief, has been formalized together with defenses \citep{jamshidi2026collective}. Companion work traces how one agent's error propagates through neighboring agents \citep{jamshidi2026hallucination}. These models provide explicit mechanisms for inter-agent error correlation.

Dependence affects the evidential value of consensus. Uncertainty also informs when deliberation should stop.
\citet{morandi2026sequential} apply sequential hypothesis testing to debate. Each round provides evidence for the hypothesis that the current consensus answer is correct. The method accumulates a log-likelihood ratio and uses a Wald-style stopping rule:
\begin{equation}
\text{continue while }\; \log\frac{\beta}{1-\alpha} \;<\; \Lambda_n \;<\; \log\frac{1-\beta}{\alpha}, \qquad \Lambda_n=\sum_{i=1}^{n}\log\frac{p_1(x_i)}{p_0(x_i)}
\end{equation}
Crossing the upper threshold accepts the consensus, while crossing the lower threshold triggers escalation or abstention. Here, $\alpha$ and $\beta$ are the tolerated false-acceptance and false-rejection rates.
This rule ends deliberation early for clear cases and allocates more rounds to contested cases. Cost-aware model cascades use similar anytime logic \citep{chang2026cascadedebate}.
Wald's test belongs to a broader framework for anytime-valid inference, which is directly relevant to sequential agent monitoring.
Repeated use of a fixed-sample test inflates its error rate. A $p$-value is calibrated for one decision at a specified sample size, whereas an agent may inspect uncertainty after every step. Anytime-valid inference instead uses an \emph{e-value}, a non-negative statistic $E$ with $\mathbb{E}_{H_0}[E]\le 1$. Markov's inequality then controls the type-I error at threshold $1/\alpha$. The sequential form is an \emph{e-process}, a process $(E_n)$ bounded above by a non-negative supermartingale under the null. Ville's inequality gives $\Pr_{H_0}(\exists n:\, E_n \ge 1/\alpha)\le\alpha$ uniformly over all stopping times \citep{ramdas2023game,grunwald2024safe}. This guarantee remains valid under continuous monitoring and data-dependent stopping. E-values can be multiplied across evidence and merged by averaging, so their accumulation along a trajectory is well defined. The corresponding estimation object is a confidence sequence, a family $(C_n)$ satisfying $\Pr(\forall n:\, \theta\in C_n)\ge 1-\alpha$ \citep{howard2021timeuniform}. It is the time-uniform analogue of the per-decision marginal coverage in Section~\ref{sec:conformal}, for which a union bound over a horizon gives $T\alpha$. Calibration against an adaptively generated sequence rather than an i.i.d.\ sample is also a classical problem \citep{foster1998asymptotic}. This framework is relevant when a monitor's interventions affect the data it later evaluates.
Two limitations prevent direct application. Communication makes evidence adaptively dependent, so debate rounds and agents cannot be treated as independent observations. A valid e-process must be defined with respect to the joint interaction history rather than marginal outputs. The null hypothesis must also be specified. For the claim that the current consensus is correct, this requires a model of how correct and incorrect groups generate votes, which is itself a target of multi-agent UQ. Constructing e-processes for LLM deliberation and for trajectories in which agent actions generate the evidence remains an open problem. This problem may benefit from methods in neighboring statistical literature (Appendix~\ref{app:corpus}).
Recent methods provide finer-grained diagnosis, identifying \emph{confidently wrong} debaters from log-probability signals \citep{keramati2026confident}, predicting debate quality from early-token confidence \citep{keramati2026early}, and standardizing confidence semantics across agents before aggregation \citep{armstrong2026margin}. Distribution-free methods reach the protocol level, filtering inter-agent messages, adding coverage guarantees to final aggregation, and connecting budgeted act-or-defer rules to per-agent reliability bounds \citep{zhang2026commcp,wang2026from,wang2026budgeted}.
Protocol design itself can reduce uncertainty in consensus. Selective consensus weights sources by estimated reliability \citep{li2026trusttrade}; specialist panels with consistency verification improve calibration in medical QA because specialists err more independently than clones of one generalist \citep{martinez2026multiagent}; confidence tracked through a sequential software pipeline flags low-confidence upstream decisions before they propagate downstream, which is propagation across agents rather than steps in the sense of Section~\ref{sec:propagation} \citep{essam2026trustaware,ogunsusi2026uachatdev}. Hallucination control is also an explicit design objective, through reinforcement-trained self-check loops, game-theoretic protocols that make truthful reporting the stable strategy, and specialist pipelines for citation-level detection \citep{li2026march,liu2026game,li2026source}.

\takeaway[sec:multiagent]{In multi-agent systems, the identity of an uncertain agent and the
dependence between agents both affect reliability. After communication, votes
are no longer independent evidence. An aggregation rule that uses only the
marginal $(v_j, c_j)$ pairs discards the dependence information needed to
assess consensus.}

\subsection{Perception-Heavy Surfaces: Multimodal, Computer-Use, and Code Agents}
\label{sec:surfaces}

The previous subsections organize methods by the stage of the agent loop in which they operate. Three deployment surfaces cut across those stages and concentrate a failure mode that a stage-based reading scatters. On these surfaces, uncertainty enters through \emph{perception} or through the artifact being produced, before any of the reasoning steps that the estimators above examine. An agent may misread a screen, ground the wrong object, or misjudge a program it has just written. Once such an error enters $h_t$, a well-calibrated text-side estimator can reason confidently from an incorrect premise.
\par
For multimodal reasoning, the emerging pattern is to give perception and reasoning separate confidence estimates rather than one pooled score. \citet{xiao2026vlcalibration} decouple confidence calibration for the perception and reasoning stages of large vision-language models, and \citet{he2025mmboundary} calibrate reasoning-step confidence to make a multimodal model aware of its knowledge boundary. \citet{madhusudhan2026knowing} benchmark whether multimodal reasoning systems abstain when they should. \citet{binz2026uncertainty} extend the knowledge-conflict distinction of Section~\ref{sec:rag} to multimodal retrieval, where retrieved text can contradict visual evidence, and \citet{he2026confidence} use confidence to orchestrate tools for video understanding. The embodied instances of this surface appear elsewhere in this paper. \citet{darabi2026groundcontrol} read failure from the temporal shape of an uncertainty trace in vision-language navigation (Section~\ref{sec:propagation}), and \citet{yeke2026yesman} benchmark abstention under physical uncertainty (Section~\ref{sec:control}).
\par
Computer-use agents make the perception problem concrete at scale, because every action is grounded in a screenshot. \citet{kumar2026uncertainty} benchmark uncertainty quantification for these agents across vision-language models and GUI-grounding datasets, separating the model stack from the grounding data. The operative distinction on this surface is between \emph{grounding} uncertainty, which concerns whether the agent has identified the right element of the interface, and \emph{decision} uncertainty, which concerns whether the chosen action is correct given that grounding. A benchmark that labels only final task success cannot separate the two, and the estimators of Sections~\ref{sec:toolplanning} and~\ref{sec:propagation} address only the second.
\par
Code agents differ from both in one respect that estimation can exploit. A candidate program has executable semantics, so part of a step's correctness can be checked by running it against tests or type checkers rather than estimated from the model alone. \citet{shi2026code} build uncertainty estimation for code generation on this observation. \citet{edwards2026ask} use uncertainty to decide whether a coding agent should ask a clarifying question or proceed under an assumption when instructions are underspecified. At the level of software pipelines, confidence has been attached to shared artifacts and propagated through development workflows \citep{essam2026trustaware,ogunsusi2026uachatdev}, and evidence-calibrated multi-agent auditing has been applied to repository-level vulnerability detection \citep{meng2026vulnagentr}. Relative to how widely code agents are deployed, dedicated uncertainty work on this surface remains sparse in our corpus.

\takeaway[sec:surfaces]{On perception-heavy surfaces, a grounding error enters the history
before any reasoning step, so a calibrated text-side estimator can be
confidently wrong about a premise it never questioned. Estimators should
separate grounding uncertainty from decision uncertainty, and code agents
additionally expose executable signals that make step correctness partially
checkable. Coverage of these surfaces remains thin relative to how commonly
they are deployed; Problem~10 in Appendix~\ref{app:challenges} states the
open problem.}

\section{Agentic Uncertainty Across the Whole Trajectory}
\label{sec:families}

This section examines three method families designed specifically for agents rather than adapted from single-turn UQ. Because they span all pipeline stages, we organize them by how they operate: propagation and trajectory-level estimation, abstention and control, and uncertainty-aware training.

\subsection{Propagation and Trajectory-Level Estimation}
\label{sec:propagation}

When uncertainty accumulates across steps, step-level confidence must be combined into a trajectory-level estimate.
This problem is prominent in the reviewed corpus.
The main quantity is trajectory reliability, $R=\Pr(Y=1)$, which equals $\mathbb{E}\bigl[\prod_{t=1}^{T}Y_t\bigr]$ under the absorbing-failure convention of Section~\ref{sec:formal}, where $Y_t\in\{0,1\}$ indicates whether step $t$ succeeds. During execution, it can be estimated as $\hat R(h_t)$ conditional on the current history $h_t=(s_0,a_1,o_1,\dots,a_t,o_t)$.
A simple method multiplies the confidence scores of individual steps under an independence assumption,
\begin{equation}
\hat R_{\mathrm{ind}}(\tau)
=
\prod_{t=1}^{T} c_t
\end{equation}
This estimate can fail for several reasons.
Small calibration errors accumulate across a trajectory. If each $c_t$ slightly overestimates the probability of step success, the product exceeds the product of the step marginals, and the difference increases with the horizon $T$.
Step outcomes are also dependent. An early observation error can affect every later decision, so the product can be inaccurate even when each step-level confidence is well calibrated.
These failure modes have opposite effects and should be distinguished. Overconfident step estimates inflate the product. Positive prefix coupling, which agents are expected to exhibit, makes the product of step marginals \emph{conservative} (Proposition~\ref{prop:positive}) and therefore lowers the same estimate. Their relative effects must be measured empirically, and Section~\ref{sec:empirical} measures both terms. On the short-horizon chained-QA traces, step-level overconfidence dominates and the product overestimates reliability, while on the long-horizon ALFWorld traces the dependence term itself becomes substantial, so neither term can be neglected a priori. The signed decomposition in Appendix~\ref{app:proofs-pathwise} formalizes the two contributions.
Methods in this section address these limitations by improving the step-level uncertainty signal $u_t$, modeling dependencies across steps, or representing the dynamics of the full trajectory.
Multi-turn confidence estimation makes the history dependence of these signals a primary object of study \citep{zhang2026confidence}.
Table~\ref{tab:propagation} organizes existing methods along these three directions, which also structure the discussion below.

\begin{table}[t]
\centering
\caption{Trajectory-level estimators in family B6, compared by their input
signal, combination rule, and output. Each provides an alternative to the
simple product of step-level confidence scores.}
\label{tab:propagation}
\small
\setlength{\tabcolsep}{4.5pt}
\renewcommand{\arraystretch}{1.15}
\begin{tabular}{@{}L{3.5cm}L{3.4cm}L{4.0cm}L{3.6cm}@{}}
\toprule
\thc{Method} & \thc{Signal} & \thc{Combination} & \thc{Output}\\
\midrule
SAUP \citep{zhao2024saup} & Per-step uncertainty $u_t$ & Learned situational weights $w_t = g(s_t, h_{t-1})$ & Weighted trajectory risk score\\
\addlinespace[2pt]
UProp \citep{duan2025uprop} & Shifts in action distributions across sampled paths & Mutual-information chain rule across steps & Epistemic uncertainty of the decision sequence\\
\addlinespace[2pt]
HTC \citep{zhang2026agentic} & Process-level features across the whole trajectory, macro dynamics to micro stability & Simple interpretable calibration model & Calibrated trajectory confidence with failure diagnostics\\
\addlinespace[2pt]
Bayesian network \citep{donaldson2026bayesian} & Per-stage divergence and self-evaluation scores & Conditional dependencies among latent stage-correctness variables & Failure posterior with node-level attribution\\
\addlinespace[2pt]
Trace denoising \citep{yan2026denoiseflow} & Noisy per-step uncertainty trace & Separation of transient fluctuations from persistent drift & Cleaned reliability trace for downstream aggregation\\
\addlinespace[2pt]
Belief--action divergence \citep{singh2026agentbrace} & Separate belief and action confidences & Divergence measured across the horizon & Decoupling alarm used as a risk signal\\
\addlinespace[2pt]
Temporal signatures \citep{darabi2026groundcontrol} & Trajectory-consistent uncertainty trace & Shape of the trace rather than its level & Online prediction of terminal failure\\
\addlinespace[2pt]
Markov absorption \citep{trantruong2026measuring} & Qualitative transitions between execution states & Absorbing-chain algebra $(I-Q)^{-1}\mathbf{b}$ & Dynamical reliability and expected time to failure\\
\addlinespace[2pt]
Pipeline conformal \citep{kotte2026pasc} & Per-stage nonconformity scores & Jointly calibrated thresholds across stages & Prediction sets with joint coverage\\
\addlinespace[2pt]
Systems-level \citep{zhang2026managing} & Pipeline health signals beyond the model & Propagation across data, agents, and humans & Operational risk of the deployed system\\
\bottomrule
\end{tabular}
\end{table}

\subsubsection{Dependence-Aware Aggregation and Failure Attribution}
Aggregation methods combine uncertainty signals from multiple steps into one trajectory-level estimate.
SAUP (Situation Awareness Uncertainty Propagation) is an early example \citep{zhao2024saup}. It computes $\hat U(\tau) = \sum_t w_t u_t$ using learned situational weights $w_t = g(s_t, h_{t-1})$. A step receives greater influence when it is more critical in the current state and history. This weighting reflects the fact that the same uncertainty can have different consequences at different points in a trajectory.
\par
UProp provides an information-theoretic formulation of multi-step uncertainty \citep{duan2025uprop}. Unlike SAUP's weighted heuristic signals, UProp targets epistemic uncertainty over the full decision sequence and decomposes it with the mutual-information chain rule. Let $\theta$ denote the latent factors about which the model is uncertain. Then $I(\theta;\, a_{1:T}\mid s_0) = \sum_t I(\theta;\, a_t \mid h_{t-1})$, so each decision contributes a conditional term. This term can be estimated at the corresponding step by sampling alternative decision paths and comparing their action distributions.
The decomposition separates reducible epistemic uncertainty from uncertainty inherent in the task \citep{hullermeier2021aleatoric}. Like SAUP, it also identifies the important decision points instead of averaging uncertainty over the entire trajectory.
\par
\citet{donaldson2026bayesian} use a Bayesian network to model uncertainty propagation in agentic RAG.
Observable signals from the planner, evaluator, and generator include semantic disagreement across samples and self-evaluation scores. These signals are linked to latent variables that represent the correctness of each stage, while the graph models how errors propagate to later stages. Inference returns a trajectory-level failure estimate and a stage-level posterior $\Pr(\text{stage } k \text{ faulty} \mid \text{evidence})$. The latter supports debugging and intervention by identifying the stage responsible for the estimated failure.
Similar attribution methods assign confidence to individual steps in black-box reasoning chains and trace hallucinated intermediate results to their source in multi-agent workflows \citep{liu2026diagnosing,badave2026beyond}.

\subsubsection{Per-Step Signals, Online Monitoring, and Trajectory Dynamics}
Aggregation depends on the quality of its input signals. Several studies therefore treat the per-step signal as a separate estimation problem.
Raw step-level uncertainty estimates are noisy. Sampled self-evaluations fluctuate, and token-level entropy can spike because of harmless formatting choices. \citet{yan2026denoiseflow} denoise the per-step trace before aggregation by separating transient fluctuations from persistent changes in reliability.
\citet{singh2026agentbrace} identify a different long-horizon failure: an agent's stated beliefs can become decoupled from its chosen actions. Confidence measured from the belief channel then no longer represents the action taken. They use the divergence between separately estimated belief and action uncertainties as a risk signal. \citet{yi2026measuring} find that the agent harness can induce the same divergence, showing that the scaffold influences the belief trace as well as the model.
During training, \citet{pan2026tiar} reweight advantage estimates across the trajectory so that reinforcement learning assigns credit to timely abstention rather than treating it as a failure. This incorporates propagation-aware uncertainty into the policy.
Other step-level detectors infer the onset of hallucination from hidden-state transport geometry, score traces using the ratio of hedging to verification behavior, or calibrate step confidence against temporal-logic constraints on the reasoning chain \citep{alvarez2026where,pandey2026selfdoubt,mao2026confidence}.
Another line of work conditions the estimate on the agent's position in the trajectory. This turns $\hat R(h_t)$ into an online monitor instead of a post-hoc score.
\citet{darabi2026groundcontrol} estimate trajectory-consistent uncertainty in vision-language navigation. The uncertainty trace anticipates failure dynamics such as oscillation between states, stagnation, and increasing detours from the goal before they lead to terminal failure. The warning comes from the temporal shape of the trace rather than its value at one step.
This result is consistent with evidence that internal signals can reveal hallucinations before they appear in the output \citep{snyder2024early}. Partial supporting evidence can also \emph{non-monotonically amplify} confident hallucinations, causing the confidence trace to rise while content quality declines \citep{lathkar2026anchored}.
\citet{lu2025auditing} test whether step-level confidence reports correspond to step-level competence. Stepwise self-evaluation also produces traces whose temporal patterns, including monotone decline, sudden collapse, and oscillation, provide useful information \citep{mavi2025selfevaluating,xie2024calibrating}.
This approach builds on process supervision, where rewarding correct \emph{steps} outperforms rewarding correct \emph{outcomes} \citep{uesato2022solving,lightman2024lets}, and on step-level reasoning evaluation \citep{golovneva2023roscoe,prasad2023receval}. Both estimate a trajectory-level property by evaluating the process that produces it.

\citet{trantruong2026measuring} explicitly model the agent as a Markov chain over qualitative execution states. Success and failure are absorbing states, and trajectory reliability is the probability of eventual absorption into success.
Let $Q$ denote the transient-to-transient block of the transition matrix. Given transient-to-absorbing transitions, the absorption probabilities follow from the fundamental matrix,
\begin{equation}
\mathbf{r} \;=\; (I - Q)^{-1}\, \mathbf{b}
\end{equation}
where $b_i$ is the one-step probability of moving from transient state $i$ to success, and $r_i$ is the reliability of an agent currently in state $i$.
This formulation makes reliability a property of the agent's \emph{dynamics} rather than of a single sampled trajectory, which allows estimation from modest amounts of data. The same matrix algebra also gives the expected time to failure and the sensitivity of $R$ to individual transition probabilities.
Guarantees pose a related problem. Per-step conformal methods can certify each decision at level $1-\alpha$ \citep{ren2023robots,lindemann2023safe}. Without a model of dependence, the union bound is the strongest generic composition, and trajectory-level risk grows to $T\alpha$. This provides another reason to model trajectories directly instead of composing stepwise certificates.
Pipeline-aware conformal prediction addresses this problem by constructing prediction sets with \emph{joint} coverage across all stages of a multi-stage pipeline \citep{kotte2026pasc}. Information-theoretic analysis also bounds the performance of closed-system multi-step reasoning, even when each step is well calibrated \citep{shin2026the}.
At the system level, \citet{zhang2026managing} model uncertainty as propagating through agent coordination, data pipelines, and human-in-the-loop stages as well as through the model.
Their study of safety-critical deployment shows that improving model accuracy alone does not remove system risk. Stale data feeds, ambiguous human handoffs, and coordination latency contributed more than model error in several failure post-mortems.
This systems perspective treats uncertainty estimation as a property of the full deployed pipeline. It connects agent UQ to the established requirement that uncertainty estimates remain useful under dataset shift \citep{ovadia2019can}. An estimator validated on an isolated model may provide little information about a pipeline in which that model is only one component.

\takeaway[sec:propagation]{Marginal step confidences do not capture the error-correlation
structure across steps. Online estimation should target the history-resolved
hazard $\lambda_t(h_{t-1})$ in Section~\ref{sec:formal}, rather than the
marginal confidence $c_t$. Two opposing effects distort the simple product:
compounded step-level overconfidence inflates it, whereas positive step
dependence makes it conservative, and Section~\ref{sec:empirical} measures
both effects. Propagation methods address the limitation by weighting critical
steps, propagating conditional information, denoising per-step traces,
modeling trajectory dynamics, or attributing failures to pipeline stages.}

\subsection{Abstention, Clarification, Deferral, and Control}
\label{sec:control}

For an agent, an uncertainty estimate is useful when it guides an action such as abstaining, requesting clarification, deferring to a human, or replanning. \citet{zhang2026from} describe this as a broader shift from uncertainty reported as a post-hoc metric to uncertainty used as an in-loop control signal. This section examines that shift.
The formal foundations are the reject option and learning to defer to an expert \citep{chow1970optimum,elyaniv2010foundations,madras2018predict,mozannar2020consistent}. Agent-level control extends these formulations to sequential decisions.
In Chow's formulation, a classifier either predicts a label or abstains at cost $\lambda \in [0, \tfrac{1}{2})$ under a $0$--$1$ loss. The Bayes-optimal rejector applies the posterior threshold
\begin{equation}
r^\star(x) \;=\; \mathbb{1}\!\left[\,\max_{y}\,\Pr(y \mid x) \;<\; 1-\lambda\,\right]
\end{equation}
The system predicts when its maximum posterior probability reaches a threshold determined by the cost structure rather than by the model. \citet{elyaniv2010foundations} express this trade-off through the risk--coverage frontier introduced in Section~\ref{sec:metrics}, which should be evaluated over its full range rather than at a single operating point. \citet{geifman2017selective} and \citet{kamath2020selective} extend this perspective to deep networks and question answering.
A direct agent analogue replaces $\max_y \Pr(y \mid x)$ with the trajectory reliability estimate $\hat{R}(h_t)$: continue while $\hat{R}(h_t) \geq 1-\lambda$ and abstain otherwise.
Sequential agent control differs because the available alternatives, their costs, and recoverability can change over time.

\subsubsection{Sequential Deferral and Clarification}

In agent settings, the alternative to acting is often another decision-maker rather than abstention at a fixed cost.
Learning to defer \citep{madras2018predict} incorporates the performance of the downstream expert into the objective.
Let $m(x)$ denote the expert's prediction on $x$ and $c_{\mathrm{ask}}$ the query cost. The predictor and rejector are trained jointly by minimizing
\begin{equation}
\min_{f,\,r}\;\; \mathbb{E}\Big[\,(1-r(x))\,\ell\big(f(x),y\big) \;+\; r(x)\,\Big(\ell\big(m(x),y\big) + c_{\mathrm{ask}}\Big)\Big]
\end{equation}
Unlike Chow's rule, the optimal rejector does not defer wherever model uncertainty is highest. It defers when the model's expected loss exceeds the expert's expected loss plus the query cost. The system therefore answers cases for which the expert offers no expected advantage and defers cases that the expert is expected to handle better.
\citet{mozannar2020consistent} show that common confidence-based surrogate losses for this objective are statistically inconsistent: minimizing them may fail to recover the Bayes-optimal predictor--rejector pair even with unlimited data. They obtain a consistent surrogate by augmenting the label space with a deferral class. Thus, plausible threshold heuristics can be suboptimal even in a one-shot setting, and their optimality over trajectories requires additional justification.
\citet{piatrashyn2026redact} extend deferral to sequential decision making, where irreversible early errors increase the value of timely deferral.
In a one-shot problem, the cost of an incorrect answer or unnecessary abstention is incurred once. In a sequential problem, an irreversible action at step $t$ can eliminate feasible continuations. The utility recoverable through deferral, $\mathbb{E}[U \mid \mathrm{defer}, h_t]$, is therefore greatest before commitment and may decline as side effects accumulate. Sequential deferral can be viewed as an optimal-stopping problem: waiting may provide evidence about whether assistance is needed, but it also risks an irreversible action. Beyond calibration, an important property of an uncertainty estimate is its \emph{earliness}: whether it triggers while recovery remains possible.
The quality of this trigger depends on confidence--correctness alignment, which can be optimized directly \citep{xiaohu2026know}. Related applications include clinical models trained to recommend or defer in underrepresented specialties \citep{rajesh2026teaching}, and medical QA evaluations that test whether models decline to answer when they should \citep{machcha2025know}.

Clarification differs from deferral because the agent retains control and seeks information that can resolve aleatoric uncertainty.
Relevant single-turn work includes ambiguity benchmarks, selective answering \citep{min2020ambigqa,cole2023selectively,rajpurkar2018know}, and clarifying questions triggered by ambiguity estimates \citep{kuhn2023clam,zhang2024clarify,zhang2024clamber}.
In agent settings, \citet{edwards2026ask} study the ask-or-assume decision in coding agents and use uncertainty to trigger clarification when instructions are underspecified.
\citet{deng2026uncertaintyaware} formalize this decision by selecting clarifying questions that maximize expected information gain about latent user intent.
Let $\theta$ denote the intent, $q$ a candidate question, and $A_q$ its random answer. The selected question is
\begin{equation}
q^\star \;=\; \arg\max_{q}\; \mathbb{E}_{a \sim p(a \mid q,\, h_t)}
\Big[\, H\big(\theta \mid h_t\big) \;-\; H\big(\theta \mid h_t,\, q,\, a\big) \Big]
\;=\; \arg\max_{q}\; I\big(\theta;\, A_q \mid h_t\big)
\end{equation}
The agent asks a question only when the maximum expected information gain exceeds the interaction cost.
This criterion applies Bayesian experimental design to dialogue \citep{lindley1956measure,chaloner1995bayesian} and instantiates a value-of-information rule with the user as the information source \citep{howard1966information}. The preferred question need not concern the aspect with the highest raw uncertainty; it is the question whose answer produces the largest reduction in decision-relevant uncertainty (Appendix~\ref{app:pomdp}).
The aleatoric--epistemic distinction is operationally relevant in this setting. \citet{matsnev2026uncertainty} query the user when ambiguity rather than ignorance dominates. Retrieval cannot resolve genuine underspecification of intent, whereas asking users for facts they may not know consumes the interaction budget without resolving epistemic uncertainty.
\citet{suri2025structured} model clarification as a managed sequential dialogue rather than a single gated question. The decision can also be represented as a three-action policy that answers, asks, or abstains, trained through fine-tuning or reinforcement learning using contextual evidence \citep{baidya2026passiveqa,zhao2026grace}.

\subsubsection{Embodied Abstention, Capability Gating, and Intervention Advantage}

In embodied settings, ask-or-act decisions can have physical consequences.
\citet{yeke2026yesman} benchmark abstention in robotic agents and identify a \emph{yes-man} pattern in which agents proceed despite perceptual and physical uncertainty. This work complements the ask-for-help mechanism in KnowNo \citep{ren2023robots}, which uses conformal prediction to request assistance when the set of plausible plans is too large for unilateral action.
This behavior need not result from classical miscalibration. A model may produce a relevant uncertainty signal while the policy layer fails to use it. The resulting failure illustrates the recurring distinction between uncertainty estimation and control.
Recent work examines when complete agents should stop rather than act, when multimodal reasoning systems should decline, and how risk estimates can gate capabilities before execution \citep{luo2026agentic,madhusudhan2026knowing,iyer2026capability}.
Capability gating differs from step-wise abstention because it restricts the \emph{action space} before execution. It grants or withholds classes of irreversible operations according to risk tiers rather than evaluating each proposed action. The approaches are complementary: a coarse but reliable tier may limit worst-case harm when fine-grained confidence estimates are inaccurate.
\citet{zhang2026calibration} argue that the common pipeline of predicting risk, applying a threshold, and then intervening is mis-specified. They instead propose estimating \emph{intervention advantage}, the quantity $A_i(h_t)$ defined in Eq.~\eqref{eq:advantage}.
Abstention, clarification, deferral, stopping, and gating can each be represented as an intervention $i$. Under a rational control rule, intervention $i$ is selected when $A_i(h_t)$ exceeds its cost.
Chow's rule is recovered as a special case in which abstention is the only available intervention and has a fixed payoff. In this setting, intervention advantage is a monotone function of risk, making thresholding of $1-\hat{R}(h_t)$ optimal.
This monotonicity need not hold for agents. Two histories may have the same failure risk even though one is recoverable through a low-cost clarifying question and the other has progressed beyond the point at which intervention is effective.
The relation between trajectory risk and intervention benefit remains unresolved.
This distinction suggests that agent uncertainty estimates may need to represent intervention-specific advantages rather than only the scalar $\hat{R}(h_t)$.

\takeaway[sec:control]{Abstention, clarification, deferral, stopping, and gating can all be
represented as interventions with advantage $A_i(h_t)$. Risk thresholding is
optimal only when intervention advantage varies monotonically with risk, a
condition that may fail along agent trajectories. Sequential deferral also
has an optimal-stopping structure, so uncertainty estimates must be timely
enough to support intervention while recovery remains possible.}

\subsection{Uncertainty-Aware Training and Distillation}
\label{sec:training}

Most methods discussed above estimate uncertainty during inference.
Other work trains agents to improve calibration or uses uncertainty as a learning signal.
Related single-turn methods include calibration-aware tuning \citep{lin2022teaching,band2024linguistic,xu2024sayself}, linguistic calibration of dialogue models \citep{mielke2022reducing}, refusal training \citep{zhang2024rtuning,cheng2024ai,yang2024alignment}, and uncertainty-calibrated heads \citep{kapoor2024large}.
Agent methods can adapt the structure used in refusal training. R-tuning \citep{zhang2024rtuning} partitions supervised data according to the model's own competence. A pair $(x,y)$ is assigned to $\mathcal{D}_{\mathrm{sure}}$ when the pretrained model answers $x$ correctly and to $\mathcal{D}_{\mathrm{unsure}}$ otherwise. Fine-tuning then uses targets augmented with the corresponding certainty expression and maximizes
\begin{equation}
\sum_{\mathcal{D}_{\mathrm{sure}}} \log p_{\theta}(y \oplus \texttt{[sure]} \mid x)
+
\sum_{\mathcal{D}_{\mathrm{unsure}}} \log p_{\theta}(y \oplus \texttt{[unsure]} \mid x)
\end{equation}
This construction defines honesty relative to the model's own knowledge boundary rather than an external measure of difficulty, a property relevant to agent abstention policies.
\citet{band2024linguistic} optimize the decisions made by a listener after reading hedged text rather than the likelihood of confidence tokens. The objective seeks calibrated decision utility for readers who use the expressed confidence. This perspective treats verbalized confidence \citep{tian2023just,kadavath2022language} as communication with decision consequences and parallels Section~\ref{sec:multiagent}, where one agent's stated confidence becomes another agent's input.

\subsubsection{Uncertainty in Reward Shaping and Process Supervision}

Preference-based fine-tuning \citep{christiano2017deep,stiennon2020learning,bai2022training} contributes to the utility of assistant models but can reduce calibration \citep{leng2024taming,openai2023gpt4}. This motivates uncertainty-aware reward modeling \citep{pan2026uncertaintyaware}. A reward model trained on finite preference data may produce an overconfident scalar outside its training distribution, while policy optimization favors regions in which errors in that scalar are exploitable. One response represents reward as a distribution and optimizes a pessimistic functional such as $\tilde{r}(x,a) = \mu(x,a) - \beta\, \sigma(x,a)$. The policy is then rewarded only for gains supported by sufficiently confident reward estimates.
For agents, \citet{zhang2026selaur} use intrinsic LLM uncertainty as a reward signal for a self-evolving agent and treat it as a source of step-level credit assignment. \citet{zhou2026exploring} similarly incorporate step-level uncertainty into a reinforcement-learning objective for tool use.
Both approaches can be expressed through a common objective:
\begin{equation}
J(\pi_{\theta}) \;=\; \mathbb{E}_{\tau \sim \pi_{\theta}}
\left[\; \sum_{t} r_t \;+\; \beta\, \psi\big(u_t\big) \;\right]
\end{equation}
where $u_t$ is a step-level uncertainty signal computed from $h_t$ and the proposed action. The sign and form of $\psi$ determine how this signal affects learning. It may provide an exploration bonus when uncertainty identifies informative experience, as in self-evolving agents \citep{zhang2026selaur}, or penalize confident errors when the objective is to align step confidence $c_t$ with tool-call success \citep{zhou2026exploring}.
\citet{pan2026tiar} train abstention through trajectory-informed advantage reweighting. Rather than adding a reward term, their method rescales the advantage at each step using trajectory-level outcome information. Abstention receives positive credit on trajectories that would fail if continued and negative credit when it prevents an otherwise successful trajectory.
This implements the intervention-advantage perspective of Section~\ref{sec:control} at training time: the policy gradient targets $A_{\mathrm{abstain}}(h_t)$ rather than raw risk.
Related work reports improved calibration when reflection is credited appropriately, develops denser credit assignment over long horizons without relying on uncalibrated intermediate rewards, and restricts training to trajectory segments with reliable outcome attribution \citep{zhu2026closing,li2026when,qi2026stapo}.
These agent objectives are closely related to process supervision \citep{lightman2024lets,uesato2022solving,wang2024mathshepherd} and outcome-supervised verifiers \citep{cobbe2021training}, although these connections are not always made explicit.
Math-Shepherd's automated step labels \citep{wang2024mathshepherd} are obtained by rolling out continuations from each intermediate state and recording their empirical success rate. These labels are Monte Carlo estimates of $\Pr(Y = 1 \mid h_t)$. A process reward model trained on them can therefore be interpreted as a learned reliability estimator $\hat{R}(h_t)$ under the terminology of this paper.
Reflection-based self-improvement \citep{shinn2023reflexion,madaan2023selfrefine} also relies implicitly on such a signal. Models cannot reliably identify their own errors \citep{huang2024large}, and this known limitation can be viewed as an uncertainty-estimation problem. A reflection step is useful only when the model's implicit estimate of where an error occurred is sufficiently reliable.

\subsubsection{Distillation and the Calibration-Transfer Gap}

Distillation introduces a related calibration question. Agent capabilities are often transferred to smaller, deployable student models \citep{hinton2015distilling,xu2024surveykd,sharma2025small}, including through on-policy methods that train students on their own sampled trajectories \citep{agarwal2024onpolicy,gu2024minillm}, and a student need not inherit the teacher's calibration. Appendix~\ref{app:proofs-distill} formalizes the \emph{calibration-transfer gap}, the worst-case miscalibration of a teacher-calibrated reliability estimator when it is scored on histories generated by the student. The gap is zero for the teacher by assumption but otherwise unconstrained. It can grow because calibration is defined relative to a distribution and residual imitation errors compound with the horizon, so teacher and student encounter different history distributions, and because lower student capability makes the teacher-fitted estimator systematically optimistic on states the teacher would handle but the student cannot. On-policy distillation trains on the student's state distribution and can mitigate exposure bias \citep{agarwal2024onpolicy,gu2024minillm}, but whether it also shrinks this gap is untested in the reviewed literature.
Few corpus papers examine the intersection (Figure~\ref{fig:heatmap}, distillation column). Emerging single-turn work transfers the teacher's confidence about future tokens, gates which teacher signals the student may imitate, and restricts imitation to steps whose validity the teacher can certify \citep{kale2026future,sermsri2026gatekd,saadi2026validity}, but does not yet address trajectory-level transfer. Calibration-preserving distillation for agents is the narrowest gap our search identified. We found no corpus paper that evaluates end-to-end trajectory-success calibration on histories generated by the student policy; the taxonomy-independent audit in Appendix~\ref{app:corpus} finds neighboring work that distills calibrated turn-level beliefs, but no evaluation of this narrower endpoint. Section~\ref{sec:challenges} states this as a research problem.

\takeaway[sec:training]{Training can integrate uncertainty into the policy rather than
produce only a post-hoc report. Refusal tuning defines honesty relative to
the model's knowledge boundary, listener-calibrated objectives treat
confidence as consequential communication, and reward shaping can credit
timely abstention. These objectives may also create incentives to manipulate
the uncertainty signal without resolving the underlying errors. Whether
calibration transfers to smaller student agents remains largely untested.}

\section{Evaluation: Benchmarks and Metrics}
\label{sec:eval}

Reliable evaluation of agent uncertainty requires benchmarks that measure uncertainty quality as well as task success. Current resources provide only limited support for this goal.
Table~\ref{tab:benchmarks} in Appendix~\ref{app:expdetails-tables} inventories the evaluation landscape.
This section reviews benchmarks (Section~\ref{sec:eval-benchmarks}) and metrics (Section~\ref{sec:eval-metrics}). It then applies the proposed metric to real agent traces (Section~\ref{sec:empirical}) and identifies the remaining evaluation gaps (Section~\ref{sec:eval-missing}).

\subsection{Benchmarks}
\label{sec:eval-benchmarks}

The agent-specific block of Table~\ref{tab:benchmarks} lists early benchmarks that evaluate more than task correctness, and each addresses one part of agent uncertainty. URAG measures how retrieval changes uncertainty but its short trajectories do not test long-horizon control \citep{nguyen2026urag}; MIRAGE-Bench attributes hallucinations to pipeline components without testing whether a confidence score predicts failure in advance \citep{zhang2025miragebench}; \citet{trantruong2026measuring} model trajectory reliability with a Markov chain whose validity depends on the chosen states; \citet{han2026can} evaluate delayed feedback within one financial domain. Others evaluate decisions rather than numeric scores, testing whether embodied agents abstain when appropriate \citep{yeke2026yesman} or whether tool-use agents recognize their limits \citep{kirmayr2026carbench}; \citet{mavi2025selfevaluating} come closest to step-level calibration but require step labels; \citet{kumar2026uncertainty} extend the evaluation to computer-use agents. More recent resources follow the same pattern, attributing hallucinations, constructing controlled world models, evaluating guardrails along tool-use trajectories, or comparing black-box estimators under a shared single-turn protocol \citep{liu2026agenthallu,liu2026halluworld,chen2026tracesafe,wang2026a}.
We found no resource in the corpus that releases cross-domain agent \emph{trajectories} under a common protocol for comparing estimators $\hat{R}(h_t)$, although trajectories with final outcomes would already support the checkpoint metrics of Section~\ref{sec:eval-metrics} and step-resolved labels would additionally support step calibration and attribution. The gap agrees with broader surveys of agent evaluation, which report that reliability receives far less attention than capability \citep{mohammadi2025evaluation}.

\subsection{Metrics}
\label{sec:eval-metrics}

Most existing studies apply single-turn uncertainty metrics to the final answer. These include expected calibration error, AUROC for failure detection, risk--coverage curves, and proper scoring rules such as the Brier score. Section~\ref{sec:metrics} defines these metrics, and Table~\ref{tab:metrics} summarizes them.
These metrics do not directly evaluate uncertainty across a trajectory.
We therefore introduce a simple trajectory-level metric.
The metric illustrates the evaluation target and is not an established standard.

\begin{definition}[Trajectory-checkpoint expected calibration error]
\label{def:tcece}
Let $\{\tau^{(j)}\}_{j=1}^{N}$ be trajectories with outcomes
$Y^{(j)} \in \{0,1\}$. For each trajectory, choose a set of checkpoints
$\mathcal{T}_j$ using a rule that is independent of the estimator under
evaluation. Collect the prediction--outcome pairs
$\big(\hat{R}(h^{(j)}_t),\, Y^{(j)}\big)$ for all $j$ and
$t \in \mathcal{T}_j$. Partition $[0,1]$ into bins $B_1,\dots,B_M$, and define
$\mathcal{I}_m = \{(j,t) : \hat{R}(h^{(j)}_t) \in B_m\}$ and
$n_{\mathrm{ck}} = \sum_m |\mathcal{I}_m|$. Here, $n_{\mathrm{ck}}$ is the
number of checkpoint--outcome \emph{pairs}, not the prediction count $n$ used
by the single-turn ECE of Section~\ref{sec:metrics}. The bins $B_m$ partition
confidence values as they do there. The trajectory-checkpoint expected calibration error of
$\hat{R}$ is
\begin{equation}
\mathrm{TC\text{-}ECE}(\hat{R}) \;=
\sum_{m \,:\, |\mathcal{I}_m| > 0} \frac{|\mathcal{I}_m|}{n_{\mathrm{ck}}}\,
\Bigg|\;
\frac{1}{|\mathcal{I}_m|}\sum_{(j,t) \in \mathcal{I}_m} Y^{(j)}
\;-\;
\frac{1}{|\mathcal{I}_m|}\sum_{(j,t) \in \mathcal{I}_m} \hat{R}\big(h^{(j)}_t\big)
\;\Bigg|
\end{equation}
\end{definition}

Definition~\ref{def:tcece} is a finite-sample, bin-dependent plug-in estimator of a population checkpoint-calibration target. We use \emph{TC-ECE} to denote this empirical statistic. Like standard ECE, it involves a trade-off between discretization error and finite-sample bias. Its magnitude and model ranking may change with the number of bins and with the use of equal-width or equal-mass bins \citep{kumar2019verified,roelofs2022mitigating}. TC-ECE is therefore neither an unbiased nor a binning-invariant population quantity.
The distinction between a metric and an estimator also locates the closest related work. \citet{zhang2026agentic} propose Holistic Trajectory Calibration (HTC), an \emph{estimator} that extracts process-level features across an agent's entire trajectory, from macro dynamics to micro stability, and maps them through a simple interpretable model to a calibrated trajectory confidence with failure diagnostics. They evaluate it on eight benchmarks across multiple LLMs and agent frameworks, reporting both calibration and discrimination, including out-of-domain transfer. TC-ECE is not a competitor to HTC but the yardstick by which such an estimator should be scored. Definition~\ref{def:tcece} accepts an arbitrary $\hat{R}(h_t)$, HTC included, and the protocol around it specifies how the resulting score should be read, against a baseline restricted to causally available features, stratified by trajectory position with a calibrated null, with signed gaps, and with refitting inside the bootstrap. HTC's evaluation establishes that a well-designed trajectory estimator can be calibrated and discriminative across benchmarks. The experiments in Section~\ref{sec:empirical} establish something different, namely that these protocol choices change the conclusion an evaluation reaches, which is a claim about the evaluation rather than about any estimator. The two contributions are therefore complementary, and what this paper proposes is the evaluation object and reporting protocol rather than a new estimator.
\paragraph{TC-ECE pools checkpoints and is strictly weaker than Definition~\ref{def:traj}.}
Definition~\ref{def:traj} requires calibration separately at every checkpoint index $t$. Definition~\ref{def:tcece} instead pools all pairs $\big(\hat{R}(h^{(j)}_t), Y^{(j)}\big)$ over $(j,t)$ before binning. Within one bin, overconfidence at $t=1$ can therefore offset underconfidence at $t=4$. A predictor may have a small pooled TC-ECE while violating Definition~\ref{def:traj} at every horizon. TC-ECE does not operationalize Definition~\ref{def:traj}; it is only a descriptive aggregate across checkpoints.
Stratified metrics correspond more closely to Definition~\ref{def:traj}. Partition checkpoints into groups $g$ that a live monitor can compute, such as the step index $t$ or coarser step bands. Let $\mathrm{ECE}_g$ denote Definition~\ref{def:tcece} applied only to group $g$, and report
\begin{equation}
\overline{\mathrm{ECE}} = \frac{1}{|\mathcal{G}|}\sum_{g \in \mathcal{G}} \mathrm{ECE}_{g},
\qquad
\mathrm{ECE}_{\max} = \max_{g \in \mathcal{G}} \mathrm{ECE}_{g}
\label{eq:strat-ece}
\end{equation}
together with the group sizes.
$\mathrm{ECE}_{\max}$ is closest to Definition~\ref{def:traj}, which requires calibration at every $t$ and fails when any checkpoint index is miscalibrated, and it dominates both the pooled value and $\overline{\mathrm{ECE}}$. The unweighted $\overline{\mathrm{ECE}}$ carries no such guarantee and should be read together with the group sizes. Because late strata contain few trajectories (Section~\ref{sec:formal}), we report $\mathrm{ECE}_{g}$ only for groups above a stated minimum size, state that size explicitly, and state the horizon at which the stratified statistics stop; the excluded thin strata are exactly the ones Definition~\ref{def:traj} is most likely to flag, so their sample counts should be reported rather than silently dropped. Neither statistic is equivalent to Definition~\ref{def:traj}; they are the closest statistics supported by the sample.
Two finite-sample effects complicate the stratified magnitudes, and Appendix~\ref{app:expdetails-metrics} analyzes both. Binned ECE has an upward bias that grows as group size shrinks \citep{kumar2019verified,roelofs2022mitigating}, so $\mathrm{ECE}_{\max}$ can exceed the pooled value even for a predictor with no position-dependent miscalibration, and the raw ratio of the two is not meaningful by itself. The required reference is a \emph{calibrated null}: simulate a predictor calibrated at every checkpoint index under the same checkpoint structure, stratification, and binning, and record the resulting distribution of $\mathrm{ECE}_{\max}/\mathrm{ECE}_{\mathrm{pooled}}$; the observed ratio supports horizon-dependent miscalibration only when it exceeds this null, which in Section~\ref{sec:empirical} explains most of the apparent effect in the raw ratios.
The null does not establish per-horizon calibration either, because pooling can cancel \emph{signed} errors that absolute-gap statistics never see. The quantity that detects the cancellation is the signed gap, mean stated reliability minus observed success within a stratum, which needs far smaller $n$ than a binned ECE and shows the drift's direction; a stratified report should include it alongside Eq.~\eqref{eq:strat-ece}, and Section~\ref{sec:empirical} finds a large drift in it that the magnitude statistics miss. When data suffice, a multicalibration-style criterion conditioning jointly on confidence and position is stronger \citep{hebertjohnson2018multicalibration,detommaso2024multicalibration}; we report Eq.~\eqref{eq:strat-ece} because joint conditioning splits an already small sample twice (Appendix~\ref{app:expdetails-metrics}). In practice, reports should carry pooled TC-ECE as the main aggregate, $\overline{\mathrm{ECE}}$ and $\mathrm{ECE}_{\max}$ against cancellation, and the calibrated-null ratio distribution to separate finite-sample bias from predictor behavior.
The sum runs over non-empty bins, and the weights $|\mathcal{I}_m| / n_{\mathrm{ck}}$ make TC-ECE \emph{checkpoint-weighted}, so long trajectories carry more influence; a trajectory-weighted variant gives each trajectory total weight $1/N$, and we report both because one scores the average checkpoint and the other the average trajectory.
\par
Trajectory-level calibration also differs from single-turn calibration in several ways.
Checkpoints from the same trajectory share the final outcome $Y^{(j)}$ and are not independent. Confidence intervals should therefore resample complete trajectories rather than individual checkpoints.
Calibration should also be reported at multiple trajectory positions. A fixed step $t$ or a fixed fraction of a known step budget gives a calibration profile that can be computed during execution. The normalized position $t/T$ is available only after execution because $T$ is unknown while the trajectory is in progress. Either profile prevents good calibration at one stage from masking poor calibration at another.
The calibration gap should be reported with a proper scoring rule such as the checkpoint Brier score,
\begin{equation}
\frac{1}{n_{\mathrm{ck}}}
\sum_{(j,t)}
\left(
\hat{R}(h_t^{(j)})-Y^{(j)}
\right)^2
\end{equation}
which also measures the quality of the predicted probabilities \citep{gneiting2007strictly}. A robust evaluation should state the binning rule and test sensitivity across plausible bin counts and equal-width or equal-mass schemes. When the sample size permits, it may also report a debiased, adaptive, kernel-based, or consistent $L_p$ calibration estimator \citep{nixon2019measuring,kull2019beyond,kumar2019verified,widmann2019calibration,roelofs2022mitigating,popordanoska2022consistent}.
If the uncertainty estimator controls stopping or intervention, it also selects the observed checkpoints. Evaluation should then disable the control policy or correct for this selection.
\par
Current evaluations repeatedly face confounding and target mismatch.
Single-turn metrics can be strongly affected by nuisance variables such as output length \citep{santilli2025revisiting}. If confidence $c$ and correctness $Y$ both depend on a variable $L$, then $\mathrm{AUROC}(c,Y)$ may be high even when $c$ provides no additional information conditional on $L$, that is, when $c \perp Y \mid L$.
This confounding can be stronger for agents because trajectory length varies widely. Many propagated uncertainty scores depend on the horizon $T$ \citep{zhao2024saup,duan2025uprop}, while task success often decreases with trajectory length. A method may therefore appear predictive only because it captures length.
Evaluation should report metrics separately by horizon and compare the estimator with a baseline that uses only horizon information. This baseline must use features available at the checkpoint, such as the elapsed step count and the step budget. Realized trajectory length $T$ is future information. When $T$ depends on the outcome, it also leaks the label and creates an artificially strong baseline (Section~\ref{sec:empirical}).
Step-level and final-answer metrics also do not measure trajectory-level calibration. An agent can be calibrated at individual steps but remain overconfident about full-task completion because step errors are correlated.
No trajectory-level reliability metric has yet been adopted as widely as ECE in single-turn evaluation. Definition~\ref{def:tcece} provides a starting point rather than a final standard.

\subsection{Empirical Illustration: TC-ECE on Real Agent Traces}
\label{sec:empirical}

We include three experiments to show that the proposed metrics can be computed on real agent traces. The first two use one model on short-horizon tasks, a ReAct agent on HotpotQA and chains of QA items. The third runs the same elicitation protocol with a second open-weight model on ALFWorld under a $50$-step budget, where a real action--observation loop and exact per-step labels are available. The HotpotQA protocol is additionally replayed with a frontier API model on the same questions, and the ALFWorld protocol with a third open-weight model. They illustrate the evaluation protocol and are not a benchmark or method comparison. None of the learned propagation estimators of Table~\ref{tab:propagation} is included, because SAUP, UProp, and the Bayesian-network monitor require multi-path sampling, trained situational weights, or stage-level instrumentation that these single-pass traces do not provide; the prefix product is retained as the naive baseline shared by that family (Section~\ref{sec:propagation}).
HTC \citep{zhang2026agentic} is the omission that needs its own reason, since it fits a simple model to process-level features and would run on traces of this kind. We exclude it because it is a trained estimator whose published feature set and fitting procedure are specified for its own eight benchmarks, so reimplementing it here would compare our reconstruction rather than the method, and because doing so would turn a protocol demonstration into the estimator comparison this section explicitly does not attempt. The protocol is estimator-agnostic by construction: Definition~\ref{def:tcece} accepts any $\hat{R}(h_t)$, and scoring HTC or any other trajectory estimator under it, on released traces, is the benchmark that Section~\ref{sec:challenges} states as Problem~1 rather than a gap this section fills. All results apply only to these configurations and do not support general claims about agents or verbalized confidence. The supplementary material contains the code, traces, and analysis.

\par

\paragraph{Setup.}

We use a ReAct agent \citep{yao2023react} with two actions, $\texttt{Search[title]}$ and $\texttt{Finish[answer]}$. Using Qwen2.5-32B-Instruct with greedy decoding and a maximum of six steps, the agent answers $300$ HotpotQA questions in the distractor setting \citep{yang2018hotpotqa}.
Two properties limit the scope of this configuration. The model is an open-weight checkpoint released in 2024. Because verbalized calibration depends on model scale and the alignment procedure \citep{tian2023just,leng2024taming}, the results should not be treated as evidence about current frontier agents; the frontier replication later in this section addresses this directly. The realized horizons are also short and concentrated: $201$ of the $300$ trajectories contain exactly three steps. The step index therefore covers a narrow range, and the horizon-stratified analysis below is based on that limited variation.
At each step, the agent reports its confidence in the current action and its estimated probability that the final answer will be correct. We use the latter as the verbalized trajectory reliability estimate $\hat{R}(h_t)$ evaluated by Definition~\ref{def:tcece}.
Each step serves as a checkpoint, making checkpoint selection independent of the uncertainty estimator. This produces $954$ checkpoints from $300$ trajectories, with an average of $3.2$ steps per trajectory and $100\%$ format compliance.
The trajectory success rate under containment-tolerant exact match is $0.657$. Twelve trajectories never produce \texttt{Finish} and are counted as failures. These choices may introduce label noise and should be considered when interpreting the results.
We evaluate five estimators. They comprise the zero-shot verbalized estimate $\hat{R}(h_t)$; the same signal after isotonic recalibration with trajectory-level cross-fitting; the running product of step confidences, corresponding to the naive propagation method in Section~\ref{sec:propagation}; and two horizon baselines, each implemented as a logistic regression with trajectory-level cross-fitting.
The horizon baselines differ in their available features. The \emph{causal} baseline uses only information available to a live monitor at step $t$: the elapsed step index and the fixed step budget, represented by $(t,\,t/B)$ with $B=6$. Because $B$ is constant across trajectories, $t/B$ is a rescaled version of $t$ and adds no information. The baseline is therefore effectively univariate in the step index. We retain the two-feature representation because a monitor with a task-specific budget would compute both values; it should be interpreted as \emph{step index alone}. The \emph{post-hoc} baseline also uses the realized trajectory length through the features $(t/T,\,T)$.
The post-hoc baseline cannot be deployed because $T$ is unavailable while a trajectory is running. In this experiment, $T$ is also strongly associated with the outcome: the success rate is $0.771$ among the $201$ trajectories that finish in three steps and $0.133$ among the $15$ trajectories that reach the six-step limit, many of which never emit \texttt{Finish}. Conditioning on $T$ therefore leaks outcome information. We include this baseline only to measure the apparent performance gained from this unavailable feature. All conclusions use the causal baseline.
We follow the protocol in Section~\ref{sec:eval-metrics}. We use ten equal-width bins, bootstrap complete trajectories over $2{,}000$ resamples, report the checkpoint Brier score together with the calibration gap, and stratify both by normalized progress and by step index.
Three estimators require fitting, so all intervals reported below use a \emph{nested} bootstrap. Each replicate resamples the $300$ trajectories and refits the isotonic map and both logistic baselines under the same trajectory-level cross-fitting, because resampling already cross-fitted predictions omits fitting variation and understates uncertainty for the fitted estimators. Table~\ref{tab:tcece} also reports the naive intervals to show the effect of refitting.
Because all five estimators are evaluated on the same trajectories, pairwise comparisons use paired differences from joint bootstrap samples rather than differences between marginal intervals. Each reported paired difference is a mean over nested replicates with refitting, and it therefore need not equal the difference between the point estimates in Table~\ref{tab:tcece}, which come from models fitted once on the complete sample. Appendix~\ref{app:expdetails} specifies the resampling and refitting procedure, explains this difference with a worked example, and states the scope and mild conservatism of the resulting intervals.

\begin{table}[t]
\centering
\caption{TC-ECE on $954$ checkpoints from $300$ ReAct trajectories on
HotpotQA. TC-ECE weights checkpoints, TC-ECE$_{\mathrm{tr}}$ trajectories;
$\overline{\mathrm{ECE}}$ and $\mathrm{ECE}_{\max}$ stratify by step index
over strata with at least $30$ checkpoints, Eq.~\eqref{eq:strat-ece}.
Fitted rows ($^{\dagger}$) show naive and nested bootstrap intervals. The
last column is $\mathrm{ECE}_{\max}/\mathrm{TC\text{-}ECE}$ with its
$p$-value under the calibrated null discussed in the text.}
\label{tab:tcece}
\footnotesize
\setlength{\tabcolsep}{2pt}
\renewcommand{\arraystretch}{1.15}
\begin{tabular}{@{}lccccccccc@{}}
\toprule
\thc{Estimator $\hat{R}(h_t)$} & \thc{TC-ECE} & \thc{CI (naive)} & \thc{CI (nested)} & \thc{TC-ECE$_{\mathrm{tr}}$} & \thc{$\overline{\mathrm{ECE}}$} & \thc{$\mathrm{ECE}_{\max}$} & \thc{Brier} & \thc{AUROC} & \thc{Ratio ($p$)}\\
\midrule
Verbalized success prob. & 0.125 & [.102,.165] & [.102,.165] & 0.147 & 0.196 & 0.289 & 0.251 & 0.591 & 2.3 (.45)\\
\quad + isotonic recalibration$^{\dagger}$ & 0.028 & [.018,.076] & [.017,.133] & 0.030 & 0.125 & 0.260 & 0.229 & 0.588 & 9.3 (.001)\\
Prefix product of step conf. & 0.084 & [.060,.132] & [.059,.130] & 0.090 & 0.124 & 0.167 & 0.226 & 0.630 & 2.0 (.60)\\
\tband{10}{Horizon baselines}
Causal $(t,\,t/B)$, deployable$^{\dagger}$ & 0.046 & [.038,.083] & [.025,.127] & 0.056 & 0.082 & 0.153 & 0.228 & 0.562 & 3.3 (.15)\\
Post-hoc $(t/T,\,T)$, sees $T^{\dagger}$ & 0.097 & [.066,.154] & [.061,.168] & 0.109 & 0.096 & 0.106 & 0.204 & 0.645 & 1.1 (.99)\\
\bottomrule
\end{tabular}
\end{table}

\begin{figure}[t]
  \centering
  \includegraphics[width=0.98\textwidth]{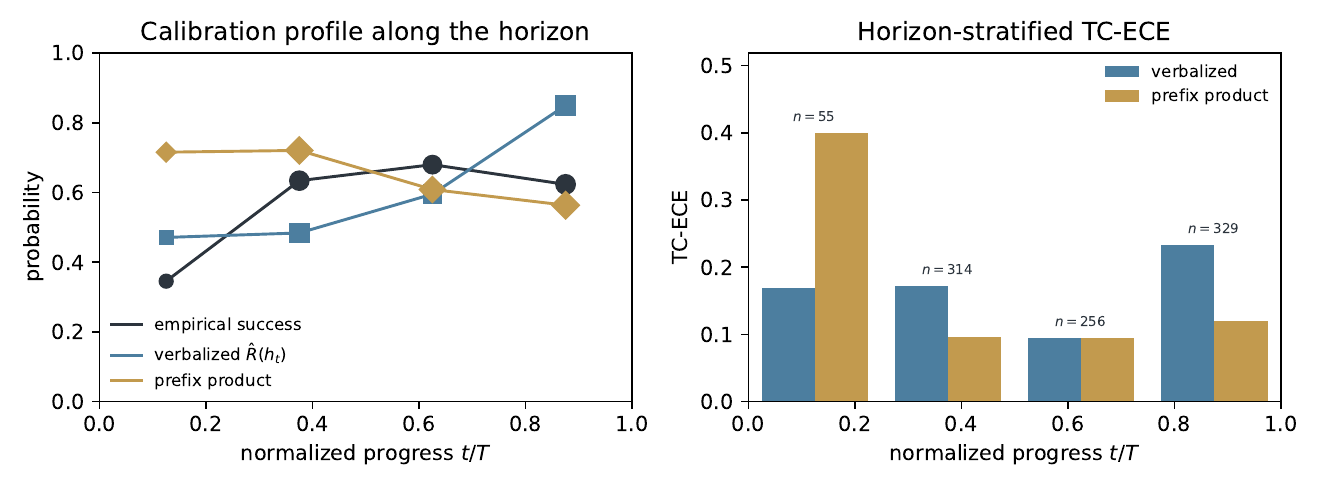}
  \caption{Calibration across normalized trajectory progress, for the
  verbalized estimate and the prefix product. Left: stated reliability against
  empirical success; marker area is proportional to $\sqrt{n}$ within the
  stratum. Right: TC-ECE within each progress quartile, annotated with the
  stratum size. Read against empirical success, the verbalized profile is
  non-monotone, with overconfidence in the opening quartile, underconfidence
  in the two central quartiles, and overconfidence in the closing quartile.
  The prefix product runs the other way, starting far above the success rate
  and decaying below it, which is why its own largest error is in the opening
  quartile rather than the closing one. The opening and closing quartiles
  contain $55$ and $329$ checkpoints, so the opening-quartile estimates are
  the least precise on both panels. Both panels stratify by normalized
  progress $t/T$, which is available only after a run ends; the step-index
  stratification used for the main reading disagrees in sign at the opening,
  where $g(1)=-0.166$ is underconfident, because the opening $t/T$ quartile
  collects the first checkpoints of the few long trajectories, which are
  disproportionately failures. Section~\ref{sec:empirical} sets the two
  stratifications side by side; the late overconfidence appears under both.}
  \label{fig:tcece}
\end{figure}

\paragraph{Computability and non-degeneracy.}
Eliciting $\hat{R}(h_t)$ adds one output line per step. All checkpoints are parsed successfully, and the verbalized estimate yields a TC-ECE of $0.125$ ($95\%$ CI $[0.102,0.165]$).
The three metrics produce different estimator rankings (Table~\ref{tab:tcece}). This disagreement supports reporting a calibration gap, a proper scoring rule, and a ranking metric together, as recommended in Section~\ref{sec:eval-metrics}. Reporting only one metric would change the interpretation of the comparisons below.
We resample complete trajectories for every metric in Table~\ref{tab:tcece} because checkpoints within a trajectory share the outcome $Y^{(j)}$, a dependence that increases with the horizon; Appendix~\ref{app:expdetails} compares the two resampling units at this horizon.
Variation due to refitting has a larger effect than the choice of resampling unit. Refitting within each replicate approximately doubles the intervals for the two fitted estimators central to the comparison, while the two unfitted estimators barely move, the verbalized signal not at all to three decimal places and the prefix product by at most $0.002$ at either endpoint (the two interval columns of Table~\ref{tab:tcece}). Fitting variation dominates sampling variation in this setting, and the naive bootstrap understates uncertainty by about a factor of two for comparisons involving the fitted estimators.
\par
\paragraph{Pooling obscures a large shift from underconfidence to overconfidence.}
Stratification by step index illustrates Eq.~\eqref{eq:strat-ece}. However, whether the shift is visible depends on the statistic being stratified. Stratified ECE magnitudes alone do not reveal it reliably.
ECE magnitude ratios do not by themselves establish horizon-dependent miscalibration. Across the five estimators, the ratio $\mathrm{ECE}_{\max}/\mathrm{TC\text{-}ECE}$ ranges from $1.1$ to $9.3$ (last column of Table~\ref{tab:tcece}). However, $\mathrm{ECE}_{\max}$ is the maximum of positively biased estimates across groups of unequal size, so it can exceed the pooled value even under perfect calibration. A simulation of a predictor calibrated at every checkpoint index, using the same checkpoint structure and binning over $20{,}000$ draws, gives a null ratio with median $2.20$ and a $95\%$ range of $[1.13,\,5.02]$. Only the pooled-isotonic ratio lies above this null distribution; the other four do not, as the per-estimator $p$-values in Table~\ref{tab:tcece} show. Without the null comparison, these ratios largely reflect the thin late strata.
ECE aggregates \emph{absolute} gaps, whereas the observed pattern is signed. Errors with opposite signs can cancel under pooling, and a magnitude-based statistic cannot show the direction of this cancellation.
Define the signed gap $g(t)$ as the mean stated reliability at step $t$ minus the observed success rate. Table~\ref{tab:drift} in Appendix~\ref{app:expdetails-tables} shows a pronounced drift: the agent is underconfident near the beginning of a run and overconfident near the end, with $g(t)$ increasing monotonically from $-0.166$ at $t=1$ to $+0.320$ at $t=6$. After aggregation over trajectory positions, the early gap is $-0.096$ ($95\%$ CI $[-0.149,-0.041]$) and the late gap is $+0.286$ ($[+0.144,+0.415]$). Their difference is $+0.381$ ($[+0.248,+0.505]$) and is positive in all $5{,}000$ trajectory-bootstrap resamples. The slope of $g$ on $t$ is $+0.133$ per step ($[+0.102,+0.161]$). These results show that the pooled score of $0.125$ averages two large errors with opposite signs.
These results motivate two additions to the protocol in Section~\ref{sec:eval-metrics}.
A stratified report should include the signed gap in addition to $\overline{\mathrm{ECE}}$ and $\mathrm{ECE}_{\max}$. The magnitude measures guard against cancellation only because they can exceed the pooled value; they neither show the sign of the error nor reveal the cancellation. In this dataset, they do not identify the observed shift.
The minimum-size filter also removes the strata with the largest signed gaps. Strata $t=5$ and $t=6$ fall below the $30$-checkpoint threshold and are excluded from $\mathrm{ECE}_{\max}$, although they carry the largest gaps in Table~\ref{tab:drift}. A threshold remains necessary for a binned magnitude estimate, but the signed gap can be estimated with a much smaller $n$. Thin strata should therefore be reported with their sample counts rather than omitted.
Recalibration shows the same limitation. A single isotonic map fitted on pooled checkpoints reduces the pooled score to $0.028$, suggesting near-perfect calibration. As the last column of Table~\ref{tab:drift} shows, it shrinks the early gaps while leaving the late gaps essentially unchanged. Pooled recalibration corrects the \emph{average} bias but leaves the horizon-dependent bias. This is the type of failure that Definition~\ref{def:traj} is intended to detect and that a pooled metric cannot reveal.
\par
\begin{table}[t]
\centering
\caption{Paired differences against the causal baseline on HotpotQA,
$\Delta =$ estimator $-$ baseline, with nominal $95\%$ nested-bootstrap
intervals; $^{*}$ marks differences that survive the Holm adjustment over
each family. Lower is better for TC-ECE and Brier, higher for AUROC.}
\label{tab:paired}
\small
\setlength{\tabcolsep}{3pt}
\renewcommand{\arraystretch}{1.15}
\begin{tabular}{@{}lccc@{}}
\toprule
\thc{Estimator $\hat{R}(h_t)$} & \thc{$\Delta$TC-ECE} & \thc{$\Delta$Brier} & \thc{$\Delta$AUROC}\\
\midrule
\tband{4}{HotpotQA, Qwen2.5-32B (Holm over twelve comparisons)}
Verbalized success prob. & $+0.068$ $[-0.001,+0.127]$ & $+0.020$ $[+0.004,+0.034]$ & $+0.045$ $[+0.002,+0.112]$\\
\quad + isotonic recalibration & $-0.007$ $[-0.056,+0.046]$ & $+0.003$ $[-0.007,+0.013]$ & $+0.016$ $[-0.029,+0.059]$\\
Prefix product of step conf. & $+0.029$ $[-0.041,+0.090]$ & $-0.005$ $[-0.026,+0.016]$ & $+0.085$ $[+0.031,+0.155]^{*}$\\
Post-hoc $(t/T,\,T)$, sees $T$ & $+0.046$ $[-0.010,+0.106]$ & $-0.023$ $[-0.043,-0.002]$ & $+0.095$ $[+0.048,+0.149]^{*}$\\
\tband{4}{HotpotQA replay, Gemini 3.1 Pro (Holm over six comparisons)}
Verbalized success prob. & $+0.213$ $[+0.131,+0.279]^{*}$ & $+0.054$ $[+0.031,+0.081]^{*}$ & $+0.089$ $[+0.019,+0.163]^{*}$\\
Prefix product of step conf. & $+0.216$ $[+0.133,+0.283]^{*}$ & $+0.053$ $[+0.029,+0.079]^{*}$ & $+0.109$ $[+0.037,+0.186]^{*}$\\
\bottomrule
\end{tabular}
\end{table}
\paragraph{A causal horizon baseline provides a meaningful comparison.}

Trajectory position alone is competitive on the calibration axis. The causal baseline uses only the elapsed step index, reaches a TC-ECE of $0.046$, and is better calibrated than every estimator here except the recalibrated verbalized signal. It is, however, nearly constant, with an AUROC of $0.562$; since calibration does not imply informativeness (Section~\ref{sec:metrics}), the comparisons report paired differences on all three metrics, drawn from the same nested bootstrap replicates rather than from marginal intervals (Table~\ref{tab:paired}).
The zero-shot verbalized estimate gives mixed results. It is nominally worse on both probability metrics and nominally better at ranking, and the prefix product shows the same pattern. Nested refitting flips both readings, turning a naive-bootstrap calibration advantage for the baseline into an interval that includes zero and a null naive AUROC difference into a detected advantage for the agent's signal, so the refitting procedure must be applied consistently across metrics.

One caveat governs every interval here. Twelve paired differences are computed, five exclude zero at the nominal $95\%$ level, the $\Delta\mathrm{AUROC}$ carrying the reading has the weakest exclusion in the family (lower bound $+0.002$), and under the same Holm adjustment that the ALFWorld and frontier analyses use, neither of the two differences carrying the reading survives. We therefore report the HotpotQA comparisons as descriptive rather than inferential: stated confidence ranks successful trajectories somewhat better than the step index while providing less accurate probabilities, and neither statement is established.
Because the baseline is fitted while the verbalized estimate is zero-shot, we also recalibrate the latter by isotonic regression under the same cross-fitting. Recalibration closes the pooled gap, leaving all three paired intervals through zero. Its apparent ranking cost (the paired $\Delta\mathrm{AUROC}$ falls from $+0.045$ for the zero-shot signal to $+0.016$ after recalibration) is a finite-sample artifact: each replicate refits the monotone map on roughly $63\%$ of trajectories, and the coarser isotonic blocks create ties, not evidence that recalibration reduces discrimination. Parity, moreover, holds only under pooled metrics, since the recalibrated estimator's worst-stratum discrepancy exceeds the null expectation. These results are specific to one configuration, and at $n=300$ with three-step horizons the intervals would not detect a modest advantage in either direction.

The post-hoc baseline shows what leaked length information buys. Adding the realized $T$ improves Brier from $0.228$ to $0.204$ and AUROC from $0.562$ to $0.645$, which would make the baseline look clearly stronger than the agent's signal, yet $T$ is unavailable to a live monitor and outcome-dependent here, so it is not a valid reference point. Evaluation should report paired improvement over a baseline restricted to causally available features, state those features, and pair the calibration gap with a proper score and a ranking metric.

The calibration profile in Figure~\ref{fig:tcece} is non-monotone across normalized progress, overconfident in the opening quartile, underconfident in the middle, and overconfident in the closing quartile, whose TC-ECE of $0.233$ is nearly twice the pooled value; the late overconfidence is where a monitor would act. The step-index and normalized-progress stratifications disagree in sign at the beginning ($g(1)=-0.166$ against opening-quartile overconfidence) because the $t/T$ quartile mainly collects opening checkpoints of the few long trajectories, which are disproportionately failures; the groups differ, not the estimates. We use the step-index stratification for the main reading, since $T$ is unknown during a run and outcome-dependent here, and retain Figure~\ref{fig:tcece} because normalized-progress profiles are commonly reported and the late overconfidence appears under both stratifications.
\par
\paragraph{A frontier API model on the same questions.}
To test whether these comparisons are an artifact of one mid-sized open-weight model, we replay the identical protocol, the same $300$ questions, prompts, and step budget at temperature $0$, with a frontier API model (Gemini~3.1~Pro, accessed through Vertex AI on 28 August 2026). The model is stronger on the task, with a success rate of $0.717$ against $0.657$, shorter trajectories of $2.7$ steps on average, and a parsable trajectory-success estimate at $99.8\%$ of steps.
Its stated reliability is nevertheless overconfident at every step index. The signed gap is $+0.245$ $[+0.194, +0.297]$ over early checkpoints and $+0.608$ $[+0.333, +0.843]$ over the late ones, which number only $13$ and should be read with that size in mind; the late-minus-early difference of $+0.363$ $[+0.084, +0.597]$ is positive in $99.1\%$ of trajectory resamples. Deterioration toward late-checkpoint overconfidence appears in every configuration in this section, while the sign at the early end varies by model rather than by task, underconfident for the two Qwen runs and overconfident here and for the Llama run below. Pooled over checkpoints, the frontier model is both more accurate and worse calibrated than the $32$B model on the same questions, with a verbalized TC-ECE of $0.248$ against $0.125$.
The baseline comparison produces a third distinct outcome. With trajectories this short and uniform, the step index carries almost no information, and the causal baseline's AUROC drops to $0.495$. The verbalized forecast and the prefix product both rank above it by margins that survive the Holm adjustment over this six-comparison family, while remaining far worse calibrated, with $\Delta$TC-ECE above $+0.2$ for both (Table~\ref{tab:paired}, lower block). Across the three configurations, which signal is informative flips. Content signals rank but position does not here, position ranks but the verbalized forecast does not on ALFWorld, and the two are nominally close on the $32$B HotpotQA run. Reporting a calibration gap, a proper score, and a ranking metric against a causally restricted baseline is what makes these reversals visible.

\paragraph{Measured correlation and composition.}
A companion experiment tests the formal claims of Section~\ref{sec:formal} on chains of $T=4$ independent QA items from GSM8K and TriviaQA \citep{cobbe2021training,joshi2017triviaqa}, $150$ chains per condition, answered either in one shared context or in separate contexts. The design provides exact step labels while removing every dependence channel except the shared context window, so it tests whether the estimators and the bound behave as predicted when labels are known rather than how strongly errors couple in real agents; Appendix~\ref{app:expdetails-chains} reports the full setup and per-condition results.
Three results carry forward. First, multiplication amplifies step-level miscalibration, turning step ECEs of $0.16$--$0.20$ into trajectory ECEs of $0.36$--$0.53$, while a dependence-blind learned aggregator over the same confidences repairs most of the error, so what the product rule spoils here is step calibration rather than composition. Second, the signed decomposition of Appendix~\ref{app:proofs-pathwise}, evaluated on the only condition with a measured positive step-error correlation, attributes $+0.468$ of the overestimate to compounded step overconfidence and $+0.048$, an order of magnitude less, to dependence. Third, the pairwise step-error correlations are near zero except in shared-context GSM8K ($\hat\rho=0.089$ $[0.023,0.159]$), a marginal exclusion among four uncorrected tests and expected under a design that removes the coupling channels; the marginal product tracks realized reliability when $\hat\rho\approx 0$ and understates it within Proposition~\ref{prop:positive}'s bound when $\hat\rho>0$. The step-dependence claim for real agents therefore rests on the ALFWorld measurement below and on the algebra, not on these chains.
\par
Two limitations affect the intervals reported above. The HotpotQA intervals include refitting within the nested bootstrap and therefore represent fitting and sampling variation, whereas the companion experiment's aggregator is fitted once and evaluated on held-out chains. No interval adjusts for the binning or finite-sample bias of the ECE estimator itself; where that bias materially affects a conclusion, as in the stratified comparison, we provide a calibrated null instead (Section~\ref{sec:eval-metrics}).
\par
\begin{table}[t]
\centering
\caption{ALFWorld replication summary for the two open-weight runs. Signed
gaps are stated confidence minus the realized rate; positive values are
overconfident. All intervals are game-bootstrap $95\%$ ranges.}
\label{tab:alfrep}
\small
\setlength{\tabcolsep}{4pt}
\renewcommand{\arraystretch}{1.15}
\begin{tabular}{@{}lcc@{}}
\toprule
 & \thc{Qwen3.5-9B} & \thc{Llama-3.1-8B}\\
\midrule
\tband{3}{Run}
Episodes (games) & 512 (134) & 416 (134)\\
Mean steps; reaching step 20 & 33.5; \; 67.8\% & 38.3; \; 74.8\%\\
Success rate & 0.514 & 0.317\\
Format compliance & 0.996 & 0.993\\
\tband{3}{Within-game prefix coupling, by step label}
Non-regressive & $+0.035$ $[+0.002,+0.067]$ & $+0.021$ $[-0.006,+0.047]$\\
Admissible action & $+0.178$ $[+0.123,+0.236]$ & $+0.267$ $[+0.202,+0.329]$\\
Strict progress & $+0.395$ $[+0.304,+0.478]$ & $+0.265$ $[+0.119,+0.404]$\\
\tband{3}{Trajectory forecast, signed gap}
Early ($t=1..10$) & $-0.117$ $[-0.180,-0.052]$ & $+0.301$ $[+0.240,+0.356]$\\
Late ($t=41..50$) & $+0.420$ $[+0.362,+0.480]$ & $+0.832$ $[+0.790,+0.867]$\\
Late $-$ early & $+0.536$ $[+0.456,+0.619]$ & $+0.530$ $[+0.463,+0.600]$\\
\tband{3}{Step-level signal, signed gap}
Early; late & $-0.225$; $-0.311$ & $-0.107$; $-0.133$\\
Late $-$ early & $-0.086$ $[-0.116,-0.058]$ & $-0.026$ $[-0.051,-0.003]$\\
\tband{3}{First error (episodes with an error: 201; 258)}
Position, fraction of episode & $0.354$ $[0.313,0.395]$ & $0.200$ $[0.175,0.229]$\\
Stated action conf.\ at that step & $0.856$ $[0.833,0.877]$ & $0.844$ $[0.818,0.867]$\\
Plan returns to pre-error length & $0.602$ $[0.512,0.694]$ & $0.632$ $[0.563,0.696]$\\
Episode still succeeds & $0.100$ $[0.061,0.145]$ & $0.078$ $[0.046,0.115]$\\
Success when first error early & $0.029$ $[0.000,0.066]$ & $0.056$ $[0.025,0.095]$\\
\bottomrule
\end{tabular}
\end{table}
\paragraph{Long-horizon validation on ALFWorld.}
The two experiments above leave the long-horizon regime untested, and the chained-QA design removes the coupling channels by construction. We therefore run the same elicitation protocol on ALFWorld \citep{shridhar2021alfworld}, a simulated household benchmark with a genuine action--observation loop.
A ReAct agent based on Qwen3.5-9B answers $512$ episodes of the unseen split, covering $134$ distinct games with about $3.8$ stochastic repeats per game at temperature $0.5$; the repeats over one game hold task difficulty fixed, which is what allows a within-game coupling estimate below. Episodes run under a $50$-step budget and are genuinely long; Table~\ref{tab:alfrep} summarizes both runs. The environment ships a hand-coded expert planner, and the length of its remaining plan from the current state is a distance to the goal. A step is labeled correct when it does not increase this distance, which provides an exact per-step label rather than the proxy available on HotpotQA. All intervals resample games rather than episodes, because repeats of one game share difficulty.

The coupling that Proposition~\ref{prop:positive} depends on is directly measurable here. Under this step label, the within-game prefix-coupling coefficient, in which the remaining correlation reflects error propagation rather than shared task difficulty, is small but positive, and stricter step labels give substantially stronger coupling (middle block of Table~\ref{tab:alfrep}); the pooled coefficient under the default label is $+0.079$. The identity in Proposition~\ref{prop:positive} holds exactly at every horizon, and the marginal product understates prefix reliability in precisely the direction the proposition predicts under positive coupling. At $T=20$ the product of step marginals is $0.490$ against a realized prefix reliability of $0.622$, a gap of $+0.132$ inside the bound of $0.233$; at the full $50$-step budget the product reaches $0.139$ against $0.274$, and the bound has grown to $0.757$, consistent with the degradation analysis in Appendix~\ref{app:proofs}. The near-zero correlations of the chained-QA design therefore reflect that design, and restoring the action--observation loop restores the coupling.

The trajectory forecast replicates the HotpotQA drift at ten times the horizon. Its signed gap, the mean stated success probability at step $t$ minus the realized success rate among episodes reaching $t$, moves from underconfident over the first ten step indices to overconfident over the last ten, and the late-minus-early difference is positive in every game-bootstrap resample (Table~\ref{tab:alfrep}). Early underconfidence followed by late overconfidence therefore appears on both tasks and both models, and it is largest exactly at the late checkpoints where a monitor would act.
The step-level signal behaves differently on the same trajectories. Its signed gap, the stated step confidence minus the realized step-correctness rate, is negative at every one of the $50$ step indices and deepens toward the end, with a late-minus-early difference that is negative in every resample (Table~\ref{tab:alfrep}). At late steps the two elicited signals are therefore miscalibrated in opposite directions, the step confidence too low and the trajectory forecast too high. This mirrors the confidence dichotomy of \citet{xuan2026confidence}, in which one agent's confidences about different decisions are miscalibrated in different directions, here between the step-level and trajectory-level signals, so neither signal's calibration can be inferred from the other's. In all cases a pooled score hides the drift.

Table~\ref{tab:alfworld} reports the pooled levels, and Table~\ref{tab:alfpaired} the paired differences against the same causal step-index baseline as before, now with $B=50$. The comparison reverses the HotpotQA result cleanly: the verbalized forecast is worse than the step-index baseline on all three metrics, each difference surviving the Holm adjustment. The step index is genuinely informative on this task because failing episodes run toward the budget while successful ones finish early; the easiest task type finishes in $16.3$ steps on average with success $0.967$, while the hardest averages $46.1$ steps with success $0.135$. The prefix product, the running product of step confidences, is indistinguishable from the causal baseline at ranking while significantly worse on both probability metrics. Isotonic recalibration of the verbalized forecast matches the baseline's calibration but remains significantly worse as a probability forecast and at ranking. Recalibration repairs calibration and cannot create discrimination; Section~\ref{sec:applications} later measures the same asymmetry as a function of the label budget (Figure~\ref{fig:costquality}). Together with the HotpotQA comparison, where the verbalized signal ranked nominally better than the step index, this shows that whether stated confidence adds information over trajectory position is itself task-dependent and must be measured per deployment.

A third open-weight model replicates these findings at full game coverage. Under the identical protocol, Llama-3.1-8B-Instruct ran $512$ scheduled episodes in $16$ parallel rounds of $32$ episodes each, covering all $134$ games at about $3.1$ retained stochastic repeats per game. The model is markedly weaker on the task than Qwen, with correspondingly longer episodes (Table~\ref{tab:alfrep}, top block). One property of this run tempers its scope. Whenever one call to the environment's expert planner exceeded a $120$-second watchdog, the whole round was abandoned and its $32$ in-flight episodes discarded; this happened to $3$ of the $16$ rounds, leaving the $416$ analyzed episodes. The rule biases the retained rounds away from states on which the expert planner grinds, although an earlier run of the same protocol saw an identical batch complete in under five minutes on one attempt and be killed on another, so the pathology is sampling-dependent rather than intrinsic to particular games. No episode is censored otherwise.
The core quantities replicate, and Tables~\ref{tab:alfrep} and~\ref{tab:alfpaired} place the two runs side by side. On the two outcome-tracking step labels, the within-game prefix coupling is positive with intervals that exclude zero. The lenient non-regressive label instead carries no stable signal; its interval includes zero here, and its Qwen counterpart excludes zero only marginally. A label whose marginal rate is $0.96$ is too weak to measure coupling; this is the labeling problem of Section~\ref{sec:formal} made concrete, and a reason to report every label convention rather than select one. The trajectory forecast is again overconfident exactly where a monitor would act, with a late-minus-early difference positive in every game resample; the step-level signal is again negative throughout and deepens more weakly than Qwen's. The verbalized forecast again loses to the causal step-index baseline on all three metrics and by larger margins than Qwen. Here the loss is not merely to the baseline: at an AUROC of $0.293$ the stated trajectory confidence ranks \emph{below chance}, so on this run it is negatively associated with episode success and a monitor that trusted it would be worse off than one that ignored it. Overconfidence and anti-ranking are separate failures, and only the second makes the signal actively misleading; a pooled calibration number reports neither. Here the prefix product is also significantly worse on all three metrics, so its ranking parity with the baseline is specific to the Qwen configuration, while isotonic recalibration again matches the baseline's calibration and again cannot repair ranking (Table~\ref{tab:alfpaired}, lower block).
That last row is worth reading against Section~\ref{sec:metrics} rather than as a success. Isotonic regression is monotone, so it cannot move AUROC except by creating ties; the rise from $0.293$ to $0.498$ is therefore the map collapsing towards a constant on a signal whose order it is required to preserve but cannot. The collapse is exact: the checkpoint-weighted success rate on this run is $0.110$, a constant forecast at that rate scores a Brier of $0.098$, and that is the value the recalibrated estimator attains, alongside a TC-ECE of $0.000$. This is precisely the case Section~\ref{sec:metrics} isolates, perfect calibration carrying no information about which episodes fail, and reporting the calibration gap alone would have ranked this degenerate estimator best of the four.

The exact step labels also locate each trajectory's \emph{first} error and what follows it, which gives the introduction's claim that agent errors arrive early and confidently its precise empirical form (Table~\ref{tab:alfrep}, bottom block). The first error falls about a third of the way through the episode for Qwen and a fifth of the way for Llama, and at that very step both models state action confidence above $0.84$. What the labels then show is sharper than irreversibility. The expert's remaining plan returns to its pre-error length in roughly six out of ten of these episodes for both models, yet the task subsequently succeeds in fewer than one in ten of them. Recovery is common and still not enough; what kills the trajectory is the budget spent undoing the damage. Episodes whose first error arrives early almost never succeed, while late first errors are directionally more survivable, although those strata are thin, with $n=18$ and $n=11$. This is the earliness requirement of Section~\ref{sec:control} in empirical form: by the time the damage is visible in the outcome, the budget to act on a warning is already gone.

\begin{table}[t]
\centering
\caption{ALFWorld (unseen split): pooled checkpoint levels under a $50$-step
budget, scored against episode success. The causal baseline is a logistic
model on $(t,\,t/B)$ with $B=50$, cross-fitted by episode. Paired differences
are in Table~\ref{tab:alfpaired}.}
\label{tab:alfworld}
\small
\setlength{\tabcolsep}{5pt}
\renewcommand{\arraystretch}{1.15}
\begin{tabular}{@{}lccc@{}}
\toprule
\thc{Estimator $\hat{R}(h_t)$} & \thc{TC-ECE} & \thc{Brier} & \thc{AUROC}\\
\midrule
\tband{4}{Qwen3.5-9B, $512$ episodes over $134$ games}
Verbalized success prob. & 0.209 & 0.300 & 0.533\\
\quad + isotonic recalibration & 0.008 & 0.196 & 0.549\\
Prefix product of step conf. & 0.193 & 0.221 & 0.732\\
Causal $(t,\,t/B)$, deployable & 0.021 & 0.168 & 0.755\\
\tband{4}{Llama-3.1-8B-Instruct, $416$ episodes over $134$ games}
Verbalized success prob. & 0.650 & 0.595 & 0.293\\
\quad + isotonic recalibration & $<$0.001 & 0.098 & 0.498\\
Prefix product of step conf. & 0.093 & 0.109 & 0.662\\
Causal $(t,\,t/B)$, deployable & 0.009 & 0.089 & 0.768\\
\bottomrule
\end{tabular}
\end{table}
\begin{table}[t]
\centering
\caption{ALFWorld paired differences against the causal baseline,
$\Delta =$ estimator $-$ baseline, with nominal $95\%$ game-bootstrap
intervals; $^{*}$ marks differences that survive the Holm adjustment over
each run's twelve comparisons. The post-hoc baseline sees the realized
length $T$. As in Table~\ref{tab:paired}, each difference is a mean over
bootstrap replicates in which both estimators are refitted, so it need not
equal the difference between the point estimates of
Table~\ref{tab:alfworld}, which come from models fitted once on the complete
sample; the isotonic rows are where the two diverge in sign
(Appendix~\ref{app:expdetails-bootstrap}).}
\label{tab:alfpaired}
\small
\setlength{\tabcolsep}{3pt}
\renewcommand{\arraystretch}{1.15}
\begin{tabular}{@{}lccc@{}}
\toprule
\thc{Estimator $\hat{R}(h_t)$} & \thc{$\Delta$TC-ECE} & \thc{$\Delta$Brier} & \thc{$\Delta$AUROC}\\
\midrule
\tband{4}{Qwen3.5-9B}
Verbalized success prob. & $+0.190$ $[+0.154,+0.226]^{*}$ & $+0.130$ $[+0.106,+0.155]^{*}$ & $-0.216$ $[-0.256,-0.170]^{*}$\\
\quad + isotonic recalibration & $+0.010$ $[-0.020,+0.059]$ & $+0.028$ $[+0.022,+0.035]^{*}$ & $-0.219$ $[-0.271,-0.174]^{*}$\\
Prefix product of step conf. & $+0.171$ $[+0.140,+0.205]^{*}$ & $+0.051$ $[+0.035,+0.071]^{*}$ & $-0.016$ $[-0.044,+0.014]$\\
Post-hoc $(t/T,\,T)$, sees $T$ & $-0.019$ $[-0.035,-0.006]^{*}$ & $-0.165$ $[-0.183,-0.148]^{*}$ & $+0.242$ $[+0.215,+0.276]^{*}$\\
\tband{4}{Llama-3.1-8B-Instruct}
Verbalized success prob. & $+0.642$ $[+0.611,+0.672]^{*}$ & $+0.507$ $[+0.472,+0.541]^{*}$ & $-0.469$ $[-0.536,-0.399]^{*}$\\
\quad + isotonic recalibration & $+0.007$ $[-0.012,+0.043]$ & $+0.010$ $[+0.007,+0.013]^{*}$ & $-0.295$ $[-0.341,-0.251]^{*}$\\
Prefix product of step conf. & $+0.084$ $[+0.065,+0.104]^{*}$ & $+0.020$ $[+0.013,+0.027]^{*}$ & $-0.099$ $[-0.151,-0.042]^{*}$\\
Post-hoc $(t/T,\,T)$, sees $T$ & $-0.005$ $[-0.014,+0.003]$ & $-0.083$ $[-0.100,-0.066]^{*}$ & $+0.217$ $[+0.169,+0.268]^{*}$\\
\bottomrule
\end{tabular}
\end{table}
\par
These experiments use three open-weight models, one frontier API model, three tasks with clear correctness labels, and horizons up to a $50$-step budget. They do not establish any estimator as state of the art. They also do not establish that verbalized confidence is equivalent to a position baseline. On the HotpotQA traces, verbalized confidence performs nominally better at ranking and nominally worse as a probability forecast, neither difference surviving a multiplicity adjustment over the twelve comparisons made, and the two become indistinguishable at $n=300$ after pooled recalibration; on ALFWorld the same comparison reverses, and the position baseline wins on every metric. Neither outcome should be read as a property of verbalized confidence in general. The experiments show that pooling hides a large signed drift on both tasks. For the trajectory forecast, late overconfidence replicates across all four configurations, while the sign of the early gap varies by model; the step-level signal on ALFWorld drifts in the opposite direction, so the two elicited signals must be audited separately. They do not show that the magnitude measures in Eq.~\eqref{eq:strat-ece} detect such a shift in general; on the HotpotQA traces they do not, and only the signed gap reveals it.
The short horizon is the main limitation of the first two experiments. The average HotpotQA trajectory has $3.2$ steps, and $201$ of the $300$ runs contain exactly three steps, so Proposition~\ref{prop:positive} predicts small composition error in that regime. The ALFWorld run addresses the long-horizon regime directly and finds positive coupling and a widening product gap exactly where the proposition predicts them, but it remains a single simulated benchmark, its step labels inherit the expert planner's notion of progress, and the long-horizon evidence comes from mid-sized open-weight checkpoints; the frontier replication is confined to the short-horizon task. The tasks also differ in model and decoding temperature, so cross-task differences are not attributable to the task alone. The experiments nevertheless establish several narrower results. The protocol can be computed on real traces at negligible cost. Refitting within the bootstrap, restricting the baseline to causally available features, stratifying by position, and comparing the stratified statistic with a calibrated null are its main design choices, and each one changes the resulting interpretation.

\subsection{What Is Still Missing}
\label{sec:eval-missing}

Section~\ref{sec:challenges} states ten open problems; three of them are evaluation gaps. Their implications for evaluation are summarized below.
No cross-domain benchmark currently provides resources comparable to those available for single-turn UQ \citep{vashurin2025benchmarking}. Existing resources focus on one domain, such as RAG, robotics, or finance, and use incompatible protocols.
Mid-trajectory labels are similarly scarce but are less essential for checkpoint evaluation than is often assumed. Definition~\ref{def:tcece} evaluates checkpoints against the final outcome, so completed trajectories with final labels are sufficient for computing this metric (Section~\ref{sec:eval-metrics}). Additional labels are needed for step-level calibration, failure attribution, and a low-variance target for $\Pr(Y=1 \mid h_t)$. Process supervision provides a practical source of these labels \citep{lightman2024lets,wang2024mathshepherd}.
When uncertainty informs control, evaluation must measure whether the resulting decisions improve outcomes. This requires the intervention advantage in Eq.~\eqref{eq:advantage}, rather than failure prediction alone \citep{zhang2026calibration}.

\takeaway[sec:eval]{No common protocol yet releases cross-domain agent trajectories.
Evaluation should compare against a horizon baseline restricted to causally
available features, resample whole trajectories while refitting every fitted
estimator within each replicate, and report a calibration gap together with
a proper scoring rule and a ranking metric, because a position baseline can
be well calibrated while providing little information. Calibration should be
stratified by trajectory position, and the stratified report should include
the \emph{signed} gap, since pooling can cancel horizon-dependent errors
with opposite signs that an absolute-value statistic cannot show; a
magnitude statistic such as $\mathrm{ECE}_{\max}$ needs a calibrated null
because a maximum over thin strata can exceed the pooled score even for a
calibrated predictor. Across the four configurations of
Section~\ref{sec:empirical}, the trajectory forecast's late overconfidence
replicates throughout, while which signal ranks above a position baseline
flips, so single-configuration conclusions about estimators should not be
generalized. Final-outcome trajectories suffice for these checkpoint
metrics; step-resolved labels are needed only for step-level calibration and
failure attribution. When uncertainty informs control, evaluation should
measure intervention outcomes.}

\section{Applications and a Practitioner's Guide}

\label{sec:applications}

Uncertainty estimates are useful when they inform concrete decisions.
Figure~\ref{fig:guide} provides a practical guide to the methods reviewed in this paper. It considers six deployment requirements and links each one to the method families in Table~\ref{tab:families} and to the relevant sections.
This section summarizes the main application patterns in the literature.

\par

Computational cost is a central constraint that grows with the trajectory. A $k$-sample estimator applied at every step of a $T$-step trajectory needs about $k \times T$ generations, before the pairwise comparisons of semantic clustering \citep{kuhn2023semantic,farquhar2024detecting}, and resampling complete trajectories may be infeasible outright, since repeating a run that sends a message or deletes a file changes the environment rather than measuring the same event \citep{ruan2024toolemu}. Deployed systems therefore favor low-cost signals such as verbalized confidence, hidden-state probes \citep{kossen2024semantic,azaria2023internal}, and generation-free scores \citep{zhu2026towards}, leaving sampling-intensive methods to offline evaluation.
\par
Four considerations cut across the six branches of Figure~\ref{fig:guide} and decide which of them a given deployment should follow. Model access decides the candidate families, since a closed API leaves the black-box pair \textcolor{axyellow!75!black}{B1}--\textcolor{axyellow!75!black}{B2} (Section~\ref{sec:verbalized}) while token probabilities and hidden states enable sequence scores and near-free per-step probes (Section~\ref{sec:tokenlevel}). The per-step budget separates $k\times$ sampling methods from the near-$1\times$ verbalized, probe, and generation-free signals (Table~\ref{tab:families}). What the estimate must trigger sets the required quality, because a score that only annotates a log tolerates miscalibration that a score gating irreversible actions cannot, and when actions differ in recoverability the intervention-advantage framing of Section~\ref{sec:control} applies. Finally, whether labeled trajectories from the deployment distribution exist decides between fitted calibrators with a pre-launch run of the protocol of Section~\ref{sec:eval} and zero-shot trust, which the evidence of Section~\ref{sec:empirical} says deserves a horizon-stratified check rather than a pooled one.

\begin{figure}[t]
  \centering
  \includegraphics[width=0.98\textwidth]{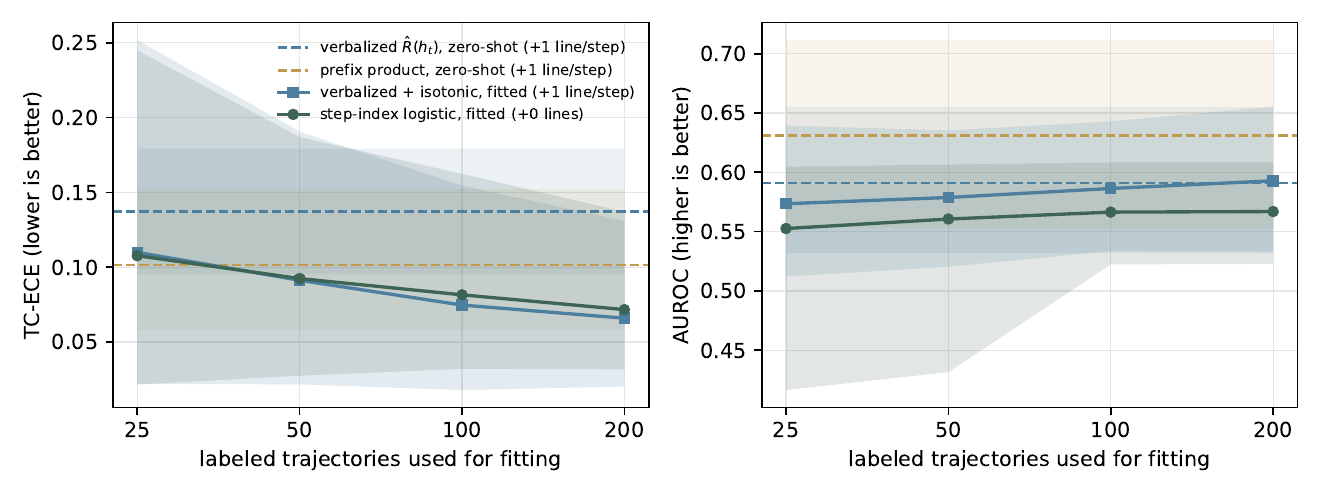}
  \caption{Measured cost--quality trade-off on the $300$ HotpotQA trajectories
  of Section~\ref{sec:empirical}. Curves show the two fitted estimators as a
  function of the number of labeled trajectories used for fitting, with bands
  giving the $[2.5,97.5]$ percentile range over $400$ disjoint fit and
  evaluation splits; dashed lines show the zero-shot estimators scored on the
  same evaluation sets. Legend entries state each estimator's exact
  elicitation overhead in generated lines per step. Labeled data buys
  calibration (left) but not discrimination (right). Sampling-based estimator
  families, which multiply generation cost by $k$ at every step, are not
  measurable on these traces.}
  \label{fig:costquality}
\end{figure}

\par
Figure~\ref{fig:costquality} makes the label-cost dimension concrete on the traces of Section~\ref{sec:empirical}. Labels buy calibration quickly, with both fitted estimators improving on the zero-shot verbalized TC-ECE in about three quarters of paired resamples at only $25$ labeled trajectories and almost always at $200$. Labels do not buy discrimination, since the step-index baseline's AUROC stays near $0.56$ at every fitting size and pooled isotonic recalibration only recovers the ranking quality the zero-shot signal already had; the best-ranking signal on these traces, the prefix product, is zero-shot at the same one-line-per-step cost. These numbers are specific to one model, task, and short horizon, and the supplementary material regenerates them from the released traces; sampling-based families would require new rollouts to place on the same axes.

\begin{figure}[t]
  \centering
  \resizebox{\textwidth}{!}{\begin{tikzpicture}[
  every node/.style={font=\footnotesize},
  q/.style={draw=axblue!55, fill=axblue!6, rounded corners=3pt, align=center,
            text width=4.7cm, inner sep=5pt},
  rec/.style={draw=axyellow!70!black, fill=axyellow!12, rounded corners=3pt,
              anchor=west, align=left, text width=9.7cm, inner sep=5pt},
  yesA/.style={-stealth, axaqua!55!black, line width=0.9pt},
  noA/.style={-stealth, gray!60, line width=0.9pt},
  lab/.style={font=\scriptsize\itshape, inner sep=1pt},
]

\node[q] (q1) at (2.9,0) {Do you need a \textbf{formal coverage guarantee}?};
\node[rec] (r1) at (6.3,0) {\textbf{Conformal methods}: prediction sets over actions, risk-controlled retrieval (Sections~\ref{sec:foundations}, \ref{sec:toolplanning}, \ref{sec:rag}). \emph{The guarantee is per step and marginal: composing $T$ of them costs $T\alpha$ under the union bound, and history dependence breaks exchangeability (Sections~\ref{sec:conformal}, \ref{sec:propagation}).} \, {\scriptsize\color{gray!48!black}\citealp{ren2023robots}; \citealp{rouzrokh2024conflare}; \citealp{quach2024conformal}}};

\node[q] (q2) at (2.9,-1.7) {Do you have \textbf{white-box access} to logits or hidden states?};
\node[rec] (r2) at (6.3,-1.7) {\textbf{Cheap internal signals}: token entropy, semantic-entropy probes, hidden-state classifiers (Section~\ref{sec:foundations}). \, {\scriptsize\color{gray!48!black}\citealp{kossen2024semantic}; \citealp{azaria2023internal}; \citealp{orgad2025llms}}};

\node[q] (q3) at (2.9,-3.4) {Does the \textbf{inference budget} rule out $k$-sampling at every step?};
\node[rec] (r3) at (6.3,-3.4) {\textbf{Single-pass signals only}: verbalized confidence, generation-free scores; spend sampling on a few pivotal checkpoints, never uniformly over the horizon (Sections~\ref{sec:foundations}, \ref{sec:applications}). \, {\scriptsize\color{gray!48!black}\citealp{zhu2026towards}; \citealp{tian2023just}; \citealp{wei2026longhorizon}}};

\node[q] (q4) at (2.9,-5.1) {Is \textbf{ambiguity in the request} the dominant risk?};
\node[rec] (r4) at (6.3,-5.1) {\textbf{Decompose, then clarify}: route aleatoric uncertainty to a question, not a guess (Section~\ref{sec:control}). \, {\scriptsize\color{gray!48!black}\citealp{matsnev2026uncertainty}; \citealp{edwards2026ask}; \citealp{kuhn2023clam}}};

\node[q] (q5) at (2.9,-6.8) {Is the \textbf{horizon long}, with compounding steps?};
\node[rec] (r5) at (6.3,-6.8) {\textbf{Propagate step signals}: situation-aware aggregation, early-warning dynamics (Section~\ref{sec:propagation}). \, {\scriptsize\color{gray!48!black}\citealp{zhao2024saup}; \citealp{duan2025uprop}; \citealp{darabi2026groundcontrol}}};

\node[q] (q6) at (2.9,-8.5) {Do \textbf{several agents} communicate?};
\node[rec] (r6) at (6.3,-8.5) {\textbf{Confidence-weighted communication}: discount unreliable messages, test consensus for correlated failure (Section~\ref{sec:multiagent}). \, {\scriptsize\color{gray!48!black}\citealp{yoffe2024debunc}; \citealp{huang2026counterfactual}}};

\node[rec, draw=axblue!55, fill=axblue!6] (fb) at (6.3,-10.2) {\textbf{Black-box default}: verbalized confidence plus sampling agreement, recalibrated on in-domain data (Section~\ref{sec:foundations}). \, {\scriptsize\color{gray!48!black}\citealp{tian2023just}; \citealp{lin2023generating}; \citealp{ulmer2024calibrating}}};

\foreach \i in {1,...,6} \draw[yesA] (q\i.east) -- node[lab, above, text=axaqua!45!black]{yes} (r\i.west);
\foreach \i/\j in {1/2,2/3,3/4,4/5,5/6} \draw[noA] (q\i.south) -- node[lab, right, text=gray!55!black]{no} (q\j.north);
\draw[noA] (q6.south) |- node[lab, pos=0.25, right, text=gray!55!black]{no} (fb.west);

\node[draw=axaqua!55!black, dashed, fill=axaqua!7, rounded corners=3pt, align=left,
      text width=15.1cm, inner sep=5pt, anchor=north west] at (0.45,-11.3)
  {\textbf{In every case}: evaluate at the trajectory level, not only with final-answer ECE; account for estimator cost at the deployed horizon before quality; and measure whether \emph{acting} on the estimate improves outcomes (Sections~\ref{sec:formal}, \ref{sec:eval}). \, {\scriptsize\color{gray!48!black}\citealp{zhang2026calibration}; \citealp{santilli2025revisiting}}};

\end{tikzpicture}}
  \caption{Decision guide mapping deployment constraints to estimator
  families and the relevant sections. The dashed box lists evaluation
  checks shared across choices.}
  \label{fig:guide}
\end{figure}
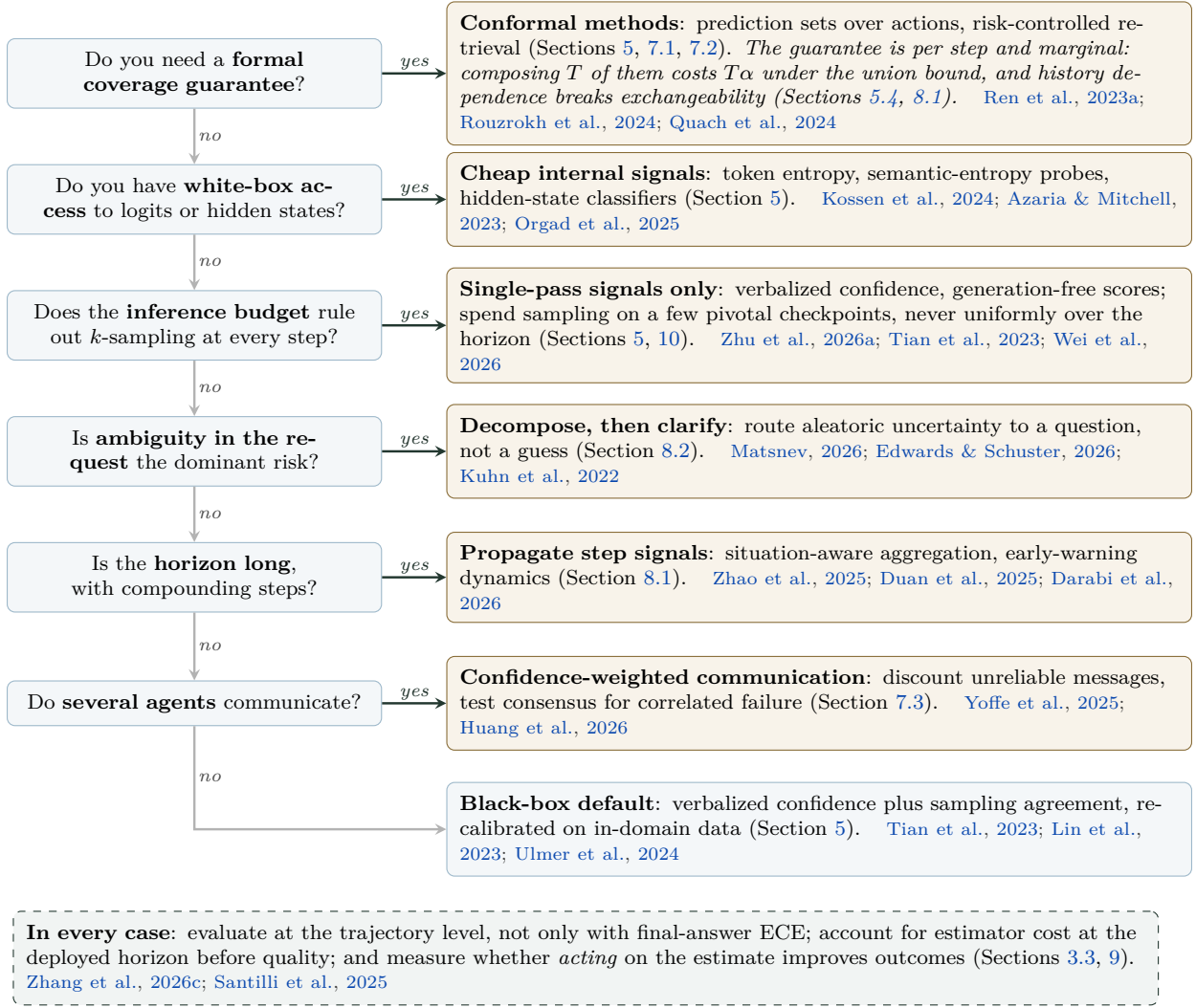

\begin{itemize}[leftmargin=1.4em,itemsep=3pt]

  \item \textbf{Safe deployment and runtime oversight.}
  Trajectory-level uncertainty can identify problems before recovery becomes impossible. Reliability estimates can trigger stopping or rerouting during navigation \citep{darabi2026groundcontrol}, provide an operational health signal in deployed multi-agent systems \citep{zhang2026managing}, and adjust an agent's permissions according to the estimated risk of its current action \citep{fleming2025uncertaintyaware}. These methods enable intervention during execution rather than only post-hoc evaluation after the trajectory ends \citep{zhang2026calibration}.

  \item \textbf{Abstention and human hand-off.}
  High uncertainty can lead an agent to stop, ask for clarification, or defer. Sequential agents can withhold high-risk actions \citep{piatrashyn2026redact}. Coding agents can ask questions when instructions are ambiguous \citep{edwards2026ask,deng2026uncertaintyaware}. Embodied agents can request help when autonomous action is unsafe \citep{yeke2026yesman,ren2023robots}, and high-stakes systems can send uncertain cases to human experts \citep{khanmohammadi2026calibrated}. The relevant decision is whether the expected value of intervention exceeds its cost, rather than how large the uncertainty score is.

  \item \textbf{Reliable retrieval and evidence use.}
  Uncertainty can determine whether retrieval is needed and how much the returned evidence should be trusted. Self-RAG learns reflection signals that indicate when to retrieve and whether a passage is relevant \citep{asai2024selfrag}. Other methods compare the model's confidence with evidence reliability or track changes in uncertainty after retrieval \citep{jia2026balancerag,shin2026era,binz2026uncertainty}.

  \item \textbf{Trustworthy multi-agent collaboration.}
  Uncertainty can support communication and aggregation. Methods can weight agents by estimated reliability instead of counting all votes equally \citep{li2026trusttrade}, detect false consensus among correlated agents \citep{huang2026counterfactual}, and propagate confidence through communication and shared memory \citep{yoffe2024debunc,essam2026trustaware}. Reliable consensus depends on dependence among agents as well as their level of agreement.

  \item \textbf{Uncertainty as a learning signal.}
  Uncertainty can guide training as well as runtime decisions. It may serve as a reward or exploration signal for self-improvement \citep{zhang2026selaur,zhou2026exploring}, or provide trajectory-level information about when to abstain or call a tool \citep{pan2026tiar}. A policy may, however, learn to reduce the reported uncertainty without resolving its underlying cause.

  \item \textbf{Domain deployments.}
  Uncertainty-aware agents have been studied across domains in which an action is costly to undo, and the deployment usually dictates which stage of the pipeline is instrumented.
  In medicine, specialist panels with consistency verification improve calibration in medical QA \citep{martinez2026multiagent} and latent diagnostic trajectories carry uncertainty across sequential clinical decisions \citep{shen2026uncertaintyguided}.
  In finance, selective consensus weights traders by estimated reliability \citep{li2026trusttrade} and long-horizon allocation is benchmarked under delayed feedback \citep{han2026can}.
  In software engineering, confidence is attached to shared artifacts and traced through development \citep{essam2026trustaware,ogunsusi2026uachatdev}, and evidence-calibrated auditing is applied to repository-level vulnerability detection \citep{meng2026vulnagentr}.
  Elsewhere, agentic negotiation for 6G networks targets tail-event risk and uncertainty neglect \citep{chergui2025llmbased}, supply-chain knowledge graphs are built under uncertainty guidance \citep{long2026helicase}, scientific reasoning is controlled by confidence during guided reflection \citep{tang2026rethinker}, driver monitoring is made risk-aware and selective \citep{qiu2026riskaware}, and table reasoning is delegated to a programmatic agent \citep{cheng2026tablemind}. Operational risk in deployed multi-agent systems is treated at the level of the whole pipeline rather than the model \citep{zhang2026managing}.

\end{itemize}

\par
Whatever the estimator, deployment reports should follow the protocol of Section~\ref{sec:eval-metrics}, stratifying calibration by trajectory position with the signed gap, because pooling averages away exactly the late-trajectory overconfidence a monitor most needs to detect, comparing against a baseline restricted to causally available features, and, when the estimate gates interventions, reporting outcomes after intervention rather than failure prediction alone \citep{zhang2026calibration}. In domains where actions are hard to reverse, which describes most of the deployments above, the earliness of a warning matters as much as its calibration (Section~\ref{sec:control}).

\takeaway[sec:applications]{Computational cost increases with trajectory length. Online systems
therefore tend to use low-cost signals such as verbalized confidence,
hidden-state probes, and generation-free scores, while sampling-intensive
methods are more suitable for offline evaluation. Complete-trajectory
resampling may be infeasible when actions change the environment. Across
applications, an uncertainty estimate is useful when it changes an action.
The main control quantity is the intervention advantage $A_i(h_t)$ from
Section~\ref{sec:control}, rather than only the risk
$1-\hat{R}(h_t)$. On the measured traces of
Figure~\ref{fig:costquality}, labeled data buys calibration but not
discrimination.}

\section{Open Challenges and Research Agenda}

\label{sec:challenges}

The taxonomy identifies several underdeveloped areas of research.
Table~\ref{tab:agenda} presents ten open problems, the main difficulty associated with each problem, and the evidence needed for progress.
This section discusses three problems that relate directly to the contributions of this paper. Section~\ref{sec:formal} defines trajectory-level calibration, and Section~\ref{sec:empirical} measures it. Section~\ref{sec:control} explains why intervention advantage is a more suitable control target than risk alone. Calibration transfer under distillation is the sparsest cell in our corpus, as discussed in Section~\ref{sec:training}.
Problems~4--10 extend beyond the results established here and are discussed in more detail in Appendix~\ref{app:challenges}.

\paragraph{Method used to identify the problems.}
The agenda is not derived mechanically from corpus counts. Such a procedure would be circular because the search terms were based on the same taxonomy used to define the sparse cells, as described in Appendix~\ref{app:corpus}.
We use a sparse cell only to identify a combination for closer inspection. The argument for each problem instead specifies the technical difficulty, such as a failed assumption, a missing label, or an unidentified quantity.
Problems~1 and~2 do not depend on corpus size. Problem~1 follows from Definition~\ref{def:traj} and the declining number of clean prefixes described in Section~\ref{sec:formal}. Problem~2 follows from the fact that calibrated risk does not determine an action. Problem~3 is more directly based on the observed corpus distribution, and we state it as such.
Two additional limitations apply. A cell may be sparse because the relevant research community uses different terminology; Appendix~\ref{app:corpus} names these neighboring communities where possible. The corpus also has a cutoff date of 9 July 2026. A claim that a topic is little studied therefore refers to the papers found by our search by that date. Claims that would be sensitive to a small number of missed papers are stated as corpus observations rather than field-wide conclusions.

\begin{table}[t]
\centering
\caption{Ten open problems in agent uncertainty, with the main obstacle and
the evidence needed for progress. Problems~1--3 are discussed here, and
Problems~4--10 are discussed in Appendix~\ref{app:challenges}.}
\label{tab:agenda}
\small
\setlength{\tabcolsep}{5pt}
\renewcommand{\arraystretch}{1.15}
\begin{tabular}{@{}p{0.35cm}L{3.6cm}L{5.1cm}L{5.5cm}@{}}
\toprule
\thc{\#} & \thc{Open problem} & \thc{Why it is difficult} & \thc{Evidence of progress}\\
\midrule
1 & Benchmarking trajectory-level calibration & No released traces, incompatible protocols, thin late-horizon strata & Cross-domain trajectories, step labels where available, position-specific metrics for $R$\\
2 & Moving from risk to intervention advantage & Risk does not say whether to act & Estimators of $A_i(h_t)$; protocols measuring post-intervention outcomes\\
3 & Uncertainty-aware distillation & Calibration transfer from teacher to student is poorly understood & Distillation objectives that preserve trajectory-level calibration\\
4 & Conformal guarantees for trajectories & History-dependent steps break exchangeability & Coverage under explicit sequential-dependence assumptions\\
5 & Source attribution and routing & One score conflates several uncertainty sources & Estimators that name the source; policies that pick the matching response\\
6 & Correlated failure in multi-agent systems & Communication couples agent errors & Consensus methods that model dependence and keep useful disagreement\\
7 & Efficient estimation & Sampling cost grows with trajectory length & Single-pass or low-cost estimators for online use\\
8 & Communicating uncertainty to operators & The same estimate moves different operators differently & Interfaces evaluated on operator decisions, not calibration alone\\
9 & Adversarially robust uncertainty & The confidence signal is itself attackable & Estimators and benchmarks that hold up under attack\\
10 & UQ for multimodal and computer-use agents & Perception adds uncertainty text-only estimators miss & Benchmarks that separate perception from decision uncertainty\\
\bottomrule
\end{tabular}
\end{table}

\begin{enumerate}[leftmargin=1.6em,itemsep=4pt]

\item \textbf{Benchmarking trajectory-level calibration.}
The field lacks a standard method for evaluating whether a trajectory-level reliability estimate $\hat R(h_t)$ is calibrated against final trajectory success $Y$. The problem has two parts requiring different supervision.
Scoring $\hat R(h_t)$ against final success needs no intermediate labels; complete trajectories with final outcomes support Definition~\ref{def:tcece} and its stratified variants at every observed prefix. The main missing resource is a set of released trajectories collected under a shared protocol, since most benchmarks publish aggregate scores rather than traces and existing trajectory-reliability studies use incompatible protocols \citep{trantruong2026measuring,nguyen2026urag}. Step-level labels are needed only for step calibration under Definition~\ref{def:step}, failure attribution, and low-variance estimates of $\Pr(Y=1\mid h_t)$ at rare prefixes, and process supervision or rollout-based labeling can supply them \citep{lightman2024lets,wang2024mathshepherd}.
Beyond released traces, the protocol needs three components absent from the single-turn setting, the stratification in Eq.~\eqref{eq:strat-ece}, a baseline restricted to causally available features, and a calibrated null for the stratified comparison, because a worst-stratum score can exceed a pooled score even for a calibrated predictor (Section~\ref{sec:empirical}). Longer released trajectories would also make late-horizon statistics estimable rather than only bounded, giving agents a shared basis comparable to single-turn UQ benchmarks \citep{vashurin2025benchmarking}. The benchmark must define success carefully; since calibration cannot eliminate fabrication for low-frequency claims \citep{kalai2024calibrated}, as discussed in Section~\ref{sec:metrics}, a trajectory-level estimator should be scored on whether $\hat R(h_t)$ predicts failure, not on whether calibration prevents all fabrication.
There is also a theoretical gap. Proposition~\ref{prop:positive} bounds the marginal-product error, and the mean-bias bound in Appendix~\ref{app:proofs} covers only the mean bias of $\prod_t c_t(h_{t-1})$; neither conditions jointly on reported confidence and trajectory position, so neither bounds the calibration error of the propagated estimator itself. Multicalibration may close this gap \citep{hebertjohnson2018multicalibration,detommaso2024multicalibration}, with the thin late-horizon strata as its main empirical difficulty.

\item \textbf{Moving from risk to intervention advantage.}
A calibrated failure probability does not specify which action an agent should take. \citet{zhang2026calibration} propose evaluating the value of intervention instead, the quantity $A_i(h_t)$ in Eq.~\eqref{eq:advantage} that compares the expected result of intervention $i$ with continuing (Section~\ref{sec:control}). It is hard to estimate because it is counterfactual, and two histories with the same failure probability can differ greatly in recoverability. Progress requires estimators of intervention-specific value and protocols that apply interventions and measure their later effects; risk prediction remains useful but does not define a complete control policy.

\item \textbf{Uncertainty-aware distillation for agents.}
Agent capabilities are increasingly transferred to smaller models \citep{xu2024surveykd,sharma2025small}, but whether calibration transfers with them is unknown, and the question is sharpest for on-policy distillation \citep{agarwal2024onpolicy,gu2024minillm}, where the student generates a different history distribution from the teacher and the difference can grow with the horizon (Appendix~\ref{app:proofs-distill}). Existing objectives optimize prediction quality or task success without preserving trajectory-level reliability. The open problem is to transfer useful uncertainty information together with actions or logits and then evaluate calibration on histories generated by the student.

\end{enumerate}

\noindent All three share one bottleneck, the scarcity of labeled
trajectories from deployment-like distributions: trajectory-level
calibration needs released cross-domain traces, intervention advantage needs
outcomes recorded after interventions are actually applied, and
calibration-preserving distillation needs evaluation on student-generated
histories.

\section{Conclusion}

\label{sec:conclusion}

Uncertainty estimation for LLMs is a well-developed area, but most methods were designed for single-turn settings.
When LLMs act as planners, tool users, retrievers, and collaborators, uncertainty concerns more than the correctness of a final answer. It also concerns where uncertainty arises in a trajectory, how it changes over time, and which action the agent should take in response.
This paper argues that agent uncertainty should be studied through trajectory-level reliability rather than only single-turn confidence.
We organize the literature by pipeline stage, uncertainty source, and method family, and summarize the corpus in Table~\ref{tab:bigtable}. Many agent methods build on established single-turn approaches, including verbalized confidence, semantic uncertainty, conformal prediction, and calibration.
Agent settings also violate assumptions that are often adequate for single-turn tasks. Errors can accumulate across steps, tools and environments introduce additional uncertainty, trajectory labels are often unavailable, and calibrated risk does not directly determine an appropriate intervention.
These limitations motivate trajectory-level benchmarks, better separation of uncertainty sources, uncertainty-aware control, efficient estimation, and distillation methods that preserve calibration.

Three limitations define the scope of our conclusions.
The corpus is based on searches of arXiv, the ACM Digital Library, OpenReview, IEEE Xplore, and PMLR, with a cutoff of 9 July 2026. Every retained paper was labeled independently by two annotators, with an LLM judge proposing a resolution for each disagreement and a third annotator making the final call, and every taxonomy assignment was independently checked by a second annotator. A known-item audit against the reference lists of the closest surveys provides a separate relative-recall check. A statement that a combination is little studied nevertheless refers to the work found under this protocol, not to the absence of such work from the full literature, as explained in Appendix~\ref{app:corpus}. 
The formal results address only part of the problem. Proposition~\ref{prop:positive} gives an exact identity at every horizon, but the associated bound is informative only for short trajectories or near-zero step correlations. It becomes vacuous in the long-horizon setting that motivates much of this paper. The proposed metric is also a pooled statistic and is strictly weaker than the trajectory-level calibration requirement in Definition~\ref{def:traj}.
The experiments in Section~\ref{sec:empirical} illustrate the proposed protocol using four models from three families, including a frontier API model, and horizons up to a $50$-step budget. They show that protocol choices can change the conclusions of an evaluation, that the trajectory forecast's late overconfidence replicates across tasks, and that the value of a position baseline reverses between them. They do not establish how an estimator behaves in frontier-scale, tool-rich deployments. We report each case in which changing a protocol choice or the task reverses the result.
Evaluation and design principles for agent uncertainty remain unsettled. Stronger foundations may help agents identify likely failures and respond while recovery remains possible.

\section*{Broader Impact Statement}

This paper examines methods for improving the reliability of LLM agents and detecting failures earlier.
Better uncertainty estimates can support safer deployment by allowing agents to abstain, request help, defer to humans, or revise their plans.
These estimates can also create unwarranted trust when users interpret them as guarantees. The presentation of uncertainty can affect human decisions, and training may reduce an uncertainty signal without resolving its underlying cause.
Uncertainty estimates should therefore be treated as decision-support signals rather than safety guarantees. Their value depends on how agents, people, and the surrounding system use them.

\bibliographystyle{styles/tmlr}
\bibliography{bib/used-references}

\appendix
\section{Corpus Construction Protocol}
\label{app:corpus}

This appendix describes how we collected, screened, and classified the papers
used here. The search cutoff was \textbf{9 July 2026}. The released
corpus files include
\texttt{paper\_filter\_process/paper\_selection.csv}, which records the
paper-level labels described below,
\texttt{data/taxonomy\_table.csv}, which records the taxonomy labels, and the
bibliography file \texttt{bib/used-references.bib}. The search and screening
records additionally preserve query-level provenance, raw search results,
deduplication decisions, and screening dispositions.

\paragraph{Search protocol and retained search record.}

We searched five sources: arXiv, the ACM Digital Library, OpenReview, IEEE
Xplore, and PMLR. The queries combined terms about uncertainty, LLMs, and
agents. The recorded concept terms are \emph{uncertainty quantification,
confidence estimation, calibration, LLM agent, tool use, retrieval,
multi-step reasoning, multi-agent, abstention,} and \emph{uncertainty
propagation}. We searched titles, abstracts, keywords, and other indexed
metadata when available. We merged the search results and removed duplicates
using DOI, arXiv identifier, and normalized title. A paper was retained if it
studied uncertainty, confidence, or calibration in an LLM-based system, or if
it provided a method directly used by this literature. 

The search process was recorded at query level. We retained the
database-specific query strings, execution dates, raw search results,
per-query hit counts, deduplication records, and screening decisions. The
literal database-specific queries and their execution metadata are preserved
in the search records; Table~\ref{tab:search-audit} summarizes their common
conceptual structure and recorded yields.

Let $U$ denote \{``uncertainty quantification'', confidence, calibration,
abstention, ``uncertainty propagation''\}, $L$ denote \{``large language
model'', LLM\}, and $A$ denote \{agent*, ``tool use'', retrieval,
``multi-step reasoning'', multi-agent\}. The common Boolean structure is
$U \land L \land A$, with syntax adapted to the search interface and indexed
metadata supported by each source.

\begin{table}[t]
\centering
\caption{Recorded database search protocol and search yields. Here $U$
denotes the uncertainty block, $L$ the large-language-model block, and $A$
the agent and agent-component block. Terms within each block are joined by
OR, and the three blocks are joined by AND. \emph{Raw} reports the number of
records returned by the corresponding recorded search before cross-database
deduplication. \emph{Indexed in labeled set} reports the number of rows in
\texttt{paper\_selection.csv} whose \texttt{platform} value corresponds to
that source. Because \texttt{platform} records where the version of record is
indexed rather than the complete query provenance of each paper, this column
is a corpus-indexing summary rather than a per-query screening yield. The 212
rows under \emph{Other indexes} are spread over 23 further databases,
chiefly the ACL Anthology (111) and the NeurIPS proceedings (49). The column
sums to the 717 labeled rows.}
\label{tab:search-audit}
\small
\setlength{\tabcolsep}{5pt}
\renewcommand{\arraystretch}{1.12}
\begin{tabular}{@{}L{2.5cm}L{5.6cm}L{2.2cm}rr@{}}
\toprule
Source & Recorded query structure & Date/filter & Raw & Indexed in labeled set\\
\midrule
arXiv API
& \texttt{all:($U$ AND $L$ AND $A$)}
& Through 9 Jul 2026 & 842 & 298\\

ACM Digital Library
& Advanced search over title, abstract, and keywords: $U$ AND $L$ AND $A$
& Through 9 Jul 2026 & 312 & 54\\

OpenReview
& Export title/abstract metadata, then apply $U$ AND $L$ AND $A$ locally
& Through 9 Jul 2026 & 238 & 79\\

IEEE Xplore
& \texttt{"All Metadata":($U$ AND $L$ AND $A$)}
& Through 9 Jul 2026 & 527 & 10\\

PMLR
& Search indexed title/abstract metadata with $U$ AND $L$ AND $A$
& Through 9 Jul 2026 & 486 & 64\\

\midrule
Other indexes
& Not searched directly; these records were reached through the five sources
above or through reference lists
& --- & --- & 212\\

\midrule
Total & & & 2405 & 717\\
\bottomrule
\end{tabular}
\end{table}

\paragraph{Paper-level labeling.}

We labeled every work in the resulting bibliography on a three-way scale:
\textbf{0} for the core corpus, meaning a work that we classify in
the taxonomy; \textbf{1} for a background work that the discussion cites but
does not classify; and \textbf{2} for exclusion. We added 143 negative
controls to the same labeling task, described below, which gives 717 rows in
total.

Two annotators labeled every row independently. They agreed on 537 of the 717
rows ($74.9\%$) and disagreed on 180. Because a raw rate is inflated when the
label distribution is skewed, we also report the chance-corrected coefficient.
Cohen's $\kappa$ over the three labels is $0.565$ ($p_o=0.749$, $p_e=0.423$),
which is moderate agreement. The two decision boundaries are comparably hard
rather than one easy and one hard; collapsing the scale to core corpus against
the rest gives $\kappa=0.567$, and collapsing to inclusion against exclusion
gives $\kappa=0.549$. This level of disagreement is why every conflicting pair
went to the adjudication step below rather than to a tie-break rule, and why
the released records retain both initial labels.

Each disagreement was passed to an LLM judge (Claude Opus 4.8). The judge
received the title, abstract, and label definitions, and had to choose between
the two labels the annotators had already assigned; it could not introduce a
third label. A third annotator then read the paper together with the judge's
choice and set the final label. That annotator accepted the judge's label on 144
of the 180 disagreements ($80\%$) and overruled it on 36. Every final label is
therefore set by a human, but on disagreements the judge's proposal was
adopted unchanged in four cases out of five.

The final labels are 120 core papers, 454 background works, and 143
exclusions. The columns \texttt{author\_checker\_1},
\texttt{author\_checker\_2}, \texttt{conflict\_flag},
\texttt{LLM\_judgment}, \texttt{third\_adjudication},
\texttt{third\_person\_label}, and \texttt{final\_decision} record this
process row by row.

\paragraph{Screening-flow records.}

The retained search and screening records trace the corpus from the original
database queries through cross-database deduplication and the successive
screening stages. Query provenance, deduplication status, screening stage,
screening disposition, and exclusion reason are preserved at record level,
allowing the flow quantities to be recomputed directly from the corresponding
logs. Table~\ref{tab:screening-flow} reports those quantities stage by stage,
with the record that evidences each one.

The 574 bibliography records are the papers that had survived assembly into
the cited bibliography when the labeling exercise was run. The final
bibliography cites eight further works added after that exercise, namely the
two ACL-probe papers from the recall check below, the two works surfaced by
the targeted gap audit below \citep{cui2026distilling,ruan2026doomed}, the
ALFWorld benchmark used by the experiments in Section~\ref{sec:empirical},
and three classical references on imitation learning and survival analysis
added during revision.
These eight works carry no label and enter neither the 120-paper corpus counts
nor the 454-work background set. The additional 143 rows are deliberately
unrelated negative controls added to the subsequent labeling exercise; they
are not papers rejected during the literature-search screening process.

\begin{table}[t]
\centering
\caption{Screening-flow quantities computed from the retained search,
deduplication, and screening records.}
\label{tab:screening-flow}
\small
\setlength{\tabcolsep}{7pt}
\renewcommand{\arraystretch}{1.12}
\begin{tabular}{@{}L{8.4cm}rL{4.1cm}@{}}
\toprule
Stage & Count & Evidence\\
\midrule

Raw records returned across the five databases
& 2405
& Query and raw-export records\\

Unique records after cross-database deduplication
& \texttt{1736}
& Deduplication log\\

Records entering title screening
& \texttt{1736}
& Screening log\\

Records entering title--abstract screening
& \texttt{1012}
& Screening log\\

Records entering full-text screening
& \texttt{684}
& Screening log\\

Assembled bibliography entering the labeling exercise
& 574
& \texttt{paper\_selection.csv}, labels 0 or 1\\

Deliberately unrelated negative controls added
& 143
& \texttt{unrelated-controls.bib}\\

Rows independently labeled by two annotators
& 717
& \texttt{paper\_selection.csv}\\

Final core corpus
& 120
& Final label 0\\

Final background set
& 454
& Final label 1\\

Final label-2 negative controls
& 143
& Final label 2\\

\bottomrule
\end{tabular}
\end{table}

The complete row-level log for the 717-paper labeling exercise is
\texttt{paper\_filter\_process/paper\_selection.csv}. It contains the two
initial labels, disagreement flag, judge proposal, third-human decision, and
final label. The retained search and screening records separately preserve
the provenance and disposition of candidates from retrieval through
deduplication and screening. The targeted external audit described below
follows the same record-keeping principle by retaining every raw hit and
screening decision.

\begin{table}[t]
\centering
\caption{Agreement between the two independent annotators over all 717 rows.
Rows give the first annotator's label, columns the second. Diagonal cells are
agreements; every off-diagonal cell went to the adjudication step. Labels are
\textbf{0} core corpus, \textbf{1} background work, and \textbf{2}
exclusion. Cohen's $\kappa$ over the three labels is $0.565$.}
\label{tab:agreement}
\small
\setlength{\tabcolsep}{9pt}
\renewcommand{\arraystretch}{1.15}
\begin{tabular}{@{}lrrrr@{}}
\toprule
& \multicolumn{3}{c}{Annotator 2} & \\
\cmidrule(lr){2-4}
Annotator 1 & 0 & 1 & 2 & Total\\
\midrule
0\quad core corpus
& \textbf{93} & 39 & 15 & 147\\

1\quad background work
& 29 & \textbf{340} & 40 & 409\\

2\quad exclusion
& 16 & 41 & \textbf{104} & 161\\
\midrule
Total
& 138 & 420 & 159 & 717\\
\bottomrule
\end{tabular}
\end{table}

Table~\ref{tab:agreement} shows where the two annotators differed. Of the 180
disagreements, 68 lie between core corpus and background work, 81 between
background work and exclusion, and 31 between core corpus and exclusion. The
boundary that produced most disagreement is therefore not whether a paper is
relevant, but whether a relevant paper should be classified in the taxonomy
or cited as background. In chance-corrected terms the two boundaries are
nevertheless comparably hard, since the two collapsed coefficients above,
$0.567$ and $0.549$, are nearly equal.

The exclusion--exclusion cell contains negative controls only. The
bibliography holds no excluded papers by construction, so over its 574 rows
that cell is empty and the label-2 margins fall to 35 and 38 for the two
annotators.

\paragraph{Negative controls.}

The 143 rows with final label 2 are not papers rejected during the literature
search. They are a control set drawn from subject areas this paper does not
cover, harvested from four of the same sources, namely arXiv, the ACM Digital
Library, OpenReview, and PMLR, and then filtered against an exclusion
vocabulary so that nothing about LLMs, agents, uncertainty, calibration, or
confidence survives
(\texttt{scripts/build\_unrelated\_controls.py}). We mixed them into the
labeling task as distractors. All 143 were excluded.

Both annotators marked 104 of them for exclusion directly. The remaining 39
passed through the adjudication step above and were also excluded. On the 574
non-control rows the two annotators agreed on 433 ($75.4\%$).

These controls are built to be clearly out of scope, so they test whether the
criteria reject plainly irrelevant work. They do not measure precision on
borderline candidates.

\paragraph{Exclusion categories.}

Candidates removed during the original literature search are represented
separately from the 143 negative controls. The retained screening record
stores the screening disposition and exclusion reason associated with each
removed candidate, allowing exclusions to be summarized by stage and reason
directly from the recorded decisions.
The retained records distinguish deduplication dispositions from screening
exclusion reasons. Duplicate records are resolved during cross-database
deduplication. Screening exclusion reasons include work outside the LLM
scope, work outside the agentic scope, absence of a relevant uncertainty
method, unavailable full text, and non-research items.
The 143 final label-2 rows instead form the deliberately unrelated negative
control set. They comprise 36 arXiv, 36 ACM Digital Library, 36 OpenReview,
and 35 PMLR records. They were sampled specifically to be out of scope and
were mixed into the labeling task as distractors, so they are reported
separately from candidates excluded during the search and screening process.

\paragraph{Taxonomy annotation.}

Every paper in the 120-paper core corpus was manually labeled along three
axes: pipeline stage, uncertainty source, and method family. A paper could
receive more than one label on an axis.
Two human annotators classified every paper independently and then compared
their labels. When they disagreed, they checked the paper and the category
definitions again and agreed on a final label. The final labels are stored in
\texttt{data/taxonomy\_table.csv}.

\paragraph{Recall check.}

We used the reference lists of the closest surveys
\citep{oh2026uncertainty,kirchhof2025position} to check whether the search had
missed relevant work. After applying our inclusion criteria, this reference
set contained 83 papers. This paper covers 79 of them. This gives a relative
recall of $0.95$, with a Wilson 95\% interval of $[0.88,0.98]$. This is a
relative measure based on the two reference lists, rather than an estimate of
absolute recall over all relevant literature.

As a further known-item probe outside the indexed sources, we ran four web
searches restricted to the ACL Anthology on 28 August 2026, using the
survey's core vocabulary (uncertainty quantification for agents, calibration,
abstention, retrieval triggers, and multi-agent debate confidence). Screening
the Anthology-published hits against the inclusion criteria yielded $10$
distinct qualifying papers. Eight were already covered by the corpus or the
background set; the two misses, an adaptive-retrieval method
\citep{yao2025seakr} and a medical-QA abstention study \citep{machcha2025know},
are now cited at their points of use, outside the labeled set. The probe is small
and biased toward what a search engine surfaces, so it complements rather
than replaces the relative-recall estimate above. Its result suggests that
ACL-only publication is a real but modest blind spot of the arXiv-centred
retrieval.

\paragraph{Taxonomy-independent targeted gap audit.}

Sparse cells can be created by the vocabulary used to build a taxonomy. We
therefore performed a separate arXiv API audit on 27 August 2026 using three
queries written without the taxonomy labels or category indices. We retained
all returned records in
\texttt{paper\_filter\_process/targeted\_gap\_search\_raw.csv}; the complete
row-level decisions and reasons are in
\texttt{paper\_filter\_process/targeted\_gap\_screening.csv}. The script
\texttt{scripts/run\_gap\_audit.py} stores the literal API requests and
reproduces the raw log. We applied the main-search cutoff to the first arXiv
submission date before interpreting the results.
Table~\ref{tab:gap-audit} reports the yield and screening depth of each of the
three checks together with what it found.

The exact arXiv \texttt{search\_query} values were, verbatim and including the
spaces inside each quoted phrase:

\smallskip

{\raggedright
\noindent\textbf{G1, uncertainty-aware distillation:}\\
\path|all:"uncertainty" AND all:"distillation" AND (all:"large language model" OR all:"LLM") AND (all:"agent" OR all:"agentic" OR all:"trajectory" OR all:"tool")|\par}

\smallskip

{\raggedright
\noindent\textbf{G2, verbalized confidence in agent decisions:}\\
\path|(all:"verbalized confidence" OR all:"self-reported confidence") AND (all:"planning" OR all:"multi-agent" OR all:"tool use" OR all:"LLM agent")|\par}

\smallskip

{\raggedright
\noindent\textbf{G3, internal probes for agent actions:}\\
\path|(all:"hidden state" OR all:"internal state" OR all:"probe") AND (all:"action correctness" OR all:"tool use" OR all:"LLM agent") AND (all:"uncertainty" OR all:"confidence" OR all:"calibration")|\par}

\begin{table}[t]
\centering
\caption{Taxonomy-independent arXiv gap audit. Raw counts are the API totals
on 27 August 2026; pre-cutoff counts retain first submissions no later than
9 July 2026. Every pre-cutoff hit was screened at title--abstract level.
G1 received the priority full-text check because it supports a central open
problem.}
\label{tab:gap-audit}
\small
\setlength{\tabcolsep}{4pt}
\renewcommand{\arraystretch}{1.12}
\begin{tabular}{@{}L{4.65cm}rrrrL{4.65cm}@{}}
\toprule
Check & Raw & Pre-cutoff & T/A screened & Full text & Result\\
\midrule

G1 uncertainty-aware distillation
& 22 & 21 & 21 & 7
& One calibration-transfer neighbor; zero end-to-end
student-trajectory calibration evaluations\\

G2 verbalized confidence
& 5 & 4 & 4 & 0
& Three direct title--abstract neighbors\\

G3 action probes
& 34 & 24 & 24 & 0
& Seven title--abstract neighbors, including one direct early-abort probe\\

\bottomrule
\end{tabular}
\end{table}

The G1 full-text check changes the interpretation of the sparse cell. BOND
\citep{cui2026distilling} distills a Bayesian opponent-belief posterior into an
8B student, evaluates it with a turn-level Brier score, and plots posterior
trajectories. It is therefore a concrete counterexample to a broad claim that
calibration-aware agent distillation is absent. However, none of the seven G1
full texts evaluates calibration of final trajectory success on histories
generated by the student policy. The external check therefore supports the
narrower gap used in this paper, namely preservation of end-to-end
trajectory-level calibration under the teacher-to-student distribution
shift, rather than the broader claim that uncertainty-aware distillation is
generally missing.

The G2 and G3 sensitivity checks likewise show that sparse means
under-covered, not empty: recent work studies verbalized confidence in
multi-agent coordination, and \citet{ruan2026doomed} use calibrated
hidden-state probes to abort likely-failing agent episodes early.

\paragraph{Background works.}

This paper also cites 454 background works from uncertainty estimation,
statistics, decision theory, control, robotics, information retrieval,
software engineering, and human--computer interaction. These are the rows
labeled \textbf{1} above, and they are listed in
\texttt{bib/background-works.bib}. They support the discussion but do not
enter the 120-paper corpus counts or the taxonomy heatmap. They are marked
with $^{\circ}$ in Table~\ref{tab:bigtable}.

\paragraph{Limitations.}

The search terms were developed together with the taxonomy, so both may share
some blind spots. The recall check also depends on the coverage of the surveys
used to construct its reference set.

Three limitations concern the labeling procedure. An LLM judge was involved
in resolving the 180 disagreements, and although a human set every final
label, that human agreed with the judge on $80\%$ of those rows, so the two
signals are not independent. The negative controls are deliberately
unambiguous, so they bound the error rate on obvious cases only.

The field changes quickly. We therefore treat sparse taxonomy cells as
patterns in the collected literature, not as proof that no related work
exists.

The released data support recomputation of the corpus counts and
Figure~\ref{fig:heatmap}. Table~\ref{tab:bigtable} also contains background
works and is intended as a reading guide rather than a direct copy of the
core-corpus CSV.

\section{The Full Taxonomy Table}
\label{app:bigtable}

Table~\ref{tab:bigtable} applies the three panels of Figure~\ref{fig:tree} to
the complete corpus. Its length requires appendix placement. Section~\ref{sec:taxonomy}
discusses the resulting distribution.

\begin{landscape}
\begingroup
\footnotesize
\setlength{\tabcolsep}{4pt}
\renewcommand{\arraystretch}{1.32}
\setlength{\LTcapwidth}{\linewidth}
\begin{longtable}{@{}>{\raggedright\arraybackslash}p{2.75cm}r@{\hspace{5pt}}p{2.85cm}>{\raggedright\arraybackslash}p{15.1cm}@{}}
\caption{Evidence table for the 120-paper corpus and curated background
works. Rows are grouped by pipeline stage and method/evidence category;
$^{\circ}$ marks works outside the core corpus. Each paper listed
appears exactly once, while Figure~\ref{fig:heatmap} retains multiple tags.
The table covers all 120 corpus papers and 352 of the curated background
works, 472 in total; the remaining cited works are supporting references that
carry no taxonomy placement. Rows are organized by Panels~C and~B; Panel~A
appears only as the dominant-source note in each group header, so this table
is not a three-axis listing. The full per-paper Panel~A tags are in
\texttt{data/taxonomy\_table.csv}.
\emph{The counts in this column are not the taxonomy counts used elsewhere in
this paper.} They count the entries listed in each row of this table, core
corpus and background works together, with every paper placed once; the
corpus counts quoted in the body and in Figure~\ref{fig:tree} are multi-label
tag counts over the 120 core papers alone. The two therefore differ in both
population and multiplicity, and a row here will not match the corresponding
figure leaf. Calibration is the clearest case: 25 entries are listed here,
whereas 59 of the 120 core papers carry the calibration property tag
(Section~\ref{sec:recalibration}).}\label{tab:bigtable}\\
\toprule
\textbf{Category} & \multicolumn{2}{@{}l}{\textbf{Count}} & \textbf{Papers}\\
\midrule
\endfirsthead
\multicolumn{4}{@{}l}{\textit{Table~\ref{tab:bigtable} (continued)}}\\[2pt]
\toprule
\textbf{Category} & \multicolumn{2}{@{}l}{\textbf{Count}} & \textbf{Papers}\\
\midrule
\endhead
\midrule
\multicolumn{4}{r@{}}{\textit{continued on next page}}\\
\endfoot
\bottomrule
\endlastfoot
\rowcolor{axblue!9}
\multicolumn{3}{@{}p{6.35cm}@{}}{\textbf{\color{axblue!55!black}Single turn}\; {\scriptsize\color{axblue!45!black}(upstream of the pipeline; 136 papers; Section~\ref{sec:foundations})}} & {\scriptsize\itshape\color{axaqua!35!black}main sources: tool/env. (4), epistemic (3)}\\
Calibration & 25 & \nbar{17.5} & \citealp{platt1999probabilistic}$^{\circ}$; \citealp{zadrozny2002transforming}$^{\circ}$; \citealp{naeini2015obtaining}$^{\circ}$; \citealp{guo2017calibration}$^{\circ}$; \citealp{kull2019beyond}$^{\circ}$; \citealp{nixon2019measuring}$^{\circ}$; \citealp{thulasidasan2019mixup}$^{\circ}$; \citealp{desai2020calibration}$^{\circ}$; \citealp{mukhoti2020calibrating}$^{\circ}$; \citealp{minderer2021revisiting}$^{\circ}$; \citealp{zhang2021knowing}$^{\circ}$; \citealp{zhao2021calibrate}$^{\circ}$; \citealp{chen2023close}$^{\circ}$; \citealp{lin2023generating}; \citealp{si2023prompting}$^{\circ}$; \citealp{detommaso2024multicalibration}$^{\circ}$; \citealp{liu2024litcab}$^{\circ}$; \citealp{srinivasan2024selective}; \citealp{ulmer2024calibrating}$^{\circ}$; \citealp{lacombe2025dont}; \citealp{liu2025cgspg}; \citealp{dai2026aligning}$^{\circ}$; \citealp{qiu2026riskaware}; \citealp{tropeano2026dont}; \citealp{yang2026scaling}$^{\circ}$\\
\addlinespace[4pt]
Conformal prediction & 21 & \nbar{14.7} & \citealp{vovk2005algorithmic}$^{\circ}$; \citealp{romano2019conformalized}$^{\circ}$; \citealp{tibshirani2019conformal}$^{\circ}$; \citealp{angelopoulos2021uncertainty}$^{\circ}$; \citealp{bates2021distribution}$^{\circ}$; \citealp{gibbs2021adaptive}$^{\circ}$; \citealp{barber2023conformal}$^{\circ}$; \citealp{kumar2023conformal}$^{\circ}$; \citealp{cherian2024large}$^{\circ}$; \citealp{gui2024conformal}$^{\circ}$; \citealp{mohri2024language}$^{\circ}$; \citealp{quach2024conformal}$^{\circ}$; \citealp{su2024api}$^{\circ}$; \citealp{emmenegger2026conformal}$^{\circ}$; \citealp{hittesdorf2026differentiable}$^{\circ}$; \citealp{kotte2026when}$^{\circ}$; \citealp{rubashevskii2026adaptive}$^{\circ}$; \citealp{wang2026beyondsurface}$^{\circ}$; \citealp{wang2026inference}$^{\circ}$; \citealp{xu2026geometry}$^{\circ}$; \citealp{zhou2026online}$^{\circ}$\\
\addlinespace[4pt]
Verbalized confidence & 20 & \nbar{14.0} & \citealp{lin2022teaching}; \citealp{mielke2022reducing}$^{\circ}$; \citealp{tian2023just}; \citealp{xiong2023can}; \citealp{zhou2023navigating}$^{\circ}$; \citealp{band2024linguistic}$^{\circ}$; \citealp{stengel2024lacie}$^{\circ}$; \citealp{xie2024calibrating}; \citealp{xu2024sayself}$^{\circ}$; \citealp{zhou2024relying}$^{\circ}$; \citealp{jang2025verbalized}; \citealp{khanmohammadi2026calibrated}; \citealp{kim2026asking}$^{\circ}$; \citealp{kumaran2026reported}$^{\circ}$; \citealp{marashian2026speaking}; \citealp{senoglu2026just}$^{\circ}$; \citealp{sun2026score}$^{\circ}$; \citealp{xia2026influential}; \citealp{xiao2026vlcalibration}; \citealp{zhang2026confidencecalibrated}\\
\addlinespace[4pt]
Sampling and consistency & 19 & \nbar{13.3} & \citealp{chen2023quantifying}$^{\circ}$; \citealp{kuhn2023semantic}; \citealp{varshney2023stitch}$^{\circ}$; \citealp{wang2023selfconsistency}$^{\circ}$; \citealp{duan2024shifting}$^{\circ}$; \citealp{farquhar2024detecting}$^{\circ}$; \citealp{gao2024spuq}$^{\circ}$; \citealp{kossen2024semantic}; \citealp{mundler2024self}$^{\circ}$; \citealp{nikitin2024kernel}$^{\circ}$; \citealp{qiu2024semantic}$^{\circ}$; \citealp{zhang2024luq}$^{\circ}$; \citealp{bouchard2026functional}$^{\circ}$; \citealp{chun2026comet}; \citealp{lin2026fase}$^{\circ}$; \citealp{sakai2026when}$^{\circ}$; \citealp{shamsi2026density}; \citealp{zhang2026whysemantic}$^{\circ}$; \citealp{zollo2026unsupervised}\\
\addlinespace[4pt]
Bayesian and ensemble & 11 & \nbar{7.7} & \citealp{blundell2015weight}$^{\circ}$; \citealp{gal2016dropout}$^{\circ}$; \citealp{lakshminarayanan2017simple}$^{\circ}$; \citealp{ovadia2019can}$^{\circ}$; \citealp{balabanov2024uncertainty}$^{\circ}$; \citealp{hou2024decomposing}$^{\circ}$; \citealp{kapoor2024large}$^{\circ}$; \citealp{ling2024uncertainty}$^{\circ}$; \citealp{yang2024bayesian}$^{\circ}$; \citealp{ou2026origins}; \citealp{zhang2026bayesiansparse}$^{\circ}$\\
\addlinespace[4pt]
Other (systems, analyses, benchmarks) & 11 & \nbar{7.7} & \citealp{liang2018enhancing}$^{\circ}$; \citealp{xiao2022uncertainty}$^{\circ}$; \citealp{fadeeva2023lmpolygraph}$^{\circ}$; \citealp{guerreiro2023looking}$^{\circ}$; \citealp{huang2023look}$^{\circ}$; \citealp{min2023factscore}$^{\circ}$; \citealp{chuang2024dola}$^{\circ}$; \citealp{tanneru2024quantifying}$^{\circ}$; \citealp{santilli2025revisiting}; \citealp{jiang2026value}$^{\circ}$; \citealp{yang2026zerosource}$^{\circ}$\\
\addlinespace[4pt]
Abstention and deferral & 9 & \nbar{6.3} & \citealp{kamath2020selective}$^{\circ}$; \citealp{cole2023selectively}$^{\circ}$; \citealp{yin2023large}$^{\circ}$; \citealp{cheng2024ai}$^{\circ}$; \citealp{feng2024dont}$^{\circ}$; \citealp{yang2024alignment}$^{\circ}$; \citealp{zhang2024rtuning}$^{\circ}$; \citealp{daheim2026uncertainty}$^{\circ}$; \citealp{gourabathina2026answering}\\
\addlinespace[4pt]
Internal states and probes & 9 & \nbar{6.3} & \citealp{lee2018simple}$^{\circ}$; \citealp{azaria2023internal}$^{\circ}$; \citealp{burns2023discovering}$^{\circ}$; \citealp{slobodkin2023curious}$^{\circ}$; \citealp{ahdritz2024distinguishing}$^{\circ}$; \citealp{chen2024inside}$^{\circ}$; \citealp{snyder2024early}$^{\circ}$; \citealp{su2024unsupervised}$^{\circ}$; \citealp{orgad2025llms}$^{\circ}$\\
\addlinespace[4pt]
Token probability & 8 & \nbar{5.6} & \citealp{ott2018analyzing}$^{\circ}$; \citealp{fomicheva2020unsupervised}$^{\circ}$; \citealp{glushkova2021uncertainty}$^{\circ}$; \citealp{jiang2021how}$^{\circ}$; \citealp{malinin2021uncertainty}$^{\circ}$; \citealp{schuster2022confident}$^{\circ}$; \citealp{yadkori2024believe}; \citealp{zhu2026towards}\\
\addlinespace[4pt]
Self-evaluation and verification & 3 & \nbar{2.1} & \citealp{agrawal2024language}$^{\circ}$; \citealp{dhuliawala2024chain}$^{\circ}$; \citealp{ren2024selfevaluation}$^{\circ}$\\
\addlinespace[4pt]
\rowcolor{axblue!9}
\multicolumn{3}{@{}p{6.35cm}@{}}{\textbf{\color{axblue!55!black}Planning and tool use}\; {\scriptsize\color{axblue!45!black}(pipeline stages 1--2; 38 papers; Section~\ref{sec:toolplanning})}} & {\scriptsize\itshape\color{axaqua!35!black}main sources: tool/env. (3), aleatoric (1)}\\
Other (systems, analyses, benchmarks) & 28 & \nbar{19.6} & \citealp{ahn2022saycan}$^{\circ}$; \citealp{huang2022inner}$^{\circ}$; \citealp{hao2023reasoning}$^{\circ}$; \citealp{liang2023code}$^{\circ}$; \citealp{patil2023gorilla}$^{\circ}$; \citealp{schick2023toolformer}$^{\circ}$; \citealp{shen2023hugginggpt}$^{\circ}$; \citealp{shinn2023reflexion}$^{\circ}$; \citealp{wang2023plan}$^{\circ}$; \citealp{yao2023react}$^{\circ}$; \citealp{yao2023tree}$^{\circ}$; \citealp{zhou2023least}$^{\circ}$; \citealp{gur2024real}$^{\circ}$; \citealp{han2024towards}; \citealp{huang2024metatool}$^{\circ}$; \citealp{qin2024toolllm}$^{\circ}$; \citealp{sumers2024cognitive}$^{\circ}$; \citealp{moshkovich2025taming}; \citealp{basu2026tool}$^{\circ}$; \citealp{bhatta2026uncertainty}$^{\circ}$; \citealp{broecker2026the}$^{\circ}$; \citealp{he2026confidence}$^{\circ}$; \citealp{seo2026from}$^{\circ}$; \citealp{shi2026code}$^{\circ}$; \citealp{vinod2026calvert}$^{\circ}$; \citealp{wei2026longhorizon}$^{\circ}$; \citealp{ye2026uncertainty}$^{\circ}$; \citealp{zhang2026webuncertainty}$^{\circ}$\\
\addlinespace[4pt]
Calibration & 3 & \nbar{2.1} & \citealp{chen2026et}$^{\circ}$; \citealp{pang2026case}$^{\circ}$; \citealp{xu2026when}$^{\circ}$\\
\addlinespace[4pt]
Conformal prediction & 3 & \nbar{2.1} & \citealp{lindemann2023safe}$^{\circ}$; \citealp{ren2023robots}$^{\circ}$; \citealp{liang2024introspective}$^{\circ}$\\
\addlinespace[4pt]
Verbalized confidence & 3 & \nbar{2.1} & \citealp{openai2023gpt4}$^{\circ}$; \citealp{leng2024taming}$^{\circ}$; \citealp{xuan2026confidence}\\
\addlinespace[4pt]
Internal states and probes & 1 & \nbar{0.7} & \citealp{healy2026internal}$^{\circ}$\\
\addlinespace[4pt]
\rowcolor{axblue!9}
\multicolumn{3}{@{}p{6.35cm}@{}}{\textbf{\color{axblue!55!black}Retrieval and memory}\; {\scriptsize\color{axblue!45!black}(pipeline stages 3--4; 84 papers; Section~\ref{sec:rag})}} & {\scriptsize\itshape\color{axaqua!35!black}main sources: tool/env. (30), accumulated (8)}\\
Other (systems, analyses, benchmarks) & 48 & \nbar{28.0} & \citealp{guu2020realm}$^{\circ}$; \citealp{karpukhin2020dense}$^{\circ}$; \citealp{lewis2020retrieval}$^{\circ}$; \citealp{shuster2021retrieval}$^{\circ}$; \citealp{borgeaud2022improving}$^{\circ}$; \citealp{chern2023factool}$^{\circ}$; \citealp{izacard2023atlas}$^{\circ}$; \citealp{jiang2023active}$^{\circ}$; \citealp{kandpal2023large}$^{\circ}$; \citealp{mallen2023not}$^{\circ}$; \citealp{packer2023memgpt}$^{\circ}$; \citealp{ram2023incontext}$^{\circ}$; \citealp{ren2023investigating}$^{\circ}$; \citealp{shao2023enhancing}$^{\circ}$; \citealp{trivedi2023interleaving}$^{\circ}$; \citealp{wang2023selfknowledge}$^{\circ}$; \citealp{asai2024selfrag}$^{\circ}$; \citealp{shi2024replug}$^{\circ}$; \citealp{su2024dragin}$^{\circ}$; \citealp{xie2024adaptive}$^{\circ}$; \citealp{xu2024knowledge}$^{\circ}$; \citealp{yan2024corrective}$^{\circ}$; \citealp{yoran2024making}$^{\circ}$; \citealp{zhong2024memorybank}$^{\circ}$; \citealp{fadeeva2025faithfulnessaware}; \citealp{suri2025structured}; \citealp{bazarova2026intrygue}$^{\circ}$; \citealp{dey2026interpretable}$^{\circ}$; \citealp{dong2026know}$^{\circ}$; \citealp{elchafei2026facet}$^{\circ}$; \citealp{feng2026kbsd}$^{\circ}$; \citealp{geissler2026towards}$^{\circ}$; \citealp{han2026whenevidence}$^{\circ}$; \citealp{hu2026detecting}$^{\circ}$; \citealp{jiao2026prunerag}; \citealp{julka2026when}$^{\circ}$; \citealp{jung2026uncertainty}$^{\circ}$; \citealp{kovacs2026beyond}$^{\circ}$; \citealp{liao2026belief}$^{\circ}$; \citealp{liu2026beyond}$^{\circ}$; \citealp{ni2026can}$^{\circ}$; \citealp{peng2026navigating}$^{\circ}$; \citealp{shen2026evidence}$^{\circ}$; \citealp{song2026llm}$^{\circ}$; \citealp{zhang2026remember}$^{\circ}$; \citealp{zhang2026stable}$^{\circ}$; \citealp{zhu2026trust}$^{\circ}$; \citealp{zong2026belief}$^{\circ}$\\
\addlinespace[4pt]
Calibration & 14 & \nbar{9.8} & \citealp{he2025mmboundary}; \citealp{min2025qucorag}; \citealp{soudani2025why}; \citealp{wu2025search}; \citealp{zhang2025miragebench}; \citealp{binz2026uncertainty}; \citealp{ding2026calibratethenact}; \citealp{jia2026balancerag}; \citealp{liu2026nova}; \citealp{meng2026equimem}; \citealp{ren2026when}; \citealp{sarkar2026leveraging}; \citealp{shin2026era}; \citealp{yeh2026retrieval}$^{\circ}$\\
\addlinespace[4pt]
Sampling and consistency & 11 & \nbar{7.7} & \citealp{kadavath2022language}; \citealp{manakul2023selfcheckgpt}; \citealp{li2024uncertaintyrag}; \citealp{shankar2025energy}; \citealp{stoisser2025towards}; \citealp{zhang2025measuring}; \citealp{matsnev2026uncertainty}; \citealp{qiu2026surerag}; \citealp{song2026car}; \citealp{sun2026cqc}$^{\circ}$; \citealp{wang2026when}\\
\addlinespace[4pt]
Conformal prediction & 6 & \nbar{4.2} & \citealp{rouzrokh2024conflare}; \citealp{chakraborty2025principled}; \citealp{kotla2025conformal}; \citealp{zhi2025seeing}; \citealp{chen2026is}$^{\circ}$; \citealp{nguyen2026urag}\\
\addlinespace[4pt]
Propagation and trajectory & 3 & \nbar{2.1} & \citealp{long2026helicase}; \citealp{pan2026tiar}; \citealp{tang2026rethinker}\\
\addlinespace[4pt]
Bayesian and ensemble & 2 & \nbar{1.4} & \citealp{deng2026uncertaintyaware}; \citealp{jiang2026discouq}\\
\addlinespace[4pt]
\rowcolor{axblue!9}
\multicolumn{3}{@{}p{6.35cm}@{}}{\textbf{\color{axblue!55!black}Multi-step reasoning and trajectories}\; {\scriptsize\color{axblue!45!black}(pipeline stage 5; 56 papers; Section~\ref{sec:propagation})}} & {\scriptsize\itshape\color{axaqua!35!black}main sources: accumulated (17), tool/env. (9)}\\
Training and distillation & 18 & \nbar{12.6} & \citealp{hinton2015distilling}$^{\circ}$; \citealp{cobbe2021training}$^{\circ}$; \citealp{uesato2022solving}$^{\circ}$; \citealp{hsieh2023distilling}$^{\circ}$; \citealp{agarwal2024onpolicy}$^{\circ}$; \citealp{gu2024minillm}$^{\circ}$; \citealp{lightman2024lets}$^{\circ}$; \citealp{wang2024mathshepherd}$^{\circ}$; \citealp{kale2026future}$^{\circ}$; \citealp{li2026when}$^{\circ}$; \citealp{oh2026neglected}$^{\circ}$; \citealp{pan2026uncertaintyaware}$^{\circ}$; \citealp{qi2026stapo}$^{\circ}$; \citealp{saadi2026validity}$^{\circ}$; \citealp{sermsri2026gatekd}$^{\circ}$; \citealp{wang2026madopd}$^{\circ}$; \citealp{yuan2026seva}$^{\circ}$; \citealp{zhu2026closing}$^{\circ}$\\
\addlinespace[4pt]
Propagation and trajectory & 14 & \nbar{9.8} & \citealp{zhao2024saup}; \citealp{chergui2025llmbased}; \citealp{alvarez2026where}$^{\circ}$; \citealp{badave2026beyond}$^{\circ}$; \citealp{cheng2026tablemind}; \citealp{feng2026conformal}$^{\circ}$; \citealp{kotawala2026locally}$^{\circ}$; \citealp{kotte2026pasc}$^{\circ}$; \citealp{liu2026diagnosing}$^{\circ}$; \citealp{padhi2026actions}$^{\circ}$; \citealp{singh2026agentbrace}; \citealp{yi2026measuring}$^{\circ}$; \citealp{zhang2026managing}; \citealp{zhou2026exploring}\\
\addlinespace[4pt]
Sampling and consistency & 13 & \nbar{9.1} & \citealp{madaan2023selfrefine}$^{\circ}$; \citealp{weng2023large}$^{\circ}$; \citealp{xie2023self}$^{\circ}$; \citealp{besta2024graph}$^{\circ}$; \citealp{huang2024large}$^{\circ}$; \citealp{miao2024selfcheck}$^{\circ}$; \citealp{tyen2024llms}$^{\circ}$; \citealp{duan2025uprop}; \citealp{darabi2026groundcontrol}; \citealp{essam2026trustaware}; \citealp{nguyen2026beyond}$^{\circ}$; \citealp{ogunsusi2026uachatdev}; \citealp{sethi2026dont}$^{\circ}$\\
\addlinespace[4pt]
Calibration & 6 & \nbar{4.2} & \citealp{mao2026confidence}$^{\circ}$; \citealp{meng2026vulnagentr}; \citealp{morandi2026sequential}; \citealp{yan2026denoiseflow}; \citealp{zhang2026agentic}; \citealp{zhang2026selaur}\\
\addlinespace[4pt]
Other (systems, analyses, benchmarks) & 3 & \nbar{2.1} & \citealp{lathkar2026anchored}$^{\circ}$; \citealp{pandey2026selfdoubt}$^{\circ}$; \citealp{shin2026the}$^{\circ}$\\
\addlinespace[4pt]
Bayesian and ensemble & 2 & \nbar{1.4} & \citealp{donaldson2026bayesian}; \citealp{shen2026uncertaintyguided}\\
\addlinespace[4pt]
\rowcolor{axblue!9}
\multicolumn{3}{@{}p{6.35cm}@{}}{\textbf{\color{axblue!55!black}Multi-agent coordination}\; {\scriptsize\color{axblue!45!black}(pipeline stage 6; 42 papers; Section~\ref{sec:multiagent})}} & {\scriptsize\itshape\color{axaqua!35!black}main sources: inter-agent (4), epistemic (1)}\\
Other (systems, analyses, benchmarks) & 21 & \nbar{14.7} & \citealp{irving2018ai}$^{\circ}$; \citealp{li2023camel}$^{\circ}$; \citealp{park2023generative}$^{\circ}$; \citealp{wu2023autogen}$^{\circ}$; \citealp{zheng2023judging}$^{\circ}$; \citealp{chen2024agentverse}$^{\circ}$; \citealp{hong2024metagpt}$^{\circ}$; \citealp{qian2024chatdev}$^{\circ}$; \citealp{chang2026cascadedebate}$^{\circ}$; \citealp{chen2026every}; \citealp{itkin2026delayed}$^{\circ}$; \citealp{jamshidi2026collective}$^{\circ}$; \citealp{jamshidi2026hallucination}$^{\circ}$; \citealp{li2026march}$^{\circ}$; \citealp{li2026source}$^{\circ}$; \citealp{liu2026game}$^{\circ}$; \citealp{rodrigues2026hallucination}$^{\circ}$; \citealp{sedoc2026trust}$^{\circ}$; \citealp{wang2026budgeted}$^{\circ}$; \citealp{xie2026ears}$^{\circ}$; \citealp{yadav2026legalhallulens}$^{\circ}$\\
\addlinespace[4pt]
Sampling and consistency & 13 & \nbar{9.1} & \citealp{du2024improving}$^{\circ}$; \citealp{khan2024debating}$^{\circ}$; \citealp{liang2024encouraging}$^{\circ}$; \citealp{smit2024should}$^{\circ}$; \citealp{baba2026argument}$^{\circ}$; \citealp{huang2026counterfactual}; \citealp{keramati2026confident}$^{\circ}$; \citealp{keramati2026early}$^{\circ}$; \citealp{li2026trusttrade}; \citealp{ma2026concat}$^{\circ}$; \citealp{martinez2026multiagent}; \citealp{tang2026the}$^{\circ}$; \citealp{zhu2026demystifying}$^{\circ}$\\
\addlinespace[4pt]
Calibration & 6 & \nbar{4.2} & \citealp{chen2024reconcile}$^{\circ}$; \citealp{yoffe2024debunc}$^{\circ}$; \citealp{armstrong2026margin}$^{\circ}$; \citealp{edwards2026ask}; \citealp{liu2026calibrating}$^{\circ}$; \citealp{wang2026orchestrating}$^{\circ}$\\
\addlinespace[4pt]
Conformal prediction & 2 & \nbar{1.4} & \citealp{wang2026from}$^{\circ}$; \citealp{zhang2026commcp}$^{\circ}$\\
\addlinespace[4pt]
\rowcolor{axblue!9}
\multicolumn{3}{@{}p{6.35cm}@{}}{\textbf{\color{axblue!55!black}Control: abstention and deferral}\; {\scriptsize\color{axblue!45!black}(acts on every stage; 31 papers; Section~\ref{sec:control})}} & {\scriptsize\itshape\color{axaqua!35!black}main sources: tool/env. (4), accumulated (3)}\\
Deferral, clarification, and oversight & 31 & \nbar{21.7} & \citealp{chow1970optimum}$^{\circ}$; \citealp{elyaniv2010foundations}$^{\circ}$; \citealp{amodei2016concrete}$^{\circ}$; \citealp{geifman2017selective}$^{\circ}$; \citealp{madras2018predict}$^{\circ}$; \citealp{min2020ambigqa}$^{\circ}$; \citealp{mozannar2020consistent}$^{\circ}$; \citealp{hendrycks2021unsolved}$^{\circ}$; \citealp{chen2023frugalgpt}$^{\circ}$; \citealp{kuhn2023clam}$^{\circ}$; \citealp{ding2024hybrid}$^{\circ}$; \citealp{greenblatt2024ai}$^{\circ}$; \citealp{ong2024routellm}$^{\circ}$; \citealp{yadkori2024mitigating}$^{\circ}$; \citealp{yue2024large}$^{\circ}$; \citealp{zhang2024clarify}$^{\circ}$; \citealp{fleming2025uncertaintyaware}; \citealp{lu2025auditing}; \citealp{baidya2026passiveqa}$^{\circ}$; \citealp{dixon2026adaptive}$^{\circ}$; \citealp{guo2026routenlp}$^{\circ}$; \citealp{iyer2026capability}$^{\circ}$; \citealp{luo2026agentic}$^{\circ}$; \citealp{madhusudhan2026knowing}$^{\circ}$; \citealp{piatrashyn2026redact}; \citealp{rajesh2026teaching}$^{\circ}$; \citealp{xiaohu2026know}$^{\circ}$; \citealp{yeke2026yesman}; \citealp{zhang2026calibration}; \citealp{zhao2026grace}$^{\circ}$; \citealp{zhou2026adaptive}$^{\circ}$\\
\addlinespace[4pt]
\rowcolor{axblue!9}
\multicolumn{3}{@{}p{6.35cm}@{}}{\textbf{\color{axblue!55!black}Evaluation and benchmarks}\; {\scriptsize\color{axblue!45!black}(measures every stage; 40 papers; Section~\ref{sec:eval})}} & {\scriptsize\itshape\color{axaqua!35!black}main sources: tool/env. (2), accumulated (1)}\\
Benchmarks and metrics & 40 & \nbar{28.0} & \citealp{rajpurkar2018know}$^{\circ}$; \citealp{chen2021evaluating}$^{\circ}$; \citealp{hendrycks2021measuring}$^{\circ}$; \citealp{lin2022truthfulqa}$^{\circ}$; \citealp{deng2023mind2web}$^{\circ}$; \citealp{golovneva2023roscoe}$^{\circ}$; \citealp{li2023apibank}$^{\circ}$; \citealp{li2023halueval}$^{\circ}$; \citealp{mialon2023gaia}$^{\circ}$; \citealp{prasad2023receval}$^{\circ}$; \citealp{srivastava2023beyond}$^{\circ}$; \citealp{zhuang2023toolqa}$^{\circ}$; \citealp{chen2024benchmarking}$^{\circ}$; \citealp{es2024ragas}$^{\circ}$; \citealp{jimenez2024swebench}$^{\circ}$; \citealp{kinniment2024evaluating}$^{\circ}$; \citealp{koh2024visualwebarena}$^{\circ}$; \citealp{liu2024agentbench}$^{\circ}$; \citealp{niu2024ragtruth}$^{\circ}$; \citealp{ruan2024toolemu}$^{\circ}$; \citealp{saadfalcon2024ares}$^{\circ}$; \citealp{xie2024osworld}$^{\circ}$; \citealp{yao2024tau}$^{\circ}$; \citealp{ye2024benchmarking}$^{\circ}$; \citealp{zhou2024webarena}$^{\circ}$; \citealp{mavi2025selfevaluating}; \citealp{vashurin2025benchmarking}$^{\circ}$; \citealp{chen2026tracesafe}$^{\circ}$; \citealp{han2026can}; \citealp{kargi2026uncertain}$^{\circ}$; \citealp{kc2026babeljudge}$^{\circ}$; \citealp{kirmayr2026carbench}; \citealp{kumar2026uncertainty}$^{\circ}$; \citealp{liu2026agenthallu}$^{\circ}$; \citealp{liu2026halluworld}$^{\circ}$; \citealp{oh2026uncertainty}; \citealp{tong2026does}$^{\circ}$; \citealp{trantruong2026measuring}; \citealp{wang2026a}$^{\circ}$; \citealp{yang2026whencalibration}$^{\circ}$\\
\addlinespace[4pt]
\rowcolor{axblue!9}
\multicolumn{3}{@{}p{6.35cm}@{}}{\textbf{\color{axblue!55!black}Related surveys and positions}\; {\scriptsize\color{axblue!45!black}(background; 45 papers; Section~\ref{sec:related})}} & {\scriptsize\itshape\color{axaqua!35!black}main sources: tool/env. (11), epistemic (3)}\\
Other (systems, analyses, benchmarks) & 26 & \nbar{18.2} & \citealp{abdar2021review}$^{\circ}$; \citealp{baan2023uncertainty}$^{\circ}$; \citealp{gao2023retrieval}$^{\circ}$; \citealp{gawlikowski2023survey}$^{\circ}$; \citealp{ji2023survey}$^{\circ}$; \citealp{mialon2023augmented}$^{\circ}$; \citealp{xi2023rise}$^{\circ}$; \citealp{zhang2023sirens}$^{\circ}$; \citealp{anwar2024foundational}$^{\circ}$; \citealp{du2024survey}; \citealp{guo2024large}$^{\circ}$; \citealp{he2024emerged}; \citealp{qin2024toollearning}$^{\circ}$; \citealp{shorinwa2024survey}; \citealp{wang2024survey}$^{\circ}$; \citealp{zhang2024surveymemory}$^{\circ}$; \citealp{huang2025how}; \citealp{huang2025survey}$^{\circ}$; \citealp{mohammadi2025evaluation}; \citealp{sarkar2025survey}; \citealp{xia2025survey}; \citealp{yu2025survey}; \citealp{papamarkou2026position}$^{\circ}$; \citealp{wang2026agent}; \citealp{yang2026reliable}; \citealp{zhang2026from}$^{\circ}$\\
\addlinespace[4pt]
Calibration & 8 & \nbar{5.6} & \citealp{wang2023calibration}; \citealp{geng2024survey}$^{\circ}$; \citealp{cheng2025empowering}; \citealp{hu2025survey}; \citealp{kirchhof2025position}; \citealp{sharma2025small}; \citealp{chen2026five}; \citealp{sengupta2026how}$^{\circ}$\\
\addlinespace[4pt]
Abstention and deferral & 3 & \nbar{2.1} & \citealp{wen2024know}; \citealp{liu2025uncertainty}; \citealp{su2025survey}\\
\addlinespace[4pt]
Bayesian and ensemble & 3 & \nbar{2.1} & \citealp{hullermeier2021aleatoric}$^{\circ}$; \citealp{he2023survey}; \citealp{huang2024survey}\\
\addlinespace[4pt]
Conformal prediction & 3 & \nbar{2.1} & \citealp{shafer2008tutorial}$^{\circ}$; \citealp{angelopoulos2023gentle}$^{\circ}$; \citealp{campos2024conformal}\\
\addlinespace[4pt]
Training and distillation & 2 & \nbar{1.4} & \citealp{xu2024surveykd}$^{\circ}$; \citealp{chang2025survey}\\
\end{longtable}
\endgroup

\end{landscape}

\section{Proofs and Extended Analysis for the Formal Lens}
\label{app:proofs}

This appendix collects the proof of Proposition~\ref{prop:positive} and the extended analysis that Section~\ref{sec:formal} summarizes. It covers the pathwise decomposition of the observed composition error, the mean-bias bound for the reported-confidence product, the sign analysis of the two error terms, the long-horizon behavior of the bound, and an entropy-based decomposition of trajectory uncertainty. Notation follows Section~\ref{sec:formal}.

\subsection{Proof of Proposition~\ref{prop:positive}}
\label{app:proofs-prop2}

The displayed equality follows from the telescoping identity $\prod_k q_k - \prod_k p_k = \sum_k (\prod_{t<k} q_t)(q_k - p_k)(\prod_{t>k} p_t)$ together with Eq.~\eqref{eq:hazard} and $q_1 = p_1$. The bound follows because all factors multiplying $q_k-p_k$ lie in $[0,1]$.
For the closed form, let $S_{k-1} = \prod_{s<k} Y_s$. Then $q_k - p_k = \mathrm{Cov}(Y_k, S_{k-1}) / \mathbb{E}[S_{k-1}]$, where $\mathbb{E}[S_{k-1}] = R_{k-1}$. Expanding the covariance gives $\mathrm{Cov}(Y_k, S_{k-1}) = \rho_k \sqrt{p_k(1-p_k)} \sqrt{R_{k-1}(1-R_{k-1})}$, since $S_{k-1}$ is Bernoulli with mean $R_{k-1}$, and dividing by $R_{k-1}$ gives the closed form in Proposition~\ref{prop:positive}. $\square$

\subsection{The Pathwise Decomposition of the Observed Composition Error}
\label{app:proofs-pathwise}

The remark after Proposition~\ref{prop:positive} states that the observed composition error separates into two mechanisms. Adding and subtracting the marginal product makes the separation explicit:
\begin{equation}
\underbrace{\textstyle\prod_{t} c_t \;-\; R}_{\text{observed composition error}}
\;=\;
\underbrace{\Big(\textstyle\prod_{t} c_t \;-\; \prod_{t} p_t\Big)}_{\text{reported-confidence--marginal gap}}
\;-\;
\underbrace{\Big(R \;-\; \textstyle\prod_{t} p_t\Big)}_{\text{dependence, Proposition~\ref{prop:positive}}}
\end{equation}
This identity holds pathwise. Only the dependence term is controlled by Proposition~\ref{prop:positive}, which bounds it and gives a closed form that can be estimated from held-out trajectories, whereas Definition~\ref{def:step} does not constrain the reported-confidence term, which therefore requires separate empirical estimation. Appendix~\ref{app:proofs-meanbias} derives a companion bound on the mean bias of $\prod_t c_t$ whose additional term depends on correlations among the \emph{reported} confidences, and explains why neither bound establishes calibration of the product.
For the chained-QA traces in Section~\ref{sec:empirical}, the reported-confidence term is substantially larger in the only condition with a measured positive step-error correlation, with an empirical mean of $\approx +0.47$ against $\approx +0.05$ for the dependence term, and the product continues to overestimate reliability after the dependence term is taken into account. On the long-horizon ALFWorld traces of the same section, the dependence term alone reaches $+0.13$ by a $20$-step horizon.

\subsection{The Mean Bias of the Reported-Confidence Product}
\label{app:proofs-meanbias}

Step-level calibration gives $\mathbb{E}[c_t] = \mathbb{E}\big[\Pr(Y_t=1\mid c_t)\big] = p_t$, so the reported confidence $c_t$ is unbiased for the step marginal even though a realized value $c_t(h_{t-1})$ need not equal $p_t$ (Section~\ref{sec:formal}).
Assume $c_t > 0$. Define $u_k = \mathbb{E}\big[\prod_{t \le k} c_t\big] \big/ \mathbb{E}\big[\prod_{t<k} c_t\big]$. Then $\mathbb{E}\big[\prod_t c_t\big] = \prod_k u_k$ and $u_1 = p_1$. Applying the same telescoping identity twice gives
\begin{equation}
\Big|\, \mathbb{E}\Big[\textstyle\prod_{t} c_t\Big] - R \,\Big|
\;\le\;
\underbrace{\sum_{k=2}^{T} \lvert u_k - p_k \rvert}_{\text{dependence among the \emph{reported} confidences}}
\;+\;
\underbrace{\sum_{k=2}^{T} \lvert q_k - p_k \rvert}_{\text{Proposition~\ref{prop:positive}}}
\end{equation}
where $u_k - p_k = \mathrm{Cov}\big(c_k,\, \prod_{t<k} c_t\big) \big/ \mathbb{E}\big[\prod_{t<k} c_t\big]$. This has the same form as the expression for $q_k - p_k$, with the clean-prefix indicator replaced by the confidence prefix.
The result has two scope limitations.
The bound controls the \emph{mean bias} of the naive propagation estimator; it does not establish calibration. A calibrated product would require $\Pr\big(Y=1 \mid \prod_t c_t = r\big) = r$. Neither term above controls this conditional relation. The term involving $u_k$ depends on correlations within the estimator and is not bounded by step-level calibration. An estimator can satisfy Definition~\ref{def:step} at every step even when its reported confidences are strongly correlated along the trajectory. Such correlations can arise when the $c_t$ values are based on shared context.
A corresponding multicalibration-style result would condition on both confidence and trajectory position and directly bound the calibration error of $\prod_t c_t$, rather than only its mean bias. No such result is established here; Section~\ref{sec:challenges} lists it as an open problem.

\subsection{Sign of the Two Terms in the Composition Error}
\label{app:proofs-sign}

Appendix~\ref{app:proofs-pathwise} displays the pathwise identity that separates the observed composition error $\prod_t c_t - R$ into a reported-confidence--marginal gap and the dependence term of Proposition~\ref{prop:positive}. Since $R$ and $\prod_t p_t$ are constants, the identity also holds term by term in expectation. These mechanisms can act in opposite directions.
The reported-confidence term is positive if $c_t > p_t$ at every step, and such deviations can accumulate across steps. By Proposition~\ref{prop:positive}, the dependence term is positive if $\rho_k > 0$ for every relevant $k$, and it enters the identity with a minus sign. Positive outcome dependence then raises $R$ relative to the marginal product and reduces $\prod_t c_t-R$. Under these conditions, the marginal product is conservative.
Only the dependence term is controlled by Proposition~\ref{prop:positive}, which bounds it and gives a closed form that can be estimated from held-out trajectories. The reported-confidence term measures the difference between the reported-confidence product and the product of step marginals. Definition~\ref{def:step} does not constrain it, so it requires separate empirical estimation. Neither proposition determines which term is larger.
Under nonnegative clean-prefix correlations, a large overestimate of trajectory reliability cannot be explained by positive cross-step dependence alone. It instead motivates separate checks of the calibration and dependence structure of the reported confidences $c_t$.

\subsection{Degradation of the Bound at Long Horizons}
\label{app:proofs-horizon}

The bound in Proposition~\ref{prop:positive} has an important limitation at long horizons. The dependence-induced discrepancy can persist, but the bound may then be uninformative. Under a constant hazard, let $q=1-\lambda$. The coefficients then grow geometrically with the horizon, with $R_{k-1}=q^{k-1}$ and $\sqrt{(1-R_{k-1})/R_{k-1}}\sim q^{-(k-1)/2}$. With $q=0.9$, $p_k=q$, and $T=20$, the final summand alone reaches $\approx 0.76\,\lvert\rho_{20}\rvert$. The coefficients in the sum total $\approx 7.97$, so the sum exceeds $1$ when the correlation is constant at about $0.13$. Since $\lvert R-\prod_t p_t\rvert\le 1$ holds trivially, the bound is mainly informative for short horizons or for long horizons when correlations are near zero. In these settings, the bound can assess the accuracy of the product approximation. It can upper-bound small discrepancies but becomes vacuous when its right-hand side exceeds the trivial bound of $1$. A tighter long-horizon bound would need to control joint dependence directly rather than sum step-specific terms; deriving such a result remains open. The signed identity displayed in Appendix~\ref{app:proofs-pathwise}, however, remains valid at every horizon.

\subsection{An Entropy View of Trajectory Uncertainty}
\label{app:proofs-entropy}

The chain rule for entropy gives a view of trajectory uncertainty that is related to, but weaker than, trajectory reliability. Conditioned on the initial state, $H(a_{1:T}, o_{1:T} \mid s_0) = \sum_t [\,H(a_t \mid h_{t-1}) + H(o_t \mid h_{t-1}, a_t)\,]$ \citep{han2024towards}. This decomposition separates uncertainty in the agent's actions from uncertainty introduced by tools and the environment. The action-only version $H(a_{1:T}\mid s_0)=\sum_t H(a_t\mid h_{t-1})$ is invalid because the histories $h_{t-1}$ include observations that are marginalized out on the left.
Even the full decomposition does not provide trajectory reliability. Definition~\ref{def:traj} concerns $\Pr(Y=1)$. It depends on whether the trajectory succeeds, not on how uncertain its actions are. Policy entropy may also reflect benign diversity among equally valid actions \citep{kuhn2023semantic}. A high-entropy step can be safe, while a low-entropy step can be confidently wrong. Summing step entropies is well defined but does not estimate trajectory success.

\subsection{The Calibration-Transfer Gap under Distillation}
\label{app:proofs-distill}

Suppose the teacher $\pi_T$ has a reliability estimator $\hat{R}$ calibrated
on its own trajectory distribution, and an on-policy student $\pi_S$ is
trained on its own trajectories by minimizing
\begin{equation}
\mathbb{E}_{\tau \sim \pi_S}
\left[\sum_t D\bigl(\pi_T(\cdot \mid h_t)\,\|\,\pi_S(\cdot \mid h_t)\bigr)\right].
\end{equation}
The relevant question is whether calibration is preserved under the change
from histories generated by $\pi_T$ to those generated by $\pi_S$. The
student-side calibration gap is
\begin{equation}
\Delta(\pi_S) \;=\; \sup_{c}\;
\Bigl|\, \mathbb{E}_{\tau \sim \pi_S}
\bigl[\,Y \mid \hat R(h_t) = c\,\bigr] \;-\; c \,\Bigr|.
\end{equation}
This gap is zero for the teacher by assumption but is otherwise unconstrained
for the student; Section~\ref{sec:training} discusses the mechanisms that can
enlarge it and the absence of evaluations that measure it.

\section{Relation to Classical Decision Theory: POMDPs and the Value of Information}
\label{app:pomdp}

Several concepts discussed in Section~\ref{sec:background} have direct counterparts in classical decision theory. POMDPs formalize how an agent should act when the true state of the environment cannot be observed directly \citep{astrom1965optimal,smallwood1973optimal,kaelbling1998planning}. The agent instead maintains a \emph{belief}, defined as a probability distribution over possible states conditioned on the interaction history, and plans over this belief space. Approximate solution methods include Monte Carlo tree search \citep{silver2010montecarlo} and trajectory optimization \citep{platt2010belief}. \citet{kochenderfer2015decision} provide a standard introduction. Subjective logic supplies a complementary formalism for beliefs with explicit uncertainty \citep{josang2016subjective}.
This framework helps interpret several quantities used in this paper. The history-conditioned reliability $\hat{R}(h_t)$ in Definition~\ref{def:traj} summarizes the agent's current belief that the trajectory will succeed. The hazard $\lambda_t(h_{t-1})$ in Eq.~\eqref{eq:hazard} gives the probability of failure at the next step conditioned on the current history. The connection also applies to active information gathering. Re-querying a tool or consulting another source may reduce uncertainty that a passive predictor cannot resolve. In a POMDP, such behavior is represented as an information-gathering action.
Intervention and control provide another connection. The intervention advantage in Eq.~\eqref{eq:advantage} compares the expected utility of intervening with that of continuing and is related to the classical value of information \citep{howard1966information}. Asking a clarifying question because its answer is expected to reduce uncertainty can similarly be viewed as Bayesian experimental design \citep{lindley1956measure,chaloner1995bayesian}. Asking for information, taking actions to verify evidence, and planning under uncertainty therefore extend established ideas from classical decision theory. Engineering decision analysis and safety assessment distinguish epistemic from aleatoric or subjective uncertainty and connect those distinctions to probability, sensitivity, and risk \citep{kiureghian2009aleatory,faber2005treatment,helton1997uncertainty,apostolakis1990concept}.
\par
Despite these connections, POMDPs appear infrequently in the literature collected for this paper. Differences in terminology provide one explanation. Our corpus was retrieved mainly from research on LLM uncertainty and agents (Appendix~\ref{app:corpus}), whereas work on POMDPs and belief-space planning generally uses different terminology. The limited overlap therefore reflects a separation between research communities, at least in part, rather than an absence of relevant theory.
A technical difference also remains. Classical POMDP methods generally assume an explicit model of the environment. The transition and observation dynamics are known, while the state is hidden. The agent can therefore maintain and update a belief distribution using Bayes' rule.
LLM agents generally lack such an explicit model. Their internal state is represented implicitly through model parameters and context rather than by an explicit posterior distribution. Their action spaces are also less structured because actions are often expressed through language or tool calls. A further source of uncertainty concerns the correctness of the agent's own reasoning and actions. Classical POMDP formulations normally assume specified transition and observation models, whereas an LLM agent may be uncertain about the reliability of its own decisions.
Classical methods therefore cannot be transferred directly without additional assumptions. The central task is to estimate quantities such as reliability, failure probability, and the value of additional information without the explicit probabilistic model assumed in classical decision theory. From this perspective, the methods in Sections~\ref{sec:propagation} and \ref{sec:control} provide partial approximations to the classical framework. Some estimate a useful belief statistic, while others estimate whether information gathering or intervention is beneficial.
Safe reinforcement learning and runtime shielding provide a related comparison \citep{garcia2015comprehensive,alshiekh2018safe,pecka2014safe}.
A shield is a reactive monitor synthesized from a formal safety specification and an environment model. It evaluates each proposed action and either permits it or substitutes a safe alternative, ensuring that the specification is not violated regardless of the action proposed by the learned policy \citep{alshiekh2018safe}. Constrained-MDP formulations impose a related constraint in expectation rather than as an absolute guarantee by maximizing return subject to a bound on expected cost \citep{garcia2015comprehensive}. Both approaches resemble the capability gating described in Section~\ref{sec:control}: they restrict the action space before execution and do not depend on the quality of the policy's own confidence estimate. Their protection therefore remains available when the step-level uncertainty estimator is inaccurate.
The approaches differ in what they assume and what they must estimate. A shield requires an environment model and a safety specification that can be checked against it; in return, it provides a hard guarantee. An LLM agent typically has neither. Its environment model is implicit, and the desired property that its reasoning is correct cannot generally be expressed as a state predicate. The available alternative is often a learned reliability estimate $\hat{R}(h_t)$ without a formal guarantee. The approaches therefore address complementary failure classes. Shields cover cases in which unsafe states or actions can be specified in advance, including irreversible tool calls and permission boundaries. Uncertainty-based oversight addresses cases such as an incorrect answer presented with high confidence. A combined system could use a specification-based shield to limit worst-case damage and an uncertainty estimate to trigger the softer interventions in Section~\ref{sec:control}. Our corpus contains little work on this combination.
The same difference applies to the relation with POMDP planning. Classical methods obtain guarantees by assuming an environment model, whereas agent UQ attempts to estimate quantities that such a model would otherwise provide. The formal guarantees do not transfer automatically, but the relevant quantities and their decision-theoretic motivation remain applicable.

\section{Taxonomy Conventions and Count Accounting}
\label{app:taxnotes}

This appendix records the labeling conventions behind the taxonomy of Section~\ref{sec:taxonomy} and the exact scope of the counts displayed in Figure~\ref{fig:tree}. Section~\ref{sec:taxonomy} summarizes both; the full statements are given here so that the corpus counts can be audited against \texttt{data/taxonomy\_table.csv}.

\paragraph{Differences in levels of abstraction within panels.}
The labels within a panel do not always describe concepts at the same level of abstraction.
\textcolor{axaqua!65!black}{Panel A} places two classical uncertainty \emph{types}, aleatoric \textcolor{axaqua!65!black}{A1} and epistemic \textcolor{axaqua!65!black}{A2}, beside three labels that describe where uncertainty enters or how it moves. Tool and environment uncertainty \textcolor{axaqua!65!black}{A3} describes an exogenous origin, accumulated uncertainty \textcolor{axaqua!65!black}{A4} describes propagation, and inter-agent uncertainty \textcolor{axaqua!65!black}{A5} describes a transport channel.
These three labels can overlap with the classical pair. Tool uncertainty is epistemic when the agent does not know a tool's reliability and aleatoric when the tool is intrinsically noisy. Accumulated uncertainty describes how either type changes over a horizon rather than defining another uncertainty type.
We retain all five labels in one panel because the corpus distinguishes them and because they lead to different responses in Table~\ref{tab:sources}. The classical distinction alone does not separate actions such as re-querying an API and rolling back several steps. Table~\ref{tab:sources} marks the conceptual difference by placing \textcolor{axaqua!65!black}{A1}--\textcolor{axaqua!65!black}{A2} and \textcolor{axaqua!65!black}{A3}--\textcolor{axaqua!65!black}{A5} in separate blocks. Section~\ref{sec:sources} discusses cases in which this boundary is unclear.

\textcolor{axyellow!75!black}{Panel B} also combines labels at different levels. Categories \textcolor{axyellow!75!black}{B1}--\textcolor{axyellow!75!black}{B5}, together with propagation \textcolor{axyellow!75!black}{B6}, are estimation mechanisms. Abstention and deferral \textcolor{axyellow!75!black}{B7} is a decision and control mechanism, while uncertainty-aware training \textcolor{axyellow!75!black}{B8} is a learning procedure. Panel B should therefore be read as describing how uncertainty is \emph{estimated or used}.
An alternative taxonomy could place \textcolor{axyellow!75!black}{B7} and \textcolor{axyellow!75!black}{B8} on a separate response axis. This would yield four dimensions: uncertainty type, system component, estimator, and response. We do not adopt that structure for two reasons. A response axis would contain only two families, leaving most papers unlabeled and adding a largely empty dimension to each display. Estimation, control, and training are also difficult to separate in practice. A conformal abstention rule \textcolor{axyellow!75!black}{B4} is both an estimator and a decision rule, and uncertainty-aware training often changes the estimator being trained.
We instead record the difference within Panel B. The lower block of Table~\ref{tab:families} contains the agent-native families. Its last two entries describe what is done with an estimate rather than how the estimate is produced.

\textcolor{axblue!65!black}{Panel C} combines pipeline components, including planning, tool use, retrieval, and memory, with an interaction regime, multi-step reasoning, and an architectural setting, multi-agent systems. The roadmap of Section~\ref{sec:roadmap} therefore marks the interaction regime separately, with the badges introduced in Section~\ref{sec:taxonomy}. A method tagged \textcolor{axblue!65!black}{C5} or \textcolor{axblue!65!black}{C6} is classified by its setting rather than by a component in the same loop as retrieval.
The four-dimensional interpretation can still be recovered from the taxonomy. Uncertainty type corresponds to \textcolor{axaqua!65!black}{A1}--\textcolor{axaqua!65!black}{A2}. Origin and system component correspond to \textcolor{axaqua!65!black}{A3}--\textcolor{axaqua!65!black}{A5} together with \textcolor{axblue!65!black}{Panel C}. Estimators correspond to \textcolor{axyellow!75!black}{B1}--\textcolor{axyellow!75!black}{B6}, and response or control corresponds to \textcolor{axyellow!75!black}{B7}--\textcolor{axyellow!75!black}{B8}.

\paragraph{Scope of the leaf counts in Figure~\ref{fig:tree}.}
The leaves report corpus tag counts, but they do not cover every paper on every panel. They also do not partition the corpus.
On \textcolor{axaqua!65!black}{Panel A}, $32$ of the $120$ papers have no source tag. Most are benchmarks or analyses whose methods do not assume a particular uncertainty source. Thus, about one quarter of the corpus does not specify its scope on this axis.
On \textcolor{axyellow!75!black}{Panel B}, the most common corpus tag is not shown as a leaf. A total of $59$ papers report calibration, and $4$ more apply post-hoc recalibration. Calibration is intentionally treated as a property rather than a family, as explained in Section~\ref{sec:taxonomy}. The panel therefore omits a label carried by about half of the corpus.
On \textcolor{axblue!65!black}{Panel C}, the $25$ single-turn papers and the $7$ agent-general papers have no leaf because \textcolor{axblue!65!black}{C1}--\textcolor{axblue!65!black}{C6} denote agentic stages.
Multi-tagging increases the totals. The $120$ papers receive $206$ stage tags, $210$ family tags, and $136$ source tags. A leaf count is therefore a tag count rather than a paper count, and the leaves sum to more than the corpus size.
Memory provides a clear example. It has $31$ papers, a count exceeded by only one stage, but only $4$ are tagged as memory alone and $20$ are also tagged as retrieval. Memory is often a secondary tag for retrieval-augmented work. Section~\ref{sec:rag} therefore describes methods that specifically address memory uncertainty as scarce despite the larger stage count.

\section{Single-Turn Estimator Families in Detail}
\label{app:singleturn}

Section~\ref{sec:foundations} summarizes the single-turn estimator families of Table~\ref{tab:families} together with the model access they require and the limitations that matter for agents. This appendix reviews the individual methods within each family. Dedicated surveys treat these estimator internals in still greater depth \citep{geng2024survey,huang2024survey,shorinwa2024survey,xia2025survey,liu2025uncertainty}.

\subsection{Verbalized Confidence, Sampling, and Consistency}
\label{app:singleturn-blackbox}

\citet{lin2022teaching} fine-tune GPT-3 to append a stated probability $v(x)$ to each answer, using empirical accuracy over groups of similar questions as the supervision target. The training objective minimizes the squared error between stated confidence and correctness. With squared loss, this objective is the Brier score, whose pointwise minimizer is the true conditional accuracy $\Pr(\hat{y}\ \text{correct} \mid x)$. Under this interpretation, training for ``calibrated verbal confidence'' applies a proper scoring rule to a confidence value expressed in language (Section~\ref{sec:metrics}).
At inference time, the elicitation format affects calibration. For an RLHF-tuned model, eliciting a numeric score or a ranked list of guesses with probabilities can yield better calibration than using raw conditional token probabilities, whose calibration may degrade during alignment tuning \citep{tian2023just,openai2023gpt4}. Prompting alone does not consistently remove overconfidence, although combining sampling with verbalized confidence can help. Calibration also remains sensitive to prompt wording \citep{xiong2023can,si2023prompting}.
Recent work further tests whether verbalized reports recover a model's internal answer distribution, whether reasoning models express confidence more faithfully, and where model confidence diverges from human judgments \citep{yang2024verbalized,yona2024large,kirchhof2025selfreflect,yoon2025reasoning,steyvers2024large}.
Long-form benchmarks and human studies examine how linguistic uncertainty is produced, calibrated, communicated, and interpreted \citep{yang2025uncle,chaudhry2024finetuning,belem2024perceptions,huang2024calibrating,vanderbles2019communicating}.
Training-based methods can further improve confidence estimates through decision-theoretic objectives, listener-aware preference tuning, self-reflective rationales, post-hoc calibration, and calibration-aware self-improvement \citep{band2024linguistic,stengel2024lacie,xu2024sayself,ulmer2024calibrating,liu2024litcab,huang2025accuracy}. Multicalibration provides group-conditional guarantees \citep{detommaso2024multicalibration}.

For sampling-based estimators, the information-theoretic starting point is the predictive entropy $H(Y \mid x)=-\sum_y p_\theta(y\mid x)\log p_\theta(y \mid x)$. Exact computation over the sequence space is intractable, so Monte Carlo estimators use sampled sequence log-likelihoods. Scores are usually normalized by length to avoid assigning higher uncertainty solely to longer outputs \citep{jiang2021how,malinin2021uncertainty,fomicheva2020unsupervised}.
The concern that distinct strings can express one meaning predates LLMs. It appears in work on machine-translation uncertainty and hallucination \citep{ott2018analyzing,glushkova2021uncertainty,guerreiro2023looking}, structured prediction \citep{malinin2021uncertainty}, and predictive uncertainty for hallucination detection in conditional generation \citep{xiao2021hallucination}. \citet{baan2023uncertainty} also examine the interpretation of distributions over strings.
Semantic entropy \citep{kuhn2023semantic} addresses this issue by partitioning $k$ samples into semantic-equivalence classes $\mathcal{C}$ using bidirectional NLI entailment, aggregating probability mass within each class, and computing entropy over the classes:
\begin{equation}
\mathrm{SE}(x) \;=\; - \sum_{C \in \mathcal{C}} p(C \mid x)\, \log p(C \mid x),
\qquad
p(C \mid x) \;=\; \sum_{s \in C} p(s \mid x)
\end{equation}
When token likelihoods are unavailable, a discrete variant uses empirical cluster frequencies within the sample \citep{farquhar2024detecting}.
This procedure typically uses $k \approx 10$ full generations and up to $O(k^2)$ NLI calls per input. The resulting clusters inherit errors from the entailment model, especially when semantic equivalence is ambiguous \citep{zhang2026whysemantic,tomov2025illusion}. \citet{kossen2024semantic} reduce these costs by training probes that predict semantic entropy from a single forward pass.
Several refinements reweight semantically important tokens \citep{duan2024shifting}, replace hard clusters with similarity kernels that represent graded entailment \citep{nikitin2024kernel}, or measure the density of an output in semantic space \citep{qiu2024semantic}. Black-box variants use embeddings or entailment graphs without model likelihoods \citep{lin2023generating,gao2024spuq}. Internal states provide another signal \citep{chen2024inside,snyder2024early}.
Related sampling methods combine resampling with self-reflection \citep{chen2023quantifying}, detect self-contradiction \citep{mundler2024self}, or verify uncertain spans during generation \citep{varshney2023stitch}. Other methods assess consistency without semantic clustering \citep{manakul2023selfcheckgpt} or score long-form outputs claim by claim \citep{min2023factscore,zhang2024luq}. Standardized implementations provide a common interface for many of these methods \citep{fadeeva2023lmpolygraph,vashurin2025benchmarking}.
Confidence-guided reasoning methods turn these scores into a selection rule over reasoning traces. Probabilistic confidence ranks candidate chains \citep{leang2026picsar} and response-wise estimates are improved by conditioning on the chain of thought \citep{zhang2025cot}; token-level estimates score reasoning steps \citep{zhang2025tokur}; self-certainty makes best-of-$N$ selection scalable \citep{kang2025scalable} and confidence gates how much test-time computation a problem receives \citep{fu2025deep}; the same signal guides code generation \citep{zhu2025uncertainty}.
Four further lines bear on the estimators above rather than adding one. Two studies re-audit the empirical picture, revisiting calibration and uncertainty estimation across current LLMs \citep{tao2025revisiting} and decomposing uncertainty by source to select a model and metric adaptively \citep{guo2025uncertainty}. Two apply uncertainty to detecting false content, through token-level fact checking \citep{fadeeva2024fact} and focus-weighted hallucination detection \citep{zhang2023enhancing}. Three concern the estimators' formal footing: a minimum-Bayes-risk view connecting confidence to consistency \citep{vashurin2025uncertainty}, an analysis of information-theoretic uncertainty measures \citep{schweighofer2024information}, and a derivation of predictive uncertainty from Bayesian risk estimation \citep{kotelevskii2024risk}. Three document failure modes and scope: epistemic uncertainty collapses in implicit ensembles of large models \citep{kirsch2024implicit}, uncertainty transfers unevenly from source to summary \citep{kolagar2024aligning}, and \citet{ulmer2024uncertainty} surveys the treatment of uncertainty in NLP as a whole. These developments build on broader accounts of uncertainty sources in supervised machine learning \citep{gruber2023sources}.

\subsection{Token Probabilities, Internal States, and Ensembles}
\label{app:singleturn-whitebox}

P(True)-style self-evaluation has been found to be better calibrated at larger model scales, especially when alternative sampled answers are included in the prompt for comparison \citep{kadavath2022language}.
\citet{yadkori2024believe} construct a pseudo-joint distribution over answers through \emph{iterated prompting}, in which earlier answers are added to the context. They show that the mutual information of this distribution lower-bounds epistemic uncertainty. This mutual information is zero when the responses are conditionally independent samples from a fixed ground-truth belief, which supports an abstention rule that is less sensitive to inherently ambiguous questions.
Among probe-based methods, supervised probes predict truthfulness from intermediate activations \citep{azaria2023internal}. Contrast-consistent search finds a truth direction without supervision by requiring a statement and its negation to receive complementary predictions \citep{burns2023discovering}. Semantic-entropy probes learn an inexpensive linear map from single-pass hidden states to a quantity that would otherwise require $k$ samples \citep{kossen2024semantic}.
Hidden states can also reveal error type \citep{orgad2025llms} and encode unanswerability before a refusal is produced \citep{slobodkin2023curious}. They support hallucination detection before or during generation \citep{snyder2024early,su2024unsupervised,chen2024inside} and can separate epistemic from aleatoric uncertainty without labels \citep{ahdritz2024distinguishing}. These findings suggest that internal states can encode reliability information that is not expressed in model outputs.
Bayesian approximations use variational weight posteriors, Monte Carlo dropout, and deep ensembles \citep{blundell2015weight,gal2016dropout,lakshminarayanan2017simple}, with deep ensembles showing strong calibration under distribution shift \citep{ovadia2019can}. Posterior predictive uncertainty can be decomposed into aleatoric and epistemic components \citep{kendall2017what,depeweg2018decomposition,hullermeier2021aleatoric}, and this decomposition can be used to study how uncertainty propagates through an agent pipeline. Because maintaining a posterior over billions of parameters is impractical, LLM studies use ensembles over low-rank adapters \citep{balabanov2024uncertainty,yang2024bayesian}, a variational posterior over adapter weights \citep{lin2026bayesian}, input clarifications \citep{hou2024decomposing}, in-context configurations \citep{ling2024uncertainty}, or fine-tuned uncertainty heads \citep{kapoor2024large}.

\subsection{Conformal Prediction, Recalibration, and Self-Knowledge}
\label{app:singleturn-guarantees}

In split conformal prediction, a nonconformity function $s(x,y)$, such as $1-p_\theta(y\mid x)$, is evaluated on a held-out calibration set $\{(x_i,y_i)\}_{i=1}^{n}$. Let $\hat{q}$ be the $\lceil(n+1)(1-\alpha)\rceil/n$ empirical quantile of the resulting scores. The prediction set
\begin{equation}
\mathcal{C}(x_{n+1})
=
\{y:s(x_{n+1},y)\leq \hat{q}\}
\end{equation}
satisfies $\Pr\big(y_{n+1}\in\mathcal{C}(x_{n+1})\big)\geq 1-\alpha$ when the calibration and test examples are exchangeable. If the nonconformity scores are continuous, the coverage is also bounded above by $1-\alpha+\frac{1}{n+1}$ \citep{vovk2005algorithmic,shafer2008tutorial,angelopoulos2023gentle}.
Language generation is more difficult because its output space is extremely large. Multiple-choice QA restricts the output to a finite answer set \citep{kumar2023conformal}. Conformal language modeling uses calibration to determine when to stop sampling and which samples to retain, so that the resulting set contains an acceptable generation with high probability \citep{bates2021distribution,quach2024conformal}. Conformal factuality removes low-confidence claims from a long response until the remaining claims satisfy a calibrated factuality criterion \citep{mohri2024language}. Other work studies conditional validity, certified acceptance rules, and the broader conformal prediction landscape \citep{cherian2024large,gui2024conformal,campos2024conformal}.
Because raw LLM confidence is often unreliable, selective-prediction methods combine it with a calibrator fitted for domain shift or with embedding-based out-of-distribution detection. Repeated sampling can also provide a more reliable abstention signal than likelihood on ambiguous questions \citep{chow1970optimum,elyaniv2010foundations,kamath2020selective,ren2023outofdistribution,cole2023selectively,rajpurkar2018know}.
Standard post-hoc calibration methods fitted on held-out data include Platt scaling \citep{platt1999probabilistic}, isotonic regression \citep{zadrozny2002transforming}, temperature scaling \citep{guo2017calibration}, and Dirichlet calibration for the full simplex \citep{kull2019beyond}. Training-time interventions pursue the same goal without a separate fitting stage \citep{thulasidasan2019mixup,mukhoti2020calibrating}. The need for recalibration depends on the model and dataset. \citet{guo2017calibration} found that modern networks were substantially more miscalibrated than their smaller predecessors, later work reported different patterns for newer architectures \citep{minderer2021revisiting}, and calibration also varies across pre-trained transformers and evaluation settings \citep{desai2020calibration}.
``Knowing what you don't know'' benchmarks categorize questions as known or unknown, or as answerable or unanswerable, and evaluate whether the model responds appropriately. \citet{yin2023large} measure self-aware refusal. \citet{kadavath2022language} train a P(IK) head to predict whether the model will answer correctly before generation. \citet{lin2022truthfulqa} study imitative falsehoods, for which the training data can favor incorrect answers.
Training methods can improve self-knowledge by learning refusal from the model's error patterns \citep{zhang2024rtuning,cheng2024ai}, optimizing directly for honesty \citep{yang2024alignment}, using collaboration to expose knowledge gaps \citep{feng2024dont}, or applying self-verification \citep{dhuliawala2024chain,agrawal2024language}. Comparative studies examine how these signals vary across models, prompts, and tasks \citep{huang2023look,si2023prompting,desai2020calibration,zhao2021calibrate,wen2024know}.

\section{Metric and Bootstrap Details for the Empirical Illustration}
\label{app:expdetails}

This appendix records the finite-sample behavior of the stratified calibration metrics defined in Section~\ref{sec:eval-metrics} and the resampling and refitting procedure behind every interval reported in Section~\ref{sec:empirical}.

\subsection{Finite-Sample Behavior of the Stratified Metrics}
\label{app:expdetails-metrics}

$\mathrm{ECE}_{\max}$ dominates the pooled score. Within a confidence bin, the pooled gap is a size-weighted average of the per-stratum gaps, so the triangle inequality bounds the pooled score by the size-weighted mean of the $\mathrm{ECE}_g$, which is at most their maximum. The same argument gives $\mathrm{ECE}_{\max} \geq \overline{\mathrm{ECE}}$.
The unweighted $\overline{\mathrm{ECE}}$ carries no such guarantee. It averages strata without weighting them, so a large badly calibrated stratum beside a small well calibrated one can pull it below the pooled score. We report it because it is more stable than $\mathrm{ECE}_{\max}$ at the sample sizes typical of trajectory data, where late strata are thin, and it should be read together with the group sizes rather than as a lower bound on pooled error. Neither statistic lets opposite-signed errors in different strata cancel, since both aggregate non-negative per-stratum values, although the usual binning limitations remain within each stratum.
The minimum-size filter excludes the strata that Definition~\ref{def:traj} is most likely to detect as problematic. The definition requires calibration at every checkpoint index, including late indices. Dropping small groups restricts evaluation to the surviving prefix and provides no information about the tail. In Section~\ref{sec:empirical}, the filter removes $t=5$ and $t=6$.
$\mathrm{ECE}_{\max}$ is also the maximum of positively biased estimates from groups of unequal size. Binned ECE has an upward finite-sample bias that increases as group size decreases \citep{kumar2019verified,roelofs2022mitigating}, and taking the maximum across groups increases this effect. A direct comparison between $\mathrm{ECE}_{\max}$ and the pooled score therefore reflects both the estimator's bias profile and the predictor's behavior, which is why the protocol of Section~\ref{sec:eval-metrics} requires the calibrated null.
When the data are sufficient, a stronger multicalibration-style criterion is appropriate. It requires $\Pr\big(Y=1 \mid \hat{R}(h_t)=r,\, g\big)=r$ jointly over confidence level and trajectory position rather than marginally over each \citep{hebertjohnson2018multicalibration,detommaso2024multicalibration}. Section~\ref{sec:eval-metrics} reports Eq.~\eqref{eq:strat-ece} instead because joint conditioning divides an already small sample along both dimensions.

\subsection{Bootstrap and Refitting Details}
\label{app:expdetails-bootstrap}

\paragraph{Nested bootstrap.}
Three estimators require fitting, making refitting within each bootstrap replicate important. Resampling the \emph{already cross-fitted} predictions represents sampling variation but omits variation from refitting. It therefore understates uncertainty for the fitted estimators. All intervals in Section~\ref{sec:empirical} use a \emph{nested} bootstrap. Each replicate resamples the $300$ trajectories and refits the isotonic map and both logistic baselines using the same two-fold cross-fitting procedure. Folds are assigned by distinct original trajectory, so duplicated draws never appear on both sides of the split.

\paragraph{Paired differences versus point estimates.}
The paired differences and the estimates in Table~\ref{tab:tcece} are computed differently. Each reported paired difference is the mean across nested bootstrap replicates, with both estimators refitted on the resampled trajectories. Table~\ref{tab:tcece} reports point estimates from models fitted once on the complete sample. These quantities need not agree. Each bootstrap replicate contains roughly $63\%$ of the distinct trajectories, which reduces the training data available to fitted estimators and changes comparisons involving them. Because the causal baseline is fitted, every comparison against it is affected. For example, subtracting the point estimates in Table~\ref{tab:tcece} gives $\Delta\mathrm{AUROC}=0.029$ for the verbalized signal, whereas the paired bootstrap mean is $+0.045$. The reported interval corresponds to the paired bootstrap mean.

\paragraph{Scope of the intervals.}
These intervals represent sampling and fitting variation for the selected ten-bin estimator while preserving within-trajectory dependence. They do not adjust for binning or finite-sample bias and should not be interpreted as confidence intervals for a binning-invariant population calibration error. The nested procedure may also be mildly conservative because each replicate contains roughly $63\%$ of the distinct trajectories and therefore fits each map on a smaller effective sample than the complete dataset.

\paragraph{Choice of resampling unit.}
At the short average horizon in this experiment, trajectory-level and checkpoint-level resampling produce similar intervals, $\pm0.032$ and $\pm0.027$. Trajectory-level resampling remains appropriate because checkpoints within a trajectory share the outcome $Y^{(j)}$, and this dependence increases with the horizon.

\subsection{Reference Tables for Section~\ref{sec:eval}}
\label{app:expdetails-tables}

Table~\ref{tab:benchmarks} inventories the evaluation resources that
Section~\ref{sec:eval-benchmarks} discusses, and Table~\ref{tab:drift}
reports the per-step signed gaps behind the drift analysis of
Section~\ref{sec:empirical}.

\begin{table}[t]
\centering
\caption{Evaluation resources grouped into capability benchmarks, single-turn
uncertainty benchmarks, and agent-specific uncertainty benchmarks. Only the
agent-specific group focuses on uncertainty across agent trajectories.}
\label{tab:benchmarks}
\small
\setlength{\tabcolsep}{4.5pt}
\renewcommand{\arraystretch}{1.12}
\begin{tabular}{@{}L{5.1cm}L{3.9cm}L{5.6cm}@{}}
\toprule
\thc{Resource} & \thc{Setting} & \thc{What it measures}\\
\midrule
\tband{3}{Capability benchmarks: task success only}
AgentBench \citep{liu2024agentbench} & Eight interactive environments & End-to-end task success\\
WebArena \citep{zhou2024webarena} & Realistic websites & Task success on the web\\
Mind2Web \citep{deng2023mind2web} & Generalist web tasks & Cross-site generalization\\
OSWorld \citep{xie2024osworld} & Real operating systems & Open-ended computer tasks\\
SWE-bench \citep{jimenez2024swebench} & GitHub issues & Repository-level code fixes\\
GAIA \citep{mialon2023gaia} & General assistance & Tool use, browsing, multimodality\\
$\tau$-bench \citep{yao2024tau} & Tool-agent-user dialogs & Success and run-to-run consistency\\
ToolEmu \citep{ruan2024toolemu} & Emulated sandbox & Risky agent behaviors\\
Autonomous tasks \citep{kinniment2024evaluating} & Open-ended real tasks & Dangerous autonomous capability\\
\tband{3}{Single-turn uncertainty and honesty}
TruthfulQA \citep{lin2022truthfulqa} & Single-turn QA & Imitative falsehoods\\
SQuAD 2.0 \citep{rajpurkar2018know} & Reading comprehension & Unanswerable-question abstention\\
AmbigQA \citep{min2020ambigqa} & Open-domain QA & Ambiguity handling\\
PopQA \citep{mallen2023not} & Entity-centric QA & When parametric memory fails\\
RGB \citep{chen2024benchmarking} & RAG stress tests & Noise robustness, conflict, rejection\\
LM-Polygraph \citep{vashurin2025benchmarking} & UQ method suite & Standardized single-turn UQ comparison\\
\tband{3}{Early evaluations of agent uncertainty}
URAG \citep{nguyen2026urag} & Retrieval-augmented QA & How retrieval shifts uncertainty and reliability\\
MIRAGE-Bench \citep{zhang2025miragebench} & Agentic RAG & Where agent hallucinations arise\\
CAR-bench \citep{kirmayr2026carbench} & Tool-use agents & Consistency and limit awareness\\
Yes-Man Syndrome \citep{yeke2026yesman} & Embodied robots & Abstention under physical uncertainty\\
CFO benchmark \citep{han2026can} & Long-horizon allocation & Delayed-feedback reliability\\
Markov-chain reliability \citep{trantruong2026measuring} & Agent trajectories & Trajectory-level reliability estimation\\
Stepwise self-evaluation \citep{mavi2025selfevaluating} & Multi-step tasks & Per-step confidence quality\\
Computer-use agent UQ \citep{kumar2026uncertainty} & VLM agents, GUI grounding & Uncertainty across computer-use stacks\\
\bottomrule
\end{tabular}
\end{table}

\begin{table}[t]
\centering
\caption{Signed gap by step index on the HotpotQA traces: $g(t)$ is mean
stated reliability minus observed success among trajectories reaching $t$;
positive values are overconfident. Last column: after pooled isotonic
recalibration.}
\label{tab:drift}
\small
\setlength{\tabcolsep}{5pt}
\renewcommand{\arraystretch}{1.15}
\begin{tabular}{@{}rrccrr@{}}
\toprule
\thc{$t$} & \thc{$n_t$} & \thc{Stated} & \thc{Observed} & \thc{$g(t)$} & \thc{$g(t)$, recal.}\\
\midrule
1 & 300 & 0.490 & 0.657 & $-0.166$ & $-0.072$\\
2 & 299 & 0.634 & 0.659 & $-0.025$ & $-0.027$\\
3 & 256 & 0.842 & 0.680 & $+0.163$ & $+0.014$\\
4 & 55 & 0.624 & 0.345 & $+0.278$ & $+0.263$\\
5 & 29 & 0.559 & 0.276 & $+0.283$ & $+0.316$\\
6 & 15 & 0.453 & 0.133 & $+0.320$ & $+0.414$\\
\bottomrule
\end{tabular}
\end{table}

\subsection{The Chained-QA Companion Experiment in Detail}
\label{app:expdetails-chains}

Section~\ref{sec:empirical} summarizes the companion experiment; this
appendix reports its design and per-condition results in full.

We construct chains of $T=4$ independent QA items from GSM8K
\citep{cobbe2021training} and TriviaQA \citep{joshi2017triviaqa}, using $150$
chains per condition, answered either within one shared context or in
separate contexts. The design provides an exact correctness label for every
step, which real agent traces generally do not contain and which is necessary
for estimating both the pairwise step-error correlations and the
prefix-coupling coefficients $\rho_k$ of Proposition~\ref{prop:positive}. It
also removes several mechanisms that create dependence in agent trajectories.
There is no action--observation loop, no tool output that a later step must
consume, and no conditioning on the preceding step's \emph{answer}; the items
are independent by construction, leaving the shared context window as the
only channel of dependence. The measured correlations are therefore a lower
bound on dependence in interactive agents, and values near zero are expected.

\begin{table}[t]
\centering
\caption{Chained-QA composition experiment: $T=4$ independent QA items per
chain, $150$ chains per condition, Qwen2.5-32B-Instruct. Trajectory ECE
scores the prefix product and the learned aggregator against chain success;
$\hat\rho$ is the mean pairwise correlation of step errors, with a $95\%$
trajectory-bootstrap interval.}
\label{tab:chains}
\small
\setlength{\tabcolsep}{4pt}
\renewcommand{\arraystretch}{1.15}
\begin{tabular}{@{}lcccccc@{}}
\toprule
& \multicolumn{2}{c}{Step level} & & \multicolumn{2}{c}{Trajectory ECE} & \\
\cmidrule(lr){2-3}\cmidrule(lr){5-6}
\thc{Condition} & \thc{Acc.} & \thc{ECE} & \thc{Success $R$} & \thc{Product} & \thc{Learned} & \thc{$\hat\rho$ [95\% CI]}\\
\midrule
GSM8K, shared context & 0.782 & 0.175 & 0.420 & 0.420 & 0.145 & $0.089$ $[+0.023,+0.159]$\\
GSM8K, separate contexts & 0.792 & 0.184 & 0.380 & 0.528 & 0.039 & $-0.000$ $[-0.072,+0.086]$\\
TriviaQA, shared context & 0.740 & 0.159 & 0.293 & 0.361 & 0.132 & $0.001$ $[-0.069,+0.086]$\\
TriviaQA, separate contexts & 0.730 & 0.203 & 0.287 & 0.466 & 0.175 & $0.019$ $[-0.056,+0.095]$\\
\bottomrule
\end{tabular}
\end{table}

The step estimator is poorly calibrated, with stated confidence between
$0.90$ and $0.98$ against the step accuracies in Table~\ref{tab:chains}, and
multiplication roughly doubles to triples the error at the trajectory level,
by a factor of $2.3$ to $2.9$ across the four conditions: the product
predicts an average reliability of $0.65$--$0.91$ against realized chain
success of $0.29$--$0.42$. The learned aggregator is a logistic regression on
five summary statistics of the step-confidence vector (mean, minimum,
maximum, product, and length); it models no dependence between steps, so what
it repairs is step-level miscalibration rather than the composition rule.
For the shared-context GSM8K chains, the only condition with a measured
positive correlation, the step marginals give $\prod_t p_t = 0.372$ against a
realized $R = 0.420$ and a mean reported product of $0.840$, so the
step-calibration term of the decomposition in
Appendix~\ref{app:proofs-pathwise} is $+0.468$ while the dependence term is
$+0.048$. When the measured correlation is near zero, the product of marginal
accuracies tracks realized reliability ($0.392$ predicted, $0.380$ observed
for independent GSM8K chains); when it is positive, realized reliability
exceeds the marginal product by $+0.048$, within the proposition's bound of
$0.095$.

The positive correlation has several limitations. Its interval excludes zero
only narrowly, the lower bound moves between $0.016$ and $0.023$ across
bootstrap seeds, one of four tested conditions is significant with no
multiplicity correction, and the design confounds task identity with the
presence of transferable reasoning content. It supports only the hypothesis
that dependence reflects information transferred across reasoning steps
rather than context sharing alone; testing that hypothesis requires a design
in which a later step must \emph{consume} an earlier output, such as a
multi-hop question whose later hop uses the entity a preceding hop produced,
with more than $150$ chains. The aggregator's intervals do not include
refitting; it is fitted once and evaluated on held-out chains.

\section{The Remaining Open Problems}
\label{app:challenges}

Table~\ref{tab:agenda} presents an agenda of ten problems. Section~\ref{sec:challenges} develops Problems 1--3 because the results of this paper bear on them directly. The remaining seven problems are discussed below using the numbering from that table.

\begin{enumerate}[start=4,leftmargin=1.6em,itemsep=4pt]

\item \textbf{Conformal guarantees under non-exchangeability.}
Conformal prediction generally assumes exchangeability, whereas agent trajectories are sequential and history dependent. Each action changes the future state, and tool responses connect later decisions to earlier ones. Applying split conformal prediction independently at each step therefore provides no trajectory-level guarantee.
Conformal methods have been adapted to language-model generation \citep{kumar2023conformal,quach2024conformal,mohri2024language}. General methods allow bounded departures from exchangeability \citep{tibshirani2019conformal,gibbs2021adaptive,barber2023conformal}, and recent work studies guarantees and attribution in structured agent settings \citep{kotte2026when,feng2026conformal,padhi2026actions}. Several adjacent research directions may provide useful foundations but have not yet been adopted in this literature. Conformal risk control replaces set coverage with a bound on the expected value of any monotone loss \citep{angelopoulos2024conformal}. This objective is closer to a trajectory-level guarantee than set membership. Time-uniform confidence sequences provide validity under continuous monitoring \citep{howard2021timeuniform}, which matches the operating conditions of a step-wise certificate. Time-series conformal methods relax exchangeability because observations arrive in sequence and their distribution may change over time. They recover validity by adapting the level online or aggregating bootstrapped residuals rather than assuming an i.i.d.\ calibration set \citep{xu2021conformal,zaffran2022adaptive,gibbs2024conformal}. Agent trajectories are more difficult because the agent's own actions generate the distribution shift, but the statistical structure is related and this transfer has not yet been studied. Progress requires explicit assumptions about dependence between steps and guarantees for coverage or risk under those assumptions.

\item \textbf{Disentangling and routing uncertainty sources.}
The five sources in Table~\ref{tab:sources} require different responses. Most methods nevertheless compress them into one score, and methods that separate uncertainty by source remain uncommon \citep{donaldson2026bayesian,matsnev2026uncertainty}. A useful estimator should identify the dominant source and associate it with an appropriate action. Attribution would then support routing decisions rather than serve only as a measurement.

\item \textbf{Correlated failure in multi-agent systems.}
Agreement is informative only when agent errors are sufficiently independent. Agents that share a base model, context, or prior messages may fail in correlated ways, causing a shared error to appear as confident consensus \citep{huang2026counterfactual,yoffe2024debunc,smit2024should}. Aggregation rules should therefore estimate inter-agent dependence explicitly. Adding agents with similar failure patterns may reinforce a shared error rather than correct it.

\item \textbf{Efficient uncertainty estimation.}
Methods that require many samples become computationally expensive when repeated at every step. This limits the use of semantic entropy \citep{kuhn2023semantic,farquhar2024detecting} and trajectory-level propagation \citep{zhao2024saup,duan2025uprop} for real-time control. The objective is a reliability estimate that adds little overhead beyond the agent's existing computation. Candidate approaches include hidden-state probes \citep{kossen2024semantic,azaria2023internal}, single-pass propagation, and generation-free scores \citep{zhu2026towards}.

\item \textbf{Interfaces that humans can act on.}
Verbal and numerical confidence affect trust in different ways \citep{zhou2023navigating,zhou2024relying,mielke2022reducing}. Depending on its presentation, the same estimate may produce over-reliance, unnecessary escalation, or alarm fatigue. The interface is therefore part of the decision system. Uncertainty communication should be evaluated through operator outcomes, including successful interventions, missed failures, and false alarms, rather than only through calibration of the underlying score.

\item \textbf{Adversarially robust uncertainty estimation.}
Verbalized confidence is generated from the same context that may contain malicious instructions. An indirect prompt injection \citep{greshake2023not} can therefore influence both the selected action and its reported confidence. Sampling does not necessarily prevent this failure because injected content may affect every sample and create artificial agreement. Internal-state estimators avoid dependence on generated confidence reports \citep{azaria2023internal,kossen2024semantic}. Benchmarks should measure how much the performance of each estimator family degrades under attack, treating robustness as part of uncertainty estimation rather than as a later addition.

\item \textbf{Uncertainty for multimodal and computer-use agents.}
Perceptual uncertainty arises before reasoning. An agent may misread a screen or ground an object incorrectly, and once this error enters $h_t$, a well-calibrated text-side estimator may reason confidently from an incorrect input. Section~\ref{sec:surfaces} reviews the current evidence, which remains limited to early benchmarks and studies of vision-language calibration, navigation monitoring, and multimodal abstention \citep{kumar2026uncertainty,xiao2026vlcalibration,he2025mmboundary,darabi2026groundcontrol,madhusudhan2026knowing,yeke2026yesman}. The open problem is to build methods that distinguish grounding uncertainty from decision uncertainty and benchmarks that provide labels for both.

\end{enumerate}

\section{Use of Large Language Models}
\label{app:llm-use}

Large language models were used to polish the writing of this paper: grammar,
word choice, and sentence-level clarity on text the authors had already
written. They were not used to generate claims, to select or summarize papers,
to build the taxonomy, or to run the experiments, and the authors checked every
citation against the cited work. The one place an LLM enters the method rather
than the prose is the adjudication step of Appendix~\ref{app:corpus}, where a
judge model proposed a label on annotator disagreements and a human set every
final label. The authors take full responsibility for the content of this
paper.

\end{document}